\documentclass[11pt]{article}

\usepackage[colorlinks]{hyperref}

\usepackage[dvipsnames]{xcolor}

\usepackage{amsmath}
\usepackage{amsthm}
\usepackage{amsfonts}
\usepackage{amssymb}
\usepackage{amsbsy}
\usepackage{mathtools}
\usepackage{bm}
\usepackage{bbm}
\usepackage{physics}
\usepackage{mathrsfs}
\usepackage{dsfont}

\usepackage{floatrow}

\usepackage{interval}
\intervalconfig{soft open fences}

\usepackage[dvips,letterpaper,margin=1in,bottom=1in]{geometry}
\usepackage[capitalize,noabbrev]{cleveref}

\usepackage[utf8]{inputenc}
\usepackage[english]{babel}
\usepackage[autostyle]{csquotes}

\usepackage[T1]{fontenc}
\AtBeginDocument{%
  \DeclareFontShape{T1}{cmr}{m}{scit}{<->ssub*cmr/m/sc}{}%
}

\usepackage{diagbox}
\usepackage{mathtools}

\newtheorem{theorem}{Theorem}[section]

\newtheorem{lemma}[theorem]{Lemma}
\newtheorem{corollary}[theorem]{Corollary}
\newtheorem{proposition}[theorem]{Proposition}

\newtheorem{fact}[theorem]{Fact}

\newtheorem{definition}[theorem]{Definition}

\newtheorem{remark}{Remark}[section]
\newtheorem{problem}{Problem}

\DeclarePairedDelimiter\rbra{\lparen}{\rparen}

\DeclarePairedDelimiter\cbra{\{}{\}}
\DeclarePairedDelimiter\Abs{\lVert}{\rVert}

\DeclarePairedDelimiter\floor{\lfloor}{\rfloor}

\DeclareMathOperator*{\E}{\mathbb{E}}

\newcommand{\ii}{\textup{i}}

\newcommand{\mc}[1]{\mathcal{#1}}
\newcommand{\pip}[2]{\langle #1,#2\rangle_\mathrm{s}}

\newcommand{\Prob}[1]{\mathrm{Pr}\left({#1}\right)}

\newcommand{\polylog} {\operatorname{polylog}}
\newcommand{\supp} {\operatorname{supp}}

\newcommand{\BQP} {\mathsf{BQP}}

\newcommand{\QMA} {\mathsf{QMA}}

\definecolor{myblue}{rgb}{0.0, 0.0, 0.5}

\usepackage{latexsym}
\usepackage{CJK}

\usepackage{enumerate}

\usepackage[ruled,linesnumbered]{algorithm2e}
\usepackage{algpseudocode}

\usepackage{framed}
\usepackage{comment}
\usepackage{faktor}
\usepackage{changepage}

\usepackage{stmaryrd}
\usepackage{tabularx}
\usepackage{booktabs}
\usepackage{adjustbox}

\usepackage{booktabs}    
\usepackage{multirow}    
\usepackage{colortbl}    
\usepackage{bigdelim}

\newcommand{\footremember}[2]{%
    \footnote{#2}
    \newcounter{#1}
    \setcounter{#1}{\value{footnote}}%
}

\usepackage{tikz}
\usetikzlibrary{quantikz2}

\begin{document}

\title{Quantum Hamiltonian Certification}
\author{
Minbo Gao\footremember{1}{Minbo Gao is with the Key Laboratory of System Software (Chinese Academy of Sciences) and
State Key Laboratory of Computer Science, Institute of Software, Chinese
Academy of Sciences, China, 
and with University of Chinese Academy of Sciences, China.}, 
Zhengfeng Ji\footremember{2}{Zhengfeng Ji is with the Department of Computer Science and Technology, Tsinghua University, Beijing, China.}, 
Qisheng Wang\footremember{3}{Qisheng Wang is with the School of Informatics, University of Edinburgh, United Kingdom.}, 
Wenjun Yu\footremember{4}{Wenjun Yu is with QICI Quantum Information and Computation Initiative, School of Computing and Data Science,
The University of Hong Kong, Pokfulam Road, Hong Kong, China.},
Qi Zhao\footremember{5}{Qi Zhao is with QICI Quantum Information and Computation Initiative, School of Computing and Data Science,
The University of Hong Kong, Pokfulam Road, Hong Kong, China.}}
\date{}
\maketitle

\pagenumbering{gobble}

\begin{abstract}
  We formalize and study the Hamiltonian certification problem, a fundamental task in quantum physics, crucial for verifying the accuracy of
  quantum simulations and quantum-enhanced technologies.
  Given access to $e^{-\ii Ht}$ for an unknown Hamiltonian $H$, the goal of the
  problem is to determine whether $H$ is $\varepsilon_1$-close to or
  $\varepsilon_2$-far from a target Hamiltonian $H_0$.
  While Hamiltonian learning methods have been extensively studied, they often
  require restrictive assumptions and suffer from inefficiencies when adapted
  for certification tasks.

  This work introduces a direct and efficient framework for Hamiltonian
  certification, which distinguishes whether an unknown Hamiltonian matches a
  target specification within given precision bounds.
  Our approach achieves \textit{optimal} total evolution
  time $\Theta((\varepsilon_2-\varepsilon_1)^{-1})$ for
  certification under the normalized Frobenius norm, without prior structural
  assumptions.
  This approach also extends to certify Hamiltonians with respect to all Pauli norms
  and normalized Schatten $p$-norms for $1\leq p\leq2$ in the one-sided error
  setting ($\varepsilon_1=0$), where the optimality is consistently maintained.
  Notably, the result in Pauli $1$-norm suggests a quadratic advantage of our
  approach over all possible Hamiltonian learning approaches.
  We also establish matching lower bounds to show the optimality of our approach
  across all the above settings.
  We complement our result by showing that the certification problem with
  respect to normalized Schatten $\infty$-norm is $\mathsf{coQMA}$-hard, and therefore
  unlikely to have efficient solutions.
  This hardness result provides strong evidence that our focus on above metrics is not merely a technical choice but a requirement for efficient
  certification.

  To enhance practical applicability, we develop an ancilla-free certification
  method that maintains the inverse precision scaling while eliminating the need
  for auxiliary qubits, making our approach immediately accessible for near-term
  quantum devices with limited resources.
\end{abstract}

\newpage
\tableofcontents
\clearpage
\pagenumbering{roman}
\newpage

\clearpage
\pagenumbering{arabic}
\setcounter{page}{1}
\section{Introduction}

The study of time evolution, described by self-adjoint Hamiltonians through the
Schr\"{o}dinger equation, lies at the heart of quantum physics.
This concept gave rise to quantum simulation, the implementation of time evolution first envisioned by Feynman as the foundational purpose of quantum computing~\cite{feynman1982simulating}.
By enabling precise manipulation of quantum dynamics, quantum simulation has become indispensable for exploring complex phenomena in quantum many-body systems~\cite{jordanQuantumAlgorithmsQuantum2012,schreiber2015observation} and quantum chemistry~\cite{mcardle2020quantum}.
Beyond the role in quantum simulation, time evolution also serves as a critical
technique for quantum-enhanced technologies, from ultra-precise
measurements~\cite{giovannetti2006quantum,degen2017quantum} to revolutionary
algorithms~\cite{LC19,gilyen2019quantum}.
This dual role---as both a fundamental application and a key operational
primitive---cements time evolution as the cornerstone of quantum
science.

Given its fundamental importance, various approaches have been developed to
estimate and learn the time evolution of a quantum system.
Several well-established methods exist for estimating quantum operations in
general settings, including quantum process
tomography~\cite{Chuang1997,altepeter2003ancilla,leung2003choi} and unitary
estimation techniques~\cite{acin2001optimal,yang2020optimal,haah2023query}.
But they become infeasible when dealing with continuous operations like time
evolution rather than discrete gate models.
Hamiltonian learning methods provide a more sophisticated alternative that has
shown considerable
promise~\cite{Shabani2011,da-Silva2011,granade2012robust,huang2023learning,bakshi2024structure,ma2024learning}.
In this learning-from-dynamics setting, we access the time evolution operator
$e^{-\ii Ht}$ of an unknown Hamiltonian $H$ for a controllable duration $t$.
By preparing specific probe states, evolving them under the dynamics of $H$, and
performing strategic measurements, Hamiltonian learning can efficiently
reconstruct a classical description of $H$ with guaranteed
accuracy~\cite{quench,zubida2021optimal}.

A fundamental challenge in quantum experiments is verifying whether the
engineered system faithfully implements the desired Hamiltonian $H_0$.
While this shares similarities with Hamiltonian learning from real-time
dynamics, there's a crucial distinction: we possess a target Hamiltonian and aim
to certify its accurate experimental realization.
This leads us to the core question considered in this paper:
\[
\textit{How can we certify a given Hamiltonian evolution?}
\]
Concretely, given a quantum device that evolves states for specified durations
$t$, we must determine whether its dynamics match the target Hamiltonian $H_0$
or exhibit significant deviations.
We formulate this as the \emph{Hamiltonian certification} problem for real-time
evolution.
In this certification problem, the total evolution time is the key resource
metric (similar to Hamiltonian learning conventions), directly reflecting the
physical resources expended during certification.

A natural approach to Hamiltonian certification involves first learning the
system's Hamiltonian and then comparing the learned parameters with the target
specification $H_{0}$.
However, this learning-based certification method faces several fundamental
limitations.
Most Hamiltonian learning algorithms require prior knowledge about the
interaction structure within the system.
Despite recent advancements in methodologies, few proposals can efficiently
extract structural information without substantial
assumptions~\cite{yu2023robust,zhao2024learning,hu2025ansatz}.
This dependency inherently restricts the applicability for certification tasks,
particularly when the actual evolution might substantially deviate from the
target---precisely the scenario we aim to certify.
More significantly, this method poses a fundamental efficiency limitation.
For certification, the objective shifts from reconstructing the unknown
Hamiltonian to simply verifying its proximity to the target, rendering full
Hamiltonian characterization potentially a waste of resources.

In this work, we present efficient Hamiltonian certification techniques with
broader applicability and substantially lower resource requirements, as
discussed in Section~\ref{sec:Main}.
Our approach outperforms existing theoretical bounds established
in~\cite{ma2024learning} and requires substantially shorter evolution times.
The technique overcomes fundamental limitations of current Hamiltonian learning
methods, providing both rigorous theoretical guarantees and practical advantages
for the certification problem.

\subsection{Main results}\label{sec:Main}
Motivated by the above discussion, this paper studies the Hamiltonian
certification problem formally defined below:

\begin{problem}[$n$-qubit Hamiltonian certification]\label{prob:HC}
  Let $H$ (unknown) and $H_0$ (known) be two $n$-qubit traceless Hamiltonians
  with their Pauli coefficients bounded by a constant.
  Given access to the unknown $H$ through its time evolution, the problem of
  \emph{Hamiltonian certification to precision $(\varepsilon_1, \varepsilon_2)$}
  is to distinguish the two cases:
    \begin{itemize}
        \item \textsc{Accept} instance: $\norm{H -H_0}_*\leq\varepsilon_1$,
        \item \textsc{Reject} instance: $\norm{H-H_0}_* \ge \varepsilon_2$,
    \end{itemize}
    promised that it is in either case, where $\norm{\cdot}_*$ is a given norm of operators.
\end{problem}

Note that in \cref{prob:HC}, except for the ubiquitous tracelessness condition
(e.g.,~\cite{granade2012robust,quench,huang2023learning,bakshi2024structure}),
we \textit{have no constraints} on the Hamiltonians.
This marks a significant advance compared to prior settings on Hamiltonian
learning, which either asks for explicit structures or requires restrictive
assumptions such as locality~\cite{zubida2021optimal}, low
intersection~\cite{huang2023learning}, or effective
sparsity~\cite{bakshi2024structure}.
These constraints, while making analysis tractable, often limit applicability to
idealized systems rather than arbitrary quantum dynamics.
Our problem definition removes these limitations entirely, providing a
general-purpose certification framework.
Even more remarkably, despite this generality, our approach achieves an optimal
total evolution time, depending only on the precision parameters:

\begin{theorem} [Main, informal version of Theorems~\ref{thm:tolerant} and~\ref{thm:lower_Pauli}] \label{thm:main}
    With $n+2$ ancillary qubits, a total evolution time of $\Theta\rbra{(\varepsilon_2-\varepsilon_1)^{-1}}$ is necessary and sufficient for $n$-qubit Hamiltonian certification to precision $(\varepsilon_1,\varepsilon_2)$ (\cref{prob:HC}) for any $\varepsilon_2>\varepsilon_1>0$ with respect to the \emph{normalized Frobenius norm}.
\end{theorem}

Here we use the normalized Frobenius norm (e.g., as in~\cite{Low09}), also known as the normalized Schatten $2$-norm. 
This norm provides a measure of the \enquote{average-case} distance between Hamiltonians and directly relates to the average-case fidelity between output states over Haar-random inputs (see~\cite{MdW16,zhao2022hamiltonian,BCO24}). 
The normalized version is relevant as it prevents dimension factors from obscuring meaningful differences, allowing us to focus on proportional deviations rather than absolute differences that would scale with system size.

In particular, our approach achieves the optimal inverse precision scaling for the evolution time.
This optimal scaling represents a significant theoretical and practical advancement over existing approaches from Hamiltonian learning.
In comparison, under the same setting with minimal Hamiltonian assumptions, Hamiltonian learning methods exhibit worse scaling when used for certification.
These methods scale less efficiently in the precision parameter~\cite{yu2023robust}, or require additional factors in their complexities~\cite{zhao2024learning,hu2025ansatz}. 
Several state-of-the-art Hamiltonian learning methods do achieve the inverse precision scaling, but only by imposing substantial structural constraints~\cite{huang2023learning,dutkiewicz2024advantage,bakshi2024structure}.\footnote{These methods hold \enquote{Heisenberg scalings} under the metric of $\ell_\infty$-norm of Pauli coefficients.}
Our approach uniquely combines unconstrained applicability with this inverse precision scaling, eliminating the traditional trade-off between generality and performance.

\subsubsection{Implications for other norms}
While our main results focus on Hamiltonian certification using the normalized Frobenius norm, it can be extended to other important matrix norms. 
Here, we demonstrate how our certification framework with one-sided errors naturally generalizes to both Pauli norms and several normalized Schatten norms, \emph{retaining the optimality in each case}.
This extension significantly broadens the applicability of our approach, as different norms may be more relevant in specific physical contexts or applications.

\paragraph{Pauli norms.}
For an $n$-qubit operator $A$ with its Pauli decomposition $A = \sum_{\alpha \in \mathsf{P}^n} s_{\alpha} P_{\alpha}$ over the $n$-qubit Pauli matrix set ${\sf P}^n$, the Pauli $p$-norm of $A$ is defined to be the $\ell_p$-norm of the vector $\vec{s} = \rbra{s_{\alpha}}_{\alpha \in \mathsf{P}^n}$, i.e., $\Abs{A}_{\textup{Pauli},p} \coloneqq \Abs{\vec{s}}_p = \rbra{\sum_{\alpha \in \mathsf{P}^n} \abs{s_\alpha}^{p}}^{1/p}$.
Note that the Pauli $2$-norm is actually the normalized Frobenius norm, i.e., $\Abs{A}_F = \Abs{A}_{\textup{Pauli},2}$.
Moreover, $\norm{A}_{\mathrm{Pauli},0}$ represents the number of non-zero entries in $\vec{s}$.

\begin{corollary} [Hamiltonian certification in Pauli norms, informal version of Theorems~\ref{thm:Ancilla_certification_Pauli} and~\ref{thm:lower_Pauli}]
    With $n+2$ ancillary qubits, a total evolution time of $\Theta\rbra{\varepsilon^{-1}}$ is necessary and sufficient for $n$-qubit Hamiltonian certification to precision $(0,\varepsilon)$ (\cref{prob:HC}) with respect to the \emph{Pauli $p$-norm} for $p \geq 2$.
    In addition, a total evolution time of $\Theta\rbra{m^{1/p-1/2}\varepsilon^{-1}}$ is necessary and sufficient with respect to the \emph{Pauli $p$-norm} for $1 \leq p < 2$, where $m$ is an upper bound on $\norm{H}_{\mathrm{Pauli},0}$ and $\norm{H_0}_{\mathrm{Pauli},0}$. 
\end{corollary}
\noindent Notably, for the certification with respect to the Pauli 1-norm, our approach achieves $O(m^{1/2}\varepsilon^{-1})$ the evolution time, while all possible Hamiltonian learning must cost $\Omega(m\varepsilon^{-1})$ as suggested in~\cite{ma2024learning}.

\paragraph{Normalized Schatten norms.}
For a $d$-dimensional operator $A$, the normalized Schatten $p$-norm of $A$ is defined to be $\Abs{A}_{\textup{Schatten}, p} \coloneqq \rbra{\tr\rbra{\abs{A}^p}/d}^{1/p}$. 
Note that the normalized Schatten $2$-norm is actually the normalized Frobenius norm, i.e., $\Abs{A}_F = \Abs{A}_{\textup{Schatten}, 2}$. 

\begin{corollary} [Hamiltonian certification in normalized Schatten norms, informal version of Theorems~\ref{thm:Ancilla_certification_Schatten} and~\ref{thm:lower_Schatten}]
    With $n+2$ ancillary qubits, a total evolution time of $\Theta\rbra{\varepsilon^{-1}}$ is necessary and sufficient for $n$-qubit Hamiltonian certification to precision $(0,\varepsilon)$ (\cref{prob:HC}) with respect to the \emph{normalized Schatten $p$-norm} for $1 \leq p \leq 2$.
\end{corollary}

The extension to various norms leverages fundamental mathematical relationships between these norm families.
For Pauli norms with $p\geq2$ and normalized Schatten norms with $p\leq2$, we establish upper bounds by showing that certification in these norms can be reduced to certification in the Frobenius norm.
This is because $\norm{A}_{\mathrm{Pauli},p}\leq\norm{A}_{\mathrm{Pauli},2}$ for $p\geq2$ and that $\norm{A}_{\mathrm{Schatten},p}\leq\norm{A}_{\mathrm{Schatten},2}$ for $p\leq2$.
For the more challenging case of Pauli norms with $1\leq p<2$, the situation differs because these norms can be substantially smaller than the Frobenius norm. The additional factor $m^{1/p-1/2}$ arises from the norm conversion between $\ell_2$- and $\ell_p$-norms for vectors with at most $m$ non-zero entries.
We prove matching lower bounds in \cref{thm:lower_Pauli} and \cref{thm:lower_Schatten}.

We find that certification with respect to the normalized Schatten $p$-norm when $p>2$ is challenging. 
In particular, when examining the normalized Schatten $p$-norm for $p=\infty$
(which corresponds to the operator norm), we demonstrated that certification is inherently hard.
This result shares similarities 
with previous findings
on the hardness of quantum circuit non-identity check problem~\cite{JWB05,JW09}.
\begin{theorem}[Informal version of \cref{thm:QMA-hard-HC-operator}]
    The problem of Hamiltonian certification with respect to the \emph{operator norm} is $\mathsf{coQMA}$-hard.
\end{theorem}

By establishing matching upper and lower bounds across multiple norm families, we demonstrate that our certification approach achieves optimality regardless of which distance metric is chosen.
The versatility of our approach provides a powerful and flexible toolkit for Hamiltonian certification tasks across diverse applications with different physical relevances.

\subsubsection{Ancilla-free approach}

While our main theorem provides strong theoretical foundations for Hamiltonian certification, practical quantum implementations must contend with significant real-world constraints: limitations on ancillary quantum resources and unavoidable noise in multi-qubit operations. 
These requirements are particularly challenging in near-term platforms where cross-talk, decoherence, and gate fidelities constrain the practical application of ancilla-based protocols~\cite{preskill2018quantum,sarovar2020detecting}.
To relieve the need of the controlled evolution queries and multi-qubit gates, we develop an ancilla-free variation of our certification protocol:

\begin{theorem} [Ancilla-free Hamiltonian certification, informal version of Theorem~\ref{thm:Stabilizer}]\label{thm:ancillar-free-informal}
    There exists an ancilla-free quantum algorithm using product states as inputs that solves $n$-qubit Hamiltonian certification to precision $(\varepsilon,4\varepsilon)$ (\cref{prob:HC}) with respect to the \emph{normalized Frobenius norm}.
    This algorithm uses a total evolution time of $\widetilde O\rbra{m^{3/2}\varepsilon^{-1}}$, where $m$ is an upper bound on $\norm{H}_{\mathrm{Pauli},0}$ and $\norm{H_0}_{\mathrm{Pauli},0}$.\footnote{We use $\widetilde O\rbra{f} = O\rbra{f \log\log\rbra{f}}$ to suppress a \textit{double} logarithmic factor in \cref{sec:Main,sec:Tech}.} 
\end{theorem}

The widened gap between $\varepsilon$ and $4\varepsilon$ is a consequence of the limited capabilities for getting rid of ancillary qubits.
It ensures reliable discrimination while maintaining the same inverse-precision scaling.
While the evolution time complexity also increases modestly to $\widetilde O\rbra{m^{3/2}\varepsilon^{-1}}$, this approach maintains the advantage for few-to-no Hamiltonian assumptions and gets rid of the mentioned practical challenges.
This extension
of ancilla-free implementation transform our theoretical framework into a practical toolkit for Hamiltonian certification across various experimental scenarios, from idealized high-fidelity quantum processors to more constrained near-term devices.\\

\begin{table}[t]
  \centering \adjustbox{max width=\textwidth}{
    \begin{tabular}{l c c  c r r }
      \toprule
      Reference & Hamiltonian class & Explicit structure?
      & $n_\mathrm{anc}$ & $T_\mathrm{total}$ & $N_\mathrm{measure}$\\\midrule
      \cite{zubida2021optimal} & Geometrically local & No$^\dagger$ & 0
                         & $m^{3/2}\varepsilon^{-3}$ & $m^2\varepsilon^{-4}$\\\midrule
      \cite{huang2023learning} & \multirow{2}{*}{Low intersection} & Yes
      & $0$ & $m^{1/2}\varepsilon^{-1}$ &  $\polylog(m^{1/2}\varepsilon^{-1})$\\
      \cite{dutkiewicz2024advantage} & & Yes & $0$ & $m^{1/2}\varepsilon^{-1}$
                                              &  $\log(m^{1/2}\varepsilon^{-1})$\\\midrule
      \cite{bakshi2024structure}
                & \begin{tabular}{@{}l@{}}Effectively sparse, \\
                  bounded strength \end{tabular}
      & No & $0$ & $m^{1/2}\varepsilon^{-1}$ &  $\log(m^{1/2}\varepsilon^{-1})$\\\midrule
      \cite{ma2024learning} & $k$-local & No & 0 & $m\varepsilon^{-1}$
                                              & $m\log(\varepsilon^{-1})$ \\\midrule
      \cite{yu2023robust} &  \multirow{5}{*}{Sparse, nonlocal} & No$^\ddagger$ & $0$
                         & $n\norm{H}^3\varepsilon^{-4}$ &  $n\norm{H}^4\varepsilon^{-4}$\\
      \cite{castaneda2023hamiltonian} &  & Yes  & $n+m$ & $m\varepsilon^{-1}$
                                              & $\norm{H} m\varepsilon^{-1}$\\
      \cite{odake2024higher} &  & Yes & $1$ & $m^{3/2}\varepsilon^{-1}$
                                              &  $m \log(\varepsilon^{-1})$\\
      \cite{zhao2024learning} &  & No & $n+\log m$ & $m^{3/2}\varepsilon^{-1}$
                                              &  $m^2\log(\varepsilon^{-1})$\\
      \cite{hu2025ansatz} &  & No & $n$ & $m^{5/2}\varepsilon^{-1}$
                                              &  $m^2\log(\varepsilon^{-1})$\\\midrule\midrule
      \cref{thm:ancillar-free-informal} & Sparse, nonlocal & No  & $0$
                         & $m^{3/2}\varepsilon^{-1}$ &  $m^3 \log(\varepsilon^{-1})$\\\midrule
      Theorem~\ref{thm:main} & Arbitrary$^\S$ & No & $n + 2$ & $\varepsilon^{-1}$
                                              & $n+\log(\varepsilon^{-1})$\\
      \bottomrule
    \end{tabular}}
  \caption{Comparison of algorithms for certifying an $n$-qubit $m$-sparse
    Hamiltonian $H$ with constant-bounded coefficients to precision $(\varepsilon_1,\varepsilon_2)$ with constant success
    probability.
    Denote $\varepsilon\coloneqq\varepsilon_2-\varepsilon_1$.
    This comparison evaluates several critical efficiency metrics, including
    ancillary qubit number $n_\mathrm{anc}$, total evolution time $T$, and
    measurement complexity $N_\mathrm{measure}$.
    For consistent comparison, we standardized results from the
    $\ell_\infty$-norm to the
    $\ell_2$-norm based on the direct relation between $\ell_p$-norms, though
    further refinements to these conversions may be possible.
    Each Hamiltonian class in the table is subsumed by the one below it.
    Subdominant constants and polylogarithmic factors are suppressed.
    The condition of traceless Hamiltonian is consistently imposed.\\
    $^\dagger$The structure can be learned under the geometrically local setting as derived in~\cite{bakshi2024structure}.\\
    $^\ddagger$Their method learns the structure given that $m$ non-zero Pauli components are distributed uniformly.\\
    $^\S$For an arbitrary Hamiltonian $H$, we only use the most general relation
    that $m\leq 4^n-1$.
    }\label{tab:results}
\end{table}

\noindent To explicitly underscore our contributions, we present Table~\ref{tab:results}, which systematically compares certification-via-learning (using prominent Hamiltonian learning methods) against our direct certification protocols. 
As the table demonstrates, our approach achieves significant advantages by directly addressing the certification question.

\subsection{Technical overview}\label{sec:Tech}

In this section, we present key techniques underlying our Hamiltonian certification approach.
We organize our technical contributions into five components.
First, we introduce a \emph{Pauli coefficient analysis} technique that establishes a rigorous mathematical connection between the Pauli coefficients of a Hamiltonian and those of its time evolution operator.
Second, relying on this connection, we describe the \emph{Hamiltonian amplitude encoding} technique to encode and Pauli-coefficient information of the evolution operator of $H_\mathrm{res}$ to measurable quantum amplitudes.
Third, we also present the \emph{Hamiltonian stabilizer sampling} that removes potential experimental bottlenecks of the amplitude encoding technique.
Fourth, we discuss a binary-search-inspired protocol that optimally determines the Pauli coefficients using only logarithmic checks.
Fifth, we provide a sketch of lower bounds to demonstrate the optimality of our certification approaches in many settings.

\paragraph{Pauli coefficient analysis of the time evolution operators.}

The foundation of our certification algorithm is a novel analysis of the relationship between an $n$-qubit Hamiltonian's Pauli coefficients and those of its corresponding time evolution operator.
Specifically, for a traceless Hamiltonian
$H=\sum_{\alpha\in{\sf P}^n} s_{\alpha}P_{\alpha}$ with $m$ non-zero Pauli
terms, we establish precise upper and lower bounds on the $\ell_2$-norm of the
Pauli coefficient vector $\vec{v}$ of the evolution operator $e^{-\ii Ht}$, as
detailed in \cref{sec:Pauli-coefficient-analysis}.
Most importantly, these upper and lower bounds are \emph{dimension-independent}
(only with a dependence on $m$, the number of non-zero Pauli terms), enabling our certification protocol to achieve optimal scaling.

We start our analysis with the Taylor expansion
of the exponential functions.
For sufficiently small
$t$, the Taylor expansion
gives
$e^{-\ii Ht}\approx
I -\ii Ht + O(t^2)$.
This suggests that the Pauli coefficient 
$v_{\alpha}$ of
$e^{-\ii Ht}$
is approximately
$s_{\alpha}t$ to first order.
In~\cite{hu2025ansatz}, 
the authors formalized 
this idea for individual
Pauli coefficient of $e^{-\ii Ht}$.
Specifically,
they showed a coordinate-wise bound
\[\abs{v_{\alpha}}\ge 
\abs{s_{\alpha}}t -R,\]
where $R := \frac{1}{m}\sum_{k=2}^{\infty} \frac{(m\Abs{\vec{s}}_{\infty}t)^k}{k!} = O(m\Abs{\vec{s}}_{\infty}^2t^2)$
is a remainder term 
derived from the Taylor expansion.
This bound follows from
a direct counting argument 
for a specific Pauli term appearing in $H^k$. 
By the same reasoning,
one can show that
$\abs{v_{\alpha}}\le 
\abs{s_{\alpha}}t +R$.

However, since there are $4^n$ entries in the vector $\vec{v}$,
directly applying the above argument only leads to
a \emph{dimensional-dependent}
bound of 
the $\ell_2$-norm, given by
\[
\Abs{\vec{v}}_2^2\le \sum_{\alpha\in \supp(H)} (\abs{s_{\alpha}}t+R)^2 + (4^n-m)R^2.
\]
The $4^n$ factor significantly limits
the applicability of the upper bound,
as it requires $t$ to be exponentially small.
To address this issue, we refine the upper
and lower bounds to be \emph{dimensional-independent},
leveraging two key technical insights.
\begin{itemize}
    \item A refined counting argument 
    for the Pauli terms appearing in $H^k$
    (\cref{lemma:evolution-Pauli-efficient-bound}).
    \item Convex constrained optimization analysis of Taylor polynomials (\cref{prop:opt-high-order-terms}).
\end{itemize}

Intuitively, our refined 
counting argument 
strengthens the 
coordinate-wise bound
$\abs{\tr(H^kP_{\alpha})}/2^n\le s^k m^{k-1}$
to
an $\ell_1$-norm bound
$\sum_{\alpha\in{\sf P}^n} \abs{\tr(H^kP_{\alpha})}/2^n\le s^k m^{k}$,
indicating an overhead merely of $m$ instead
of $4^n$.
Now, for the $k_0$-th order
Taylor polynomial of $e^{-\ii Ht}$
as 
\[
 I - \ii Ht + \sum_{k=2}^{k_0}\frac{(-\ii)^kt^kH^k}{k!}, 
\]
consider the $k_0$-th order squared $\ell_2$-norm
\[
\sum_{\alpha\in{\sf P}^n} \rbra*{-\ii s_{\alpha}t 
+ \sum_{k=2}^{k_0} \frac{(-\ii)^k t^k \tr(H^k P_{\alpha})}{2^nk!}}^2.
\]
We aim to give an  upper bound of the above
quantity 
by first taking the term-wise absolute 
values, and then interpreting it as a convex 
constrained optimization problem
where $ \abs{\tr(H^k P_{\alpha})} $
are the variables.
A careful analysis of this kind of problems (\cref{prop:opt-high-order-terms})
suggests that the maximum is
attained when all the high-order
terms concentrate on the same set of 
Pauli terms,
i.e., $\tr(H^kP_{\alpha})= 2^n s^k m^{k-1}$
for all integers $k \in [2, k_0]$ on the same $m$ Pauli terms $P_{\alpha}$'s.
Therefore, 
the remainder term $R$, which bounds $\sum_{k=2}^{k_0} \frac{t^k \tr(H^k P_{\alpha})}{2^nk!}$,
could
contribute at most
$O(m)$ times to the squared $\ell_2$-norm.
Finally, by taking the limit $k_0\to \infty$, the above method gives
a dimensional-independent upper bound 
\[
\Abs{\vec{v}}^2 \le \sum_{\alpha\in\supp(H)} (\abs{s_{\alpha}}t+R)^2
\]
as desired.
The dimensional-independent lower bounds
can also be established using the same approach.

\paragraph{Hamiltonian amplitude encoding.}
Our Pauli analysis has revealed that the Pauli coefficients of short-time evolution serve as natural bounds for the underlying Hamiltonian's Pauli coefficients. 
In this part, we will realize the insight of utilizing the short-time evolution as a proxy to extract the information of the underlying Hamiltonian.
To this end, we introduce the \emph{Hamiltonian amplitude encoding} method, which encodes the Pauli information of the evolution operator as the measurable amplitude of a state.

The key mathematical insight enabling our approach comes from a remarkable property of the $2n$-qubit Bell state $\ket{\Phi^+}\coloneqq\frac{1}{\sqrt{2^n}}\sum_i\ket{i}\ket{i}$.
When a unitary $U$ with Pauli decomposition $U=\sum_{\alpha}c_\alpha P_\alpha$ acts on the first register of this state, it creates a superposition:
\begin{align} 
(U \otimes I)\ket{\Phi^+} = \sum_{\alpha} c_\alpha (P_\alpha \otimes I)\ket{\Phi^+} = \sum_\alpha c_\alpha \ket{\Phi_\alpha},\notag
\end{align}
where $\ket{\Phi_\alpha}\coloneqq(P_\alpha\otimes I)\ket{\Phi^+}$.
Crucially, these states $\{\ket{\Phi_\alpha}\}$ form an orthonormal basis, implying that each Pauli component of $U$ disperses into a distinct orthogonal subspace.
This dispersion transforms the Pauli coefficients into amplitude information, making them directly measurable.

From the illustration of the residual Hamiltonian and the problem reduction, we would replace $U$ by the evolution of $H_\mathrm{res}$.
To this end, we apply the second-order Trotter-Suzuki simulation as
\begin{align}
    e^{-\ii H_\mathrm{res}\delta t}\approx U_2(\delta t)\coloneqq e^{\ii H_0\delta t/2} e^{-\ii H\delta t}e^{\ii H_0\delta t/2}.\notag
\end{align}
For an arbitrary long time $t$, the simulation can be achieved by repeating $U_2(\delta)$ for $r\coloneqq t/\delta t$ times.
According to \cref{prop:Trotter_Error}, the simulation error scales as $O(t^3/r^2)$, allowing us to suppress errors to any desired precision by increasing $r$.
After applying the time evolution, we perform a controlled operation that selectively flips an ancilla based on which of $\{\ket{\Phi_\alpha}\}$ is present.
The resulting state has the form $U_{\mathsf{HAE}}\rbra{t} \ket{0}\ket{0} = p_0\rbra{t} \ket{0}\ket{\psi_0} + p_1\rbra{t} \ket{1}\ket{\psi_1}$, where $t$ denotes the evolution time consumed and the amplitude $p_1\rbra{t}$ can be bounded upper and lower according to Pauli coefficients of $H_{\mathrm{res}}$ for small $t$ (see \cref{prop:noisy_bound}). 
This encoding transforms our certification problem into distinguishing between small and large values of $p_1(t)$. 
Therefore,  we can employ the square root amplitude estimation~\cite{Wan24} to estimate $p_1\rbra{t}$ to within an additive error.\footnote{In contrast to the original quantum amplitude estimation~\cite{BHMT02}, the square root amplitude estimation method we recruited here is able to estimate the \emph{square root} of the measurement probability within a desired precision $\varepsilon$ using ${O}(1/\varepsilon)$ queries. 
} 

\paragraph{Dual-stabilizer Hamiltonian certification.}
For most current or near-term quantum devices, implementing and maintaining high-fidelity ancillary qubits and reliable multi-qubit controlled operations poses a significant experimental challenge.
Motivated by these practical concerns, a natural question arises: Can we certify Hamiltonians using methods that operate entirely within the system being verified, without requiring additional quantum resources? 
Here we introduce the \emph{dual-stabilizer} method to eliminate the ancilla requirements from the need for Bell states and amplitude encoding.
By further replacing the square root amplitude estimation with the ordinary sampling, we can achieve the ancilla-free certification.

To replace the Bell state, we utilize an $n$-qubit stabilizer state $\ket{\psi_0}$ associated with a maximal stabilizer group $\mc S_0$ (with $2^n$ elements).
According to \cref{prop:Anticom_Isomorphic}, applying Pauli operators from different cosets in the anticommutant $\mc A_{\mc{S}_0}$ (the quotient group of $\mathbb{P}^n$ over $\langle\ii\rangle\times\mc S_0$) to $\ket{\psi_0}$, we generate states with distinct syndromes, dispersing Pauli components effectively.
According to \cref{prop:Unique_Decomp}, an arbitrary unitary $U$ with Pauli decomposition $\sum_{\alpha}c_\alpha P_\alpha$ thus transforms $\ket{\psi_0}$ as:
\begin{align}\label{eq:Stabilizer_Dispersion}
    U\ket{\psi_0}=\sum_{\alpha}c_\alpha P_\alpha\ket{\psi_0}=\sum_{\beta\in\mc A_{\mc S_0}}\left(\sum_{\gamma\in\mc S_0}v_{\beta,\gamma}c_{\beta+\gamma}\right)P_\beta\ket{\psi_0},
\end{align}
where $\alpha=\beta+\gamma$ and $v$'s are phase factors of unit magnitude.
This results in a superposition of orthogonal components with distinct syndromes.
After performing a corresponding projective measurement (termed as \emph{syndrome measurement}), the probability reflects the collective contribution of Pauli coefficients within the corresponding coset $\beta$.

To prevent the potential destructive interference due to the phase factors in Eq.~\eqref{eq:Stabilizer_Dispersion}, we employ a randomization technique by selecting states uniformly from the set $\{P_\theta\ket{\psi_0}\}_{\theta\in\mc A_{\mc S_0}}$.
For each measurement outcome $\beta+\theta$, we only record the relative syndrome shift $\beta$.
By summing the probabilities of observing non-identity syndromes in $\mc A_{\mc S_0}$, we obtain the signal probability $\Pr(Z=1)$, which equals the $\ell_2$-norm of all Pauli coefficients outside $\mc S_0$ as formalized in \cref{lm:stabilizer_sample}.

However, this approach has a limitation: it cannot detect Pauli coefficients within the stabilizer group $\mc S_0$.
This is when \enquote{dual}
becomes crucial.
We introduce a second maximal stabilizer group $\mc S_2$ to complement $\mc S_1=\mc S_0$.
Importantly, these two stabilizer groups share only the identity operator.
By performing syndrome measurements with respect to both stabilizer groups, we can comprehensively probe all non-identity Pauli coefficients of the unitary.

\paragraph{Binary check.}
Based on previous discussion, we can encode and estimate the Pauli information through the syndrome measurements.
To make the gap between the measurement outcomes of different Hamiltonian cases obvious for accurate discrimination, we need to increase the evolution time $t$.
On the other hand, we need $t$ to be small to guarantee the validity of the relation from \cref{prop:noisy_bound} connecting the amplitude and the Pauli information of $H_\mathrm{res}$.
To resolve this tension, we introduce a technique we call \emph{binary check} to properly organize the steps of measurements or certifications.
This allows us to achieve the optimal evolution time while maintaining the validity of our amplitude bounds.

Our binary check method employs a bisection strategy to handle different scales of Pauli coefficients. 
We decompose the possible range of the residual Hamiltonian's coefficient magnitudes into logarithmically many sections.
For instance, with a generic bound $B$ on Pauli coefficients, we consider sections: $[B/2,B]$, $[B/4,B/2], \cdots, [\varepsilon,2\varepsilon]$ across different checks.
With the explicit bound in each check, we can choose a properly large evolution time for each check while ensuring that \cref{prop:noisy_bound} remains valid.

Formally, our approach solves logarithmically many subtasks (see \cref{prop:iter_v_stab,prop:fina_v_stab}), each to determine whether the signal probability $\Pr(Z=1)$ is small or large with $t = 2^jm^{-3/2}\varepsilon^{-1}$ for $0 \leq j \leq O\rbra{\log\rbra{\varepsilon^{-1}}}$, where $m = \Abs{H_{\mathrm{res}}}_{\mathrm{Pauli},0}$ represents the number of non-zero Pauli terms. 
To distinguish between small and large $\Pr(Z=1)$, we use \textsf{BernoulliTest} (see \cref{prop:Bernoulli-distribution-parameter-testing}) to discriminate between either cases. 
Therefore, by summing over all subtasks, the evolution time of the method achieves the inverse precision scaling that characterizes our approach's optimality.

\paragraph{Lower bounds.}
To complement the optimality of our certification algorithms, we establish matching lower bounds.
We employ a similar reduction technique used in~\cite{huang2023learning}, showing that any Hamiltonian certification method can be used to solve a corresponding hypothesis testing problem between distinct Hamiltonians.
For rigorously analyzing all possible adaptive measurement protocols, we adopt the tree representation introduced in~\cite{chen2022exponential}.
By analyzing the total variation (TV) distance between measurement outcomes from different trees, we derive fundamental limits on the required evolution time.

Previous lower bounds of 
Hamiltonian learning mainly
focused on the  Pauli $\infty$-norm case.
We extend and strengthen a new lower bound of Pauli $p$-norm for $1\leq p<2$, where we establish tighter bounds by constructing a relatively challenging hypothesis testing task.
The key insight is the existence of subsets of the Pauli group with cardinality $O(n)$ (where $n$ is the number of qubits) where elements mutually anti-commute.
Consider an algorithm that performs certification with respect to the Pauli $p$-norm.
It must be able to distinguish between $H_1=\sum_{\alpha\in\mc S}\frac{\varepsilon_1}{m^{1/p}}P_\alpha$ and $H_2=\sum_{\alpha\in\mc S}\frac{\varepsilon_2}{m^{1/p}}P_\alpha$, where $\mc S$ is one of these anti-commuting sets with cardinality $m$.
By calculating the operator-norm distance, we can bound $T=\Omega(m^{1/p-1/2}(\varepsilon_2-\varepsilon_1)^{-1})$, which matches our upper bound and confirms its optimality.

\subsection{Related work}
This section provides an overview of prior research in Hamiltonian characterizations, including learning, testing, and other related topics.
Given the extensive literature, our discussion here might not be thorough.

\paragraph{Hamiltonian learning.}
Hamiltonian learning emerged from the broader fields of quantum sensing and quantum metrology, where the primary goal is to estimate parameters governing quantum dynamics through unknown time evolution. 
Initially focused on estimating specific Hamiltonian parameters, these techniques progressively developed into comprehensive frameworks for characterizing entire unknown Hamiltonians from their dynamical effects.
Below, we review the progression of Hamiltonian learning methods, highlighting their conceptual approaches and limitations.

Early Hamiltonian learning proposals~\cite{da-Silva2011,granade2012robust,wiebe2014hamiltonian,wiebe2014quantum} established what is now termed the \emph{derivative estimation} approach. 
This method involves applying the unknown Hamiltonian dynamics to carefully selected initial states and extracting coefficient information through output measurements. 
These protocols determine Hamiltonian parameters through various mathematical techniques, including solving equations and Bayesian inference. 
This straightforward approach inspired numerous subsequent works that learn Hamiltonians through direct state preparation and measurements~\cite{quench,zubida2021optimal,stilck2024efficient,gu2024practical}.

Despite their conceptual clarity and experimental accessibility, these derivative estimation methods typically require explicit prior knowledge of the unknown Hamiltonian's structure, presenting a limitation in general learning scenarios.
A notable contribution by Shabani et al.~\cite{Shabani2011} circumvented this constraint by
formulating the learning problem as linear equations with a coefficient matrix approximately satisfying the restricted isometry property~\cite{candes2008introduction}.
Nevertheless, their method still needs to assume that the possible non-zero terms of the Hamiltonian fall into a small set, which limits the applicability.
Motivated by the same thought, Yu et al.~\cite{yu2023robust} showed for the first time a provably robust and efficient protocol to identify and extract the sparse and nonlocal Pauli coefficients of an unknown Hamiltonian.
Notably, they required that non-zero terms be uniformly distributed.
Subsequently, Caro's work~\cite{caro2024learning} can also learn constant local Hamiltonians without structure knowledge based on ancillary systems.

The field has recently experienced significant progress in efficiency, with several breakthroughs in reducing the total evolution time required for learning.
Haah et al.~\cite{haah2022optimal} proposed a learning algorithm with a total $O(\varepsilon^{-2}\log(n))$ evolution time for precision parameter $\varepsilon$.
A major milestone came when Huang et al.~\cite{huang2023learning} first improved the evolution time for learning low-intersection Hamiltonians to the inverse precision scaling $O(\varepsilon^{-1})$.
Bakshi et al.~\cite{bakshi2024structure} extended this scaling to effectively sparse Hamiltonians. 
These triggered the heat of exploring more learning methods with an inverse-precision evolution time~\cite{castaneda2023hamiltonian,dutkiewicz2024advantage,odake2024higher,ma2024learning}.
Notable recent advances include Zhao's work~\cite{zhao2024learning}, which maintained the inverse-precision evolution time while generalizing to all sparse Hamiltonians without explicit structural assumptions, and Hu et al.~\cite{hu2025ansatz}, who preserved this advancement and generality while enabling ancilla-free learning.

While our discussion has centered on qubit systems, researchers have also made progress on learning bosonic~\cite{hangleiter2024robustly,LTG+24} and fermionic~\cite{ni2024quantum,mirani2024learning} Hamiltonians. 
It is notable that our certification methods naturally extend to these domains through appropriate qubit mappings, maintaining their efficiency advantages.

Beyond learning from dynamics, an alternative approach exists that involves extracting Hamiltonian information from its static properties.
One important class of problems
falling into this
category is to learn the Hamiltonian
of a quantum system given copies of 
its thermal state~\cite{Bairey_2019,anshu2020sample,haah2022optimal,cambyse2024learning,cambyse2024efficient,gu2024practical,bakshi2024learning} or eigenstate~\cite{Qi2019determining}.

\paragraph{Hamiltonian property testing}
Property testing is a crucial area in both classical~\cite{G10} and quantum computing~\cite{MdW16}, focused on developing algorithms that efficiently approximate the detection of properties or parameters of large, complex objects. Recently, there has been growing interest in testing Hamiltonian properties using queries to its time evolution operators, particularly in testing locality~\cite{BCO24,gutierrez2024simple,arunachalam2024testing} and sparsity~\cite{arunachalam2024testing}.
In contrast to Hamiltonian certification, which aims to fully test a Hamiltonian is the desired one, these property testing algorithms only determine whether a given Hamiltonian is $k$-local (or sparse) or is far from satisfying this property.

\paragraph{Other related works}
Deciding whether two circuits have the same functionality is a fundamental problem
in computer science.
In the realm of quantum computing,
circuit equivalence testing focuses on deciding whether two unitary operators are approximately equivalent. 
A commonly used measure in this task (see~\cite{linden2021light,TGS+24}) is
(the square root of) the entanglement infidelity, a quantity closely
related to the normalized Frobenius norm,
which is our main focus.
The distinction between circuit equivalence testing and Hamiltonian certification primarily arises in three key aspects: (i) Hamiltonian certification allows queries to $e^{-\ii Ht}$ with flexible choice of $t$, while circuit equivalence testing queries only to the two unitary operators themselves;
(ii) Hamiltonian certification considers various physically meaningful norms, including Pauli $p$-norms, whereas
circuit equivalence testing focuses on, e.g., the entanglement infidelity;
(iii) Our Hamiltonian certification relies on a quantitative relationship between the Pauli coefficients of $H$ and $e^{-\ii Ht}$ that has not been previously explored in the circuit equivalence testing literature.

Certifying quantum objects has long been a significant challenge. Prior to our work on certifying Hamiltonian evolutions, extensive research has been conducted on quantum
state certification~\cite{BOW19,CLO22,CLHL22,HPS24}
and quantum channel certification~\cite{FFGO23}.

\subsection{Discussion}
In this work, we present a rigorous formalization of Hamiltonian certification
through time evolution, clearly differentiating it from the related Hamiltonian
learning problem.
Our framework establishes efficient methods for certifying whether or not an unknown
Hamiltonian deviates from the specified target within given precision bounds, requiring
only the standard traceless condition instead of additional assumptions.

By focusing exclusively on certification-relevant information, our method
achieves the optimal inverse precision scaling for
certification under normalized Frobenius distance.
We further demonstrate the versatility of this approach by extending its one-sided version to
multiple norm families, including Pauli and normalized Schatten norms, while
preserving optimal performance across different certification metrics.

To bridge theory with experimental implementation, we developed a key
extension of our framework.
Specifically, we introduced an ancilla-free variant that operates without auxiliary
quantum resources, making the certification immediately applicable to near-term
quantum devices where advanced control capabilities may be limited.
This significantly enhances the approach's practical utility across
diverse experimental platforms.

We now discuss some interesting open problems.
\begin{enumerate}
  \item Can one explore improved upper and lower bounds for normalized Schatten
        $p$-norm with $2<p<\infty$?
        Directly adapting our approach for these tasks would introduce dimensional factors, which scale exponentially with the system size.
        Therefore, it remains unclear what the tight upper bounds or hardness is for certification under these
        metrics.
  \item What is the optimal evolution time for ancilla-free certification?
        While the lower bounds proved in our work naturally apply to the
        ancilla-free case, there remains a gap between the lower bound and the achieved
        upper bound.
        Determining whether the ancilla-free setting fundamentally requires this
        additional factor or if more efficient algorithms exist represents an
        important direction for future work.
  \item It remains unclear whether certification is intrinsically harder if we
        cannot explicitly implement the inverse evolution of the target
        Hamiltonian.
        In Hamiltonian learning, inverse evolution access is believed to
        simplify the problem, but whether this potential advantage transfers to
        the certification setting (and to what extent) represents an intriguing
        direction for further study.
\end{enumerate}

\section{Preliminaries}\label{sec:prelim}

\subsection{Notations}
We denote the imaginary unit by $\ii\coloneqq\sqrt{-1}$.
For any integer $n$, we will use $[n]$ to denote the set $\{1, 2, \dots, n\}$.
For an arbitrary unitary operator $U$, we use $\mathsf{ctrl}(U)$ to denote the controlled unitary operation, which applies $U$ given the single-qubit $\ket{1}$ state and does nothing otherwise.
We use $\log (\cdot)$ to denote the logarithm function with base $2$, and $\ln (\cdot)$ to denote the logarithm function with base $e$.

Let $A = \sum_{P_\alpha\in{\sf P}^n} s_\alpha P_\alpha$ be the \emph{Pauli decomposition} of an $n$-qubit operator, where ${\sf P}^n$ denotes the $n$-qubit Pauli matrix set.
We call the coefficients $\{s_\alpha\coloneqq\Tr(AP_\alpha)/2^n\}_{P_\alpha\in{\sf P}^n}$ the \emph{Pauli coefficients} of $A$.
Furthermore, we introduce the \emph{Pauli-coefficient vector} of $A$ as an $4^n$-dimensional vector $\vec{s}$ label by the $n$-qubit Pauli group
where the $\alpha$-coordinate is just $s_\alpha$.
We will usually use $\mc M$ to denote all labels of non-zero Pauli coefficients with cardinality $\abs{\mc M}=m$.
For any set $\mathcal{X}\subseteq {\sf P}^n$,
$\vec{s}[\mathcal{X}]$ denotes
the vector with coordinate $s_{\alpha}$
for $P_\alpha\in \mathcal{X}$, and $0$ otherwise.
Therefore, $\vec{s}=\vec{s}[\mc M]$.
Moreover, we will use $\|\vec{s}\|_p$ to denote the $\ell_p$-norm of $\vec{s}$ with $p\in[1,\infty]$.

\subsection{Pauli and stabilizer formalism}\label{sec:Pauli_formalism}

Given a set of $n$ qubits with Hilbert space dimension $2^n$, we recall the following definitions to introduce Pauli and stabilizer formalism.
\begin{definition}
    The $n$-qubit Pauli group $\mathbb{P}^n$ is composed of tensor products of $I$, $X$, $Y$, and $Z$ on $n$ qubits, with a global phase of $\pm 1$ or $\pm \ii$.
    The $n$-qubit Pauli matrix set
    is defined to be ${\sf P}^n\coloneqq\{I,X,Y,Z\}^{\otimes n}$.
\end{definition}
\noindent Here ${\sf P}^n$ has $4^{n}$ elements since there are $4^n$ $n$-fold tensor products of $I, X, Y, Z$ matrices, and $\mathbb{P}^n$ consists of $4^{n+1}$ elements due to the four global phases they could have.
With this convention, ${\sf P}^n$ constitutes
a set of orthogonal basis of $(\mathbb{C}^{2\times 2})^{\otimes n}$ under the Hilbert-Schmidt inner product.
This justifies the validity of the Pauli decomposition in the last subsection.

For the $n$-qubit Pauli group $\mathbb{P}^n$, there is a special relation induced by the group multiplication named as \emph{commutation}.
Specifically, this commutation relation generates a map $c:\mathbb{P}^n\times\mathbb{P}^n\rightarrow \mathbb{Z}_2$ such that for any two Pauli matrices $P$ and $P'$ in $\mathbb{P}^n$, we have $P P'=(-1)^{c(P,P')} P' P $.

Based on this commutation relation, we introduce an important concept:
\begin{definition}
    We call a subgroup $\mc S$ of $\mathbb{P}^n$ a stabilizer group if it is an Abelian subgroup and does not contain $\{-I^{\otimes n}\}$.
    The common $+1$ eigenstates of all Pauli elements in $\mc S$ are defined as the stabilizer states.
\end{definition}
\noindent Frequently, we will pick a minimal generating set $\{P_{1},\cdots,P_{r}\}$ to represent the stabilizer group, where no generator is a product of other generators. 
In this sense, all stabilizer elements are the generators' products.
Notably, the order does not matter since the generators commute and all elements of $\mc S$ must square to the identity.  
Therefore, any element in $\mc S$ can be uniquely determined by taking a product $\prod_j P_{j}^{i_j}$
with $i_j \in \{0,1\}$.
This implies $\abs{\mc S} = 2^r$.

In this work, we exclusively focus on the maximal stabilizer group $\mc S_0$, which has a cardinality of $2^n$.
It is evident that $\mc S_0$ only has one stabilizer state, which we denote by $\ket{\psi_0}$.
Since this $\mc S_0$ achieves the maximal size, all Pauli operators outside $\mc S_0$ must anti-commute with at least one generator, which introduces the concept of \emph{error syndromes}:
\begin{definition}
    Given a maximal stabilizer group $\mc S_0$ of $\mathbb{P}^n$ and its generators $\{P_1,\cdots,P_n\}$, the error syndrome of a Pauli operator $P'\in\mathbb{P}^n$ is an $n$-bit string $\sigma_{\mc S_0}(P')\in\mathbb{Z}_2^n$ such that
    \begin{gather}
    \forall i\in[n],\ \sigma_{\mc S_0}(P')_i\coloneqq c(P',P_i).\notag
    \end{gather}
\end{definition}
\noindent We will ignore the subscript $\mc S_0$ in $\sigma$ when it is clear from the context.
To view the operational interpretation of the error syndrome, we need the following measurement:
\begin{definition}\label{def:syndrome_measurement}
    Given a maximal stabilizer state $\mc S_0$ of $\mathbb{P}^n$ and its generators $\{P_1,\cdots,P_n\}$, the syndrome measurement is the projective measurement on all projectors in $\left\{\prod_{i=1}^n\left(\frac{I+(-1)^{b_i}P_i}{2}\right)\right\}_{b\in\mathbb{Z}_2^n}$.
\end{definition}
\noindent Consequently, suppose we conduct the syndrome measurement on an arbitrary \emph{syndrome state}, $\ket{\psi_{\sigma(P)}}\coloneqq P\ket{\psi_0}$ for any $P\in\mathbb{P}^n$, the outcomes $b$ must equal to $\sigma(P)$ since
\begin{gather}
    \prod_{i=1}^n\left(\frac{I+(-1)^{b_i}P_{i}}{2}\right)\ket{\psi_{\sigma(P)}}=\delta(b,\sigma(P))\ket{\psi_{\sigma(P)}}.\notag
\end{gather}

In fact, there is a group structure over all error syndromes (as well as the corresponding Pauli operators).
To view this, note that the direct product of $\mc S_0$ and $\langle\ii I^{\otimes n}\rangle$ is a normal subgroup in $\mathbb{P}^n$.
Consequently, we can define the \emph{anticommutant} of $\mc S_0$ to be the quotient group as: $\mc A_{\mc S_0}\coloneqq\faktor{\mathbb{P}^n}{\langle\ii I^{\otimes n}\rangle\times\mc S_0}$.
This anticommutant consists of all cosets of $\langle\ii I^{\otimes n}\rangle\times\mc S_0$, each of which is closely related to a certain error syndrome:
\begin{proposition}[Adapted from Propositions 3.8 and 3.16 in~\cite{gottesman2016surviving}]\label{prop:Anticom_Isomorphic}
Let $\mc S_0$ of $\mathbb{P}^n$ be a maximal stabilizer group with cardinality of $2^n$.
Given $P$ and $P'\in\mathbb{P}^n$, $\sigma_{\mc S_0}(P)=\sigma_{\mc S_0}(P')$ iff $P$ and $P'$ belong to the same coset in $\mc A_{\mc S_0}$.
Moreover, $\mc A_{\mc S_0}$ is isomorphic to $\mathbb{Z}_2^n$ via the isomorphism $\sigma_{\mc S_0}(\cdot)$.
\end{proposition}
This proposition explicitly construct the isomorphism through the error syndrome of the all cosets, which helps to establish a close relation between the anticommutant $\mc A_{\mc S_0}$ and all possible error syndrome.
Therefore, we can refer to every elements in $\mc A_{\mc S_0}$ as syndromes or syndrome Paulis without ambiguity.

For most purposes in this work, we can ignore the phase of Pauli operators, giving us effectively $4^n$ distinct Paulis.
In this sense, we will frequently switch to the Pauli matrix set ${\sf P}^n$.
In this case, there is a bijective map between ${\sf P}^n$ and $\mathbb{Z}_{2}^{2n}$ given by
\begin{gather}
    \alpha\in\mathbb{Z}^{2n}_2,\ \alpha\longleftrightarrow P_\alpha=P_{\alpha_x\alpha_z}\coloneqq\ii^{\alpha_x\cdot \alpha_z}X[a_x]Z[a_z],\notag
\end{gather}
where $\alpha_x,\alpha_z\in\mathbb{Z}^n_2$ are the first and last $n$ bits of $\alpha$, and $X$ and $Z$ are the standard single-qubit Pauli matrices.
Here we denote $X[\alpha_x] = X^{(\alpha_{x})_1}\otimes \ldots \otimes X^{(\alpha_{x})_n}$, where $(\alpha_{x})_i$ stands for the $i$th bit of the string $\alpha_x$ with $i\in[n]$.
Similar decomposition holds for $Z[\alpha_z]$.
This bijective map implies a natural representation of Pauli operators in ${\sf P}^n$, which we refer to as the \emph{binary symplectic representation}.
Consequently, we will use  $P_\alpha$ and $\alpha$ interchangeably when indicating elements in ${\sf P}^n$.

With this representation,  the commutation relation between Pauli operators in ${\sf P}^n$ can be captured by the \emph{symplectic inner product} of the binary indices as $c(P_\alpha, P_\beta)=\pip{\alpha}{\beta}$, where
\begin{gather}
    \pip{\alpha}{\beta}\coloneqq \alpha_x\cdot \beta_z+\alpha_z\cdot \beta_x\ \bmod2,\ \forall\,\alpha,\beta\in\mathbb{Z}_2^{2n}.\notag
\end{gather}

Moreover, we can consider the maximal stabilizer group $\mc S_0$ with all elements in ${\sf P}^n$ and fix the representatives of all cosets in $\mc A_{\mc S_0}$ by elements in ${\sf P}^n$.
In this sense, we have the following useful results:
\begin{proposition}[Unique decomposition through anticommutant]\label{prop:Unique_Decomp}
    Given a maximal stabilizer group $\mc S_0$ of $\mathbb{P}^n$ with $2^n$ elements all belonging to ${\sf P}^n$ and its anticommutant $\mc A_{\mc S_0}$, the Pauli operator $P_\alpha\in{\sf P}^n$ can be uniquely decomposed as
    \begin{gather}
        P_\alpha=v_{\beta,\gamma}P_\beta P_\gamma, \notag
    \end{gather}
    where $P_\beta\in\mc A_{\mc S_0}$, $P_\gamma\in\mc S_0$, and $v_{\beta,\gamma}$ is a phase factor in $\langle \ii\rangle$.
\end{proposition}

\begin{lemma}[Applying Lemma 1 in~\cite{flammia2020efficient}]\label{lm:Sum_Delta}
    Given a maximal stabilizer group $\mc S_0$ of $\mathbb{P}^n$ with $2^n$ elements all belonging to ${\sf P}^n$ and its anticommutant $\mc A_{\mc S_0}$, we have the following equation for any $P_\gamma\in\mc S_0$:
    \begin{gather}
        \sum_{\beta\in \mc A_{\mc S_0}}(-1)^{\pip{\gamma}{\beta}}=\Tr(P_\gamma).\notag
    \end{gather}
\end{lemma}

\subsection{Norms of linear operators and maps}
In this section, we will introduce some important norms as metrics to quantify linear operators and maps.
To facilitate subsequent analyses and derivations, here we illustrate their definitions and emphasize important properties thereof.

\begin{definition}[Normalized Schatten norms]
    For any linear operator $T:\mc H(d)\rightarrow\mc H(d)$, we define the normalized Schatten $p$-norm of $T$ for $p\in[1,\infty]$ as
    \begin{gather}
        \|T\|_{\mathrm{Schatten},p}\coloneqq\left(\frac{\Tr(\abs{T}^p)}{d}\right)^{1/p}.\notag
    \end{gather}
\end{definition}

The normalized Schatten norm is prevalent in the study of bounded linear operators, such that there are specific choices of $p$ that are extensively used.
For example, $p=2$ implies the \emph{normalized Frobenius norm} $\|\cdot\|_F$, and the norm with $p=\infty$ is conventionally denoted by the operator norm $\|\cdot\|$.
For the whole family of the normalized Schatten norms, we can find a clear trend in the norm values according to different $p$:
\begin{fact}\label{fact:Schatten_order}
    For any linear operator $T:\mc H(d)\rightarrow\mc H(d)$ and $1\leq p_1\leq p_2\leq\infty$, we have
    \begin{gather}
        \|T\|_{\mathrm{Schatten},p_1}\leq\|T\|_{\mathrm{Schatten},p_2}.\notag
    \end{gather}
\end{fact}

There exists a second family of norms that is closely related to the Pauli decomposition of an operator.
Note from the previous discussion that the Pauli operators in $\mathsf{P}^n$ form an orthogonal basis for the operator space.
Therefore, we can define the following norm for every operator:

\begin{definition}[Pauli norms]
    For an $n$-qubit operator $T$ with its Pauli decomposition $T = \sum_{\alpha \in \mathsf{P}^n} s_{\alpha} P_{\alpha}$, we define the Pauli $p$-norm of $T$ for $p\in[1,\infty]$ to be the $\ell_p$-norm of the vector $\vec{s} = \rbra{s_{\alpha}}_{\alpha \in \mathsf{P}^n}$, i.e., \begin{gather}
        \Abs{T}_{\textup{Pauli},p} \coloneqq \Abs{\vec{s}}_p = \left(\sum_{\alpha \in \mathsf{P}^n} \abs{s_\alpha}^{p}\right)^{1/p}.\notag
    \end{gather}
\end{definition}

Note that the Pauli $2$-norm is closely related to the normalized Schatten norm.
As an example, we find that $\Abs{H}_F = \Abs{H}_{\textup{Pauli},p}$.
Moreover, since the Pauli norms are induced from the $\ell_p$ norms, there also exists a clear trend in the values of Pauli norms with different $p$:

\begin{fact}\label{fact:Pauli_order}
    For any linear operator $T:\mc H(d)\rightarrow\mc H(d)$ and $1\leq p_1\leq p_2\leq\infty$, we have
    \begin{gather}
        \|T\|_{\mathrm{Pauli},p_1}\geq\|T\|_{\mathrm{Pauli},p_2}.\notag
    \end{gather}
\end{fact}

Besides the norms of operators, we will also use the norm of linear maps.
In this sense, we introduce the well-known diamond norm:
\begin{definition}[Diamond norm]\label{def:Diamond_Norm}
    For any linear map $\mc E$ between linear operator spaces: $\mc E:{\sf L}(\mc H(d))\rightarrow{\sf L}(\mc H(d))$, we define the diamond norm of the map by
    \begin{gather}
        \|\mc E\|_\diamond\coloneqq\max\{\|\mc E\otimes \mc I_{d'}(X)\|_{\mathrm{Schatten},1}:\,X\in{\sf L}(\mc H(d)\otimes\mc H(d')),d'\geq0,\ |X\|_{\mathrm{Schatten},1}\leq1\}.\notag
    \end{gather}
\end{definition}
\noindent Note that we have adapted the definition of the diamond norm from the Schatten 1-norm to the normalized version.
Fortunately, this would not change the form of the definition since there is an automatic normalization for the input state.
This diamond norm is used to quantify the largest enlargement effect of the map in a quantum operation.
In this sense, we discover an important property of the physical completely positive and trace-preserving (CPTP) maps, which we refer to as channels.
\begin{fact}
    Any CPTP linear map $\mc E:{\sf L}(\mc H(d))\rightarrow{\sf L}(\mc H(d))$ has a unit diamond norm.
\end{fact}

According to the definition of the diamond norm, we can further use the diamond norm of the differences between two arbitrary channels to represent the distance.
In fact, this diamond-norm distance faithfully represents the distinguishability between the two channels, as analyzed in~\cite{wilde2011classical}.
Typically, given the unitary channels which are operations from unitary matrices applied to the state, we have the close relation between the diamond-norm distance and the operator norm distance
(see for example Proposition I.6 in~\cite{haah2023query}):
\begin{proposition}
\label{prop:diamond}
    Suppose $\mathcal{U}(\cdot)\coloneqq U(\cdot)U^\dag$ and $\mathcal{V}(\cdot)\coloneqq V(\cdot)V^\dag$ are two unitary channels.
    The diamond-norm distance between $\mathcal{U}$ and $\mathcal{V}$ can 
    be bounded as
    \begin{gather}
        \|\mathcal{U}-\mathcal{V}\|_\diamond\leq 2\|U-V\|.\notag
    \end{gather}
\end{proposition}

\subsection{Trotter-Suzuki Hamiltonian simulation}\label{sec:Trotter}
In this section, we will introduce one of the most prevalent methods, the \emph{Trotter-Suzuki formula}, for Hamiltonian simulation.
Consider a Hamiltonian $H = A+B$ with two terms.
The Lie-Trotter formula provides a first-order approximation of its evolution for small-time steps $\delta t$ by mixing up these two sub-evolutions:
\begin{gather}
    e^{-\ii H\delta t}\approx U_1(\delta t)\coloneqq e^{-\ii A\delta t} e^{-\ii B\delta t}.\notag
\end{gather}
A more accurate second-order approximation can be achieved through a symmetric arrangement of the sub-evolutions with a more fine-grained mixing:
\begin{align}
    e^{-\ii H\delta t}\approx U_2(\delta t)\coloneqq e^{-\ii B\delta t/2} e^{-\ii A\delta t}e^{-\ii B\delta t/2}.\notag
\end{align}
This back-and-forth structure significantly improves the simulation accuracy, incurring only third-order error terms, as implied by the following theorem.
\begin{proposition}\label{prop:Trotter_Error}
    Let $A$ and $B$ be two Hamiltonians 
    over an $n$ qubit system, and 
    $H=A+B$.
    For $t\geq 0$ and $r\in\mathbb{Z}_+$, denote $\delta t\coloneqq t/r$.
    Let $U_2(\delta t) \coloneqq e^{-\ii B\delta t/2} e^{-\ii A\delta t}e^{-\ii B\delta t/2}$ .
    Then, the $r$-step Trotter error can be bounded as
    \begin{gather}
    \|e^{-\ii Ht}-U_2(\delta t)^r\|\leq\frac{t^3}{12r^2}\|[A,[A,B]]\|+\frac{t^3}{24r^2}\|[B,[B,A]]\|.\notag
\end{gather}
\end{proposition}
\begin{proof}
    From the Proposition~10 in~\cite{childs2021theory}, we know for any $\delta t\geq0$, the additive error be bounded by
    \[
    \|e^{-\ii H\delta t}-U_2(\delta t)\|\leq\frac{(\delta t)^3}{12}\|[A,[A,B]]\|+\frac{(\delta t)^3}{24}\|[B,[B,A]]\|.
   \]

Since $t/\delta t$ is an integer, the error accumulates linearly with the number of steps
by the following standard argument that decomposes
the error into short-time .

\[
\begin{aligned}
    \|e^{-\ii H t}-U_2(\delta t)^{r}\|
    &= \Abs*{\sum_{j=0}^{r-1} \rbra*{e^{\ii H \delta t (r-j)}U_2(\delta t)^j - e^{\ii H \delta t (r-j-1)}U_2(\delta t)^{j+1}}} \\
     & \leq \sum_{j=0}^{r-1} \Abs*{e^{\ii H \delta t (r-j)}U_2(\delta t)^j - e^{\ii H \delta t (r-j-1)}U_2(\delta t)^{j+1}} \\
     &\leq r \|e^{-\ii H \delta t}-U_2(\delta t)\| \\
     &\leq\frac{t^3}{12r^2}\|[A,[A,B]]\|+\frac{t^3}{24r^2}\|[B,[B,A]]\|.
\end{aligned}
\]

\end{proof}
From the proposition, a crucial advantage of the second-order Trotter-Suzuki method is that by increasing the number of steps $r$, we can arbitrarily suppress the Trotter error without changing the total evolution time of $H$.
This provides a controllable approximation method for implementing the arbitrary Hamiltonian evolution needed in our certification protocol.

\subsection{Bernoulli distribution parameter testing}

In this part, we will illustrate a classical  method to decide the range of the parameter of a Bernoulli 
distribution
with optimal sample and time complexity.
To this end, we recall the Hellinger distance,
which characterizes the optimal sample complexity
of the distribution discrimination problem.

\begin{definition}
\label{def:Hellinger-distance}
Given two discrete probability distributions $\mc P, \mc Q$ over
$[n]$, the Hellinger distance $d_H(\mc P, \mc Q)$
between them is defined as:
\begin{gather}
d_H(\mc P, \mc Q) :=
\sqrt{\frac{1}{2} \sum_{i \in [n]} (\sqrt{p_i} - \sqrt{q_i})^2} = \sqrt{1 - \sum_{i \in [n]} \sqrt{p_i q_i}}.\notag
\end{gather}
\end{definition}

For two distributions
$\mathcal{P}$ and $\mathcal{Q}$,
there is an algorithm 
to discriminate them
with optimal sample complexity $\widetilde{\Theta}(1/d_H(\mathcal{P},\mathcal{Q}))$.
We modify this distribution
discrimination algorithm
to obtain
a sample and time efficient algorithm
of deciding
whether $p\le b$ or $p\ge a$
for a Bernoulli distribution with parameter $p$,
as introduced in the following propostion.

\begin{proposition}[Bernoulli distribution parameter testing]
\label{prop:Bernoulli-distribution-parameter-testing}
    Let $\mc{P}$ be a Bernoulli probability distributions with
    parameter $p$. Suppose $0\le b < a \le 1$.
    Then,  given sample access $\mathcal{O}_S$,
    there exists an algorithm
    $\mathsf{BernoulliTest}(a,b,\mathcal{O}_S,\delta)$ that,
    correctly decides either
    $p \ge a$ (in which the algorithm outputs \enquote{large}) or $p \le b$ (in which the algorithm outputs \enquote{small}), with promise that exactly one case occurs. 
    In both cases, the algorithm succeeds with probability at least $1-\delta$.
    Moreover, the algorithm uses
    $O(\log (1/\delta)/{(\sqrt{a}-\sqrt{b})^2})$ samples,
    and runs in
    $O(\log (1/\delta)/{(\sqrt{a}-\sqrt{b})^2})$ time.
\end{proposition}

The detailed description of the algorithm and the proof can be found in \cref{sec:Appendix_Bernoulli_Test}.

\subsection{Quantum Bernoulli distribution parameter testing}
\label{sec:quantum-distribution-discrimination}

We first recall the following square root 
amplitude estimation algorithm,
which is a generalization of the 
original amplitude estimation algorithm
in .
The square root amplitude estimation
algorithm could estimate
the square root of the measurement 
probability within additive error $\varepsilon$,
with $\widetilde{O}(1/\varepsilon)$
query complexity.
The details of the algorithm is introduced
in the following theorem.

\begin{theorem}[Square root amplitude estimation,
adapted from~{\cite[Theorem III.4]{Wan24}}]
\label{thm:sqrt-amp-est}
    Suppose there is a unitary $U$ satisfying
    \[
        U\ket{0} = \sqrt{p}\ket{0}\ket{\psi_0} +
        \sqrt{1-p}\ket{1}\ket{\psi_1},
    \]
    for some $p\in \interval{0}{1}$,
    normalized states $\ket{\psi_0}$ and $\ket{\psi_1}$.
    
    Then,
    there is a quantum algorithm
    $\mathsf{SqrtAmpEst}(U, \varepsilon, \delta)$,
    such that for any $\varepsilon > 0$
    and $\delta\in \interval[open]{0}{1/3}$,
    with probability at least $1-\delta$, it outputs an estimate $\mu$ satisfying
    \[
        \abs{\mu - \sqrt{p}} \le \varepsilon,
    \]
    using $O(\log(1/\delta)/\varepsilon)$
    queries to controlled-$U$ and controlled-$U^{\dagger}$, and performing $O\rbra{\log\rbra{1/\varepsilon}+\log\rbra{1/\delta}}$ one-qubit measurements.
\end{theorem}

Utilizing the above theorem,
we could design a quantum 
algorithm that tests
whether the amplitude $\sqrt{p}$
is larger than $a$
or smaller than $b$
using $\widetilde{O}(1/(a-b))$ queries,
which is formally stated below.

\begin{theorem}[Amplitude testing]
\label{thm:sqrt-amp-testing}
    Suppose there is a unitary $U$ satisfying
    \[
        U\ket{0} = \sqrt{1-p}\ket{0}\ket{\psi_0} +
        \sqrt{p}\ket{1}\ket{\psi_1},
    \]
    for some $p\in \interval{0}{1}$,
    normalized states $\ket{\psi_0}$ and $\ket{\psi_1}$.
    For $0 \le b < a \le 1$,
    suppose we are promised that either
    $\sqrt{p} \ge a$ or $\sqrt{p}\le b$.
    Then,
    there is a quantum algorithm
    $\mathsf{AmpTest}(U, a,b, \delta)$,
    such that for any
    $\delta\in \interval[open]{0}{1/3}$,
    with probability at least $1-\delta$,
    successfully decides
    $\sqrt{p} \ge a$ or $\sqrt{p}\le b$,
    using $O(\log(1/\delta)/(a-b))$
    queries to controlled-$U$ and controlled-$U^{\dagger}$ and $O\rbra{\log\rbra{1/(a-b)}+\log\rbra{1/\delta}}$ one-qubit measurements.
\end{theorem}

\begin{proof}
    The algorithm
    works by first calling
    $\mathsf{SqrtAmpEst}((X\otimes I)U, (a-b)/3, \delta)$.
    By \cref{thm:sqrt-amp-est},
    the algorithm will output
    an estimate $\mu$
    satisfying
    $\abs{\mu - \sqrt{p}}\le \frac{a-b}{3}$,
    using $O(\log(1/\delta)/(a-b))$
    queries to controlled-$U$ and controlled-$U^{\dagger}$,
    and performing $O\rbra{\log\rbra{1/(a-b)}+\log\rbra{1/\delta}}$ one-qubit measurements.
    Then, the algorithm outputs
    $\sqrt{p} \ge a$
    if $\mu \ge \frac{a-b}{2}$,
    and
    $\sqrt{p}\le b$
    otherwise.

\end{proof}

\section{Pauli coefficient analysis of Hamiltonian evolution}\label{sec:Pauli-coefficient-analysis}
In our quantum Hamiltonian certification problem, a fundamental challenge involves
extracting information from evolution dynamics to certify the underlying physical
processes. This part exhibits a mathematical connection between the Pauli coefficients of a Hamiltonian
operator and its real-time evolution.

Consider an $n$-qubit traceless Hamiltonian operator $H$ with its Pauli decomposition
\begin{equation*}
  H=\sum_{\alpha\in{\mc S}}s_\alpha P_\alpha,
\end{equation*}
where $P_\alpha \in {\sf P} ^n$ represents a Pauli operator and the size of the index set $\mc S$ is
$\abs{\mc S}=m$.
To analyze the structure of the time evolution operator, we expand it using the Taylor series:
\begin{equation*}
  e^{-\ii Ht}  = I - \ii t\sum_{\alpha\in{\mc S}}s_\alpha P_\alpha +\sum_{k=2}^\infty \frac{{(-\ii t)}^k{(\sum_{\alpha\in{\mc S}}s_\alpha P_\alpha)}^k}{k!},
\end{equation*}
where the $k$th-order term contains $m^k$ products of Pauli
operators.

While this expansion provides a formal expression for the evolution operator, it
does not immediately reveal how the Pauli coefficients of $H$ relate to the
coefficients of the evolution operator. To bridge this gap, we make the
following key observation about the distribution of coefficients in the Taylor
expansion:

\begin{lemma}\label{lemma:evolution-Pauli-efficient-bound}
Let $H = \sum_{\alpha\in \mathcal{S}} s_{\alpha} P_{\alpha}$ be an $n$-qubit traceless
Hamiltonian with $m$ non-zero Pauli terms and $\abs{s_{\alpha}}\le S$ for all $\alpha \in {\mc S}$. For any real
$t$, expand $e^{-\ii Ht}$ as an infinite series of $t$ with operator-valued
coefficients, i.e.,
\[
  e^{-\ii Ht} = I - \ii A_1 t + \sum_{k=2}^{\infty} \frac{(-\ii)^kA_k}{k!}
  t^k.
\]
Define
\[
\mathcal{S}_{k,\ell} = \{(j_1,j_2, \dots, j_k) \in \mathcal{S}^k
\mid \exists a\in \{0,1,2,3\},P_{j_1}P_{j_2}\cdots P_{j_k} = \ii^a P_{\ell}\}.
\]
Then, for $A_k = \sum_{\ell} a_{k,\ell} P_{\ell}$,
we have
\begin{itemize}
  \item $\{\mathcal{S}_{k,\ell}\}_{\ell\in \mathsf{P}^n}$ forms a partition of $\mathcal{S}^k$, with
        $\sum_{\ell} \abs{\mathcal{S}_{k,\ell}} = m^k$.
  \item $\abs{\mathcal{S}_{k,\ell}}\le m^{k-1}$.
  \item $\abs{a_{k,\ell}}\le \abs{S_{k,\ell}} s^k$.
\end{itemize}
Furthermore, we have $\abs{a_{k,\ell}}\le m^{k-1} S^k $, and $\sum_{\ell} \abs{a_{k,\ell}} \le m^k S^k $.
\end{lemma}

The above observation provides a handy bound the coefficient $a_{k,\ell}$ in the $k$-th order Taylor expansion by the number of Pauli terms $\abs{\mathcal{S}_{k,\ell}}$, which will be extremely useful in proving the following propositions.

We now introduce our new upper and lower
bounds of the Pauli coefficient vector
of the evolution operator deduced by detailed analysis
of the Taylor coefficients,
with the help of \cref{lemma:evolution-Pauli-efficient-bound}.
\begin{lemma}[$\ell_2$-norm lower bounds for the Pauli coefficient vector of the
 evolution operator]\label{lm:Pauli_Lower_Bound2}
  Let $H = \sum_{\alpha\in \mathcal{S}} s_{\alpha} P_{\alpha}$ be an $n$-qubit traceless Hamiltonian
  with $m$ non-zero Pauli terms and $\abs{s_{\alpha}}\le S$.
  For any positive real $t$, 
  denoting $R:= \frac{1}{m}\sum_{k=2}^{\infty} \frac{(mtS)^k}{k!}$, 
  let $\vec{v}$ be the Pauli coefficient of $e^{-iHt}$.
  Then, for any set $\mathcal{X}\subseteq \mathsf{P}^n\setminus \{I\}$ with
  $\abs{\mathcal{X}} \ge m$, we have
  \[
    \Abs{\vec{v}[\mathcal{X}]}_2 \ge \Abs{\vec{s}[\mathcal{S}\cap\mathcal{X}]}_2t -\sqrt{m}R.
  \]
\end{lemma}

\begin{lemma}[$\ell_2$-norm upper bounds for the Pauli coefficient vector of the evolution operator]\label{lm:Pauli_Upper_Bound}
    Let $H = \sum_{\alpha\in \mathcal{S}} s_{\alpha} P_{\alpha}$
    be an $n$-qubit traceless Hamiltonian
    with $m$ non-zero Pauli terms
    and $\abs{s_{\alpha}}\le S$.
    For any positive real $t$,
    denoting $R:= \frac{1}{m}\sum_{k=2}^{\infty} \frac{(mtS)^k}{k!}$,
    let $\vec{v}$ be the Pauli coefficient of
    $e^{-iHt}$.
    Then, for any set $\mathcal{X}\subset \mathsf{P}^n\backslash \{I\}$
    with $\abs{\mathcal{X}}\ge m$,
    we have
    \[
      \Abs{\vec{v}[\mathcal{X}]}_2^2 \le \sum_{\beta\in \mathcal{S} \cap \mathcal{X}} (\abs{s_{\beta}}t + R)^2
      + (m-\abs{ \mathcal{S}\cap \mathcal{X}})R^2.
    \]
\end{lemma}

For completeness, 
we give detailed proofs of the above lemmas in Appendix~\ref{sec:pauli-coeffcient-analysis-proof}.

\section{Bell-state assisted certifications}\label{sec:ancilla_assisted}
In this section, we will introduce algorithms for both Hamiltonian certification and certification with only one-sided errors.
As mentioned previously, the certification task aims to determine whether an unknown Hamiltonian $H$ is close to the target $H_0$.
As an equivalence, we can reduce the task to determine if the residual Hamiltonian $H_\mathrm{res}\coloneqq H-H_0$ is close to zero.
Based on this idea, we will show that the ancillary system can be adapted to fulfill the task.

\subsection{Hamiltonian amplitude encoding}

In this section, we introduce a unitary implementation that encodes information from the time evolution of an unknown residual Hamiltonian $H_\mathrm{res}$ into quantum amplitudes.
The key insight of our approach is to leverage the unique properties of Bell states to disperse Pauli components of the residual Hamiltonian evolution into orthogonal subspaces. 
This allows us to effectively \enquote{read out} the difference between the unknown Hamiltonian $H$ and the target Hamiltonian $H_0$.

\begin{figure}
    \centering
\begin{adjustbox}{width=\textwidth}
\begin{quantikz}
\lstick{$\ket{0}$} & & \qw  & \qw & \qw  & \qw & \qw & \qw & \qw & \qw & \gate[3]{V} & \qw \\
\lstick{$\ket{0}^{\otimes n}$} & \qwbundle{n} & \gate[2]{U_{\mathsf{Prep}}}  & \gate{e^{\ii H_0 t/2r}}\gategroup[1,steps=7,style={dashed,rounded corners,inner sep=2pt}]{Repeat for $r$ times} & \gate{\mathcal{O}_H(t/r)} & \gate{e^{\ii H_0 t/2r}} &\ \ldots\ & \gate{e^{\ii H_0 t/2r}}&\gate{\mathcal{O}_H(t/r)} & \gate{e^{\ii H_0 t/2r}} \qw & \qw  & \qw \\
\lstick{$\ket{0}^{\otimes n}$} & \qwbundle{n} & \qw & \qw & \qw & \qw & \qw  & \qw  & \qw & \qw & \qw  & \qw
\end{quantikz}
\end{adjustbox}
\caption{Hamiltonian Amplitude Encoding: $\mathsf{HAE}(H_0,\mc O_H,t,r)$. The Bell preparation unitary $U_{\mathsf{Prep}}$ is described in Eq.~\eqref{eq:Bell_preparation}, and the multi-qubit gate $V$ is defined in Eq.~\eqref{eq:coherent_operation}.}
\label{fig:HAE}
\end{figure}
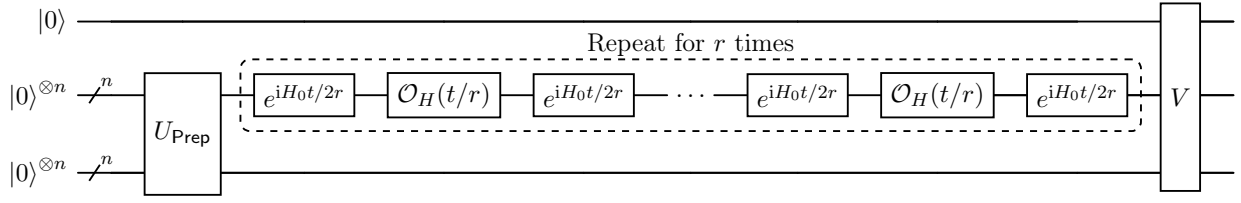

The circuit for Hamiltonian amplitude encoding unitary is summarized in Figure~\ref{fig:HAE}, which uses a total of $2n+1$ qubits.
In the first stage of the encoding, we apply the state preparation unitary $U_{\mathsf{Prep}}$ of the $2n$-qubit Bell state, where
\begin{gather}\label{eq:Bell_preparation}
    U_{\mathsf{Prep}}\ket{0}^{\otimes 2n} = \frac{1}{\sqrt{2^n}} \sum_j\ket{jj}=:\ket{\Phi^+}.
\end{gather}
The Bell state serves as an ideal probe due to a remarkable property:

\begin{proposition}\label{prop:Bell_Dispersion}
    Let $U$ be an $n$-qubit unitary operator with Pauli decomposition $U=\sum_{\alpha\in{\sf P}^n}u_\alpha P_\alpha$.
    Given a $2n$-qubit Bell state $\ket{\Phi_{2n}^+}$, applying $U\otimes I_{2^n}$ on it yields the final state 
    \begin{gather}
        (U\otimes I_{2^n})\ket{\Phi_{2n}^+}=\sum_{\alpha\in{\sf P}^n}u_\alpha\ket{\Phi_\alpha},\notag
    \end{gather}
    where $\cbra{\ket{\Phi_{\alpha}}\coloneqq (P_\alpha\otimes I_{2^n})\ket{\Phi^+}}_{\alpha \in \mathsf{P}^n}$ forms an orthonormal basis.
\end{proposition}

\noindent This proposition reveals how the Bell state effectively disperses the Pauli components of a unitary into orthogonal subspaces in the Hilbert space.
In fact, $\{\ket{\Phi_\alpha}\}$ is known as the \emph{Bell basis}.

Particularly, in our Hamiltonian certification setting, we want to examine whether the residual Hamiltonian $H_\mathrm{res}$ is close to zero.
Since we only have query access to $H$ and a classical description of $H_0$, we utilize the second-order Trotter-Suzuki with a total Trotter step $r$ to approximate the time evolution $e^{-\ii H_\mathrm{res}t}$, as stated in the loop part of the algorithm.
According to Proposition~\ref{prop:Bell_Dispersion}, applying the time evolution to the Bell state generates a superposition over $\{\Phi_\alpha\}$ from different orthogonal subspaces with coefficients equal to the Pauli coefficients of $e^{-\ii H_\mathrm{res}t}$.

In the last stage of the \textsc{HAE}, we introduce a coherent operation $V$ that marks the subspaces corresponding to non-identity Pauli-rotated Bell states.
\begin{gather}\label{eq:coherent_operation}
    V\coloneqq I_2\otimes\ket{\Phi^+}\bra{\Phi^+}+X\otimes\left(\sum_{\alpha\neq I} (P_\alpha\otimes I)\ket{\Phi^+}\bra{\Phi^+}(P_\alpha^{\dagger}\otimes I)\right).
\end{gather}
Typically, $V$ applies an $X$ gate to the ancilla when the state is any $\ket{\Phi_\alpha}\neq\ket{\Phi^+}$, while doing nothing when the state is $\ket{\Phi^+}$.
This incorporation allows us to effectively realize an amplitude encoding unitary that satisfies the following statement:

\begin{proposition}[Properties of the Hamiltonian amplitude encoding]\label{prop:noisy_bound}
Given a classical description of a traceless Hamiltonian $H_0$, and a query access $\mathcal{O}_H$ to the time evolution of an unknown traceless Hamiltonian $H$.
Suppose  $H_\mathrm{res}\coloneqq H-H_0=\sum_{\alpha\in{\sf P}^n}s_\alpha P_\alpha$
has at most $m_0$ non-zero Pauli terms,
with $\vec{s}$ being its Pauli-coefficient vector, and $\mc S$ being the support of $\vec{s}$.
We denote $S\coloneqq\|\vec{s}\|_\infty$.
For any $t\geq 0$ and integer $r\geq1$, let $U_{\mathsf{HAE}} :=\mathsf{HAE}(H_0,\mc O_H,t,r)
= VU_{t,r}U_{\mathsf{Prep}}$.
Then, it holds that
\begin{gather}\label{eq:validity-appendix}
    U_{\mathsf{HAE}}\ket{0}^{\otimes (2n+1)}= p_0(t)\ket{0}\ket{\phi_0} + p_1(t)\ket{1}\ket{\psi_1},
\end{gather}
where $p_0(t), p_1(t) \in \interval{0}{1}$, $\ket{\psi_0}$ and $\ket{\psi_1}$ are normalized states. We call $p_1(t)$
the signal amplitude.
For $t\in\interval{0}{1/Sm_0}$, 
denote $R\coloneqq\sum_{k=2}^\infty\frac{m_0^{k-1}(St)^k}{k!}$.
Then, we have
\begin{itemize}
    \item The error comes from the Trotter-Suzuki formula can be bounded as
    \[
    \epsilon_\mathrm{Trotter} := \Abs{e^{-\ii H_\mathrm{res} t}-U_{t,r}}
    =O\left(\frac{Sm_0\|H_0\|(\|H_0\|+Sm_0)t^3}{r^2}\right).
    \]
    \item The signal amplitude  has a lower bound
    \[
      p_1(t)\geq(\|\Vec{s}\|_2t-\sqrt{m_0}R)-\epsilon_\mathrm{Trotter}.
    \]
    \item The signal amplitude has an upper bound
    \[
    p_1(t)\leq\sqrt{\sum_{\alpha\in \mathcal{S}}(\abs{\vec{s}_\alpha} t+R)^2}+\epsilon_\mathrm{Trotter}.
    \]
\end{itemize}
\end{proposition}

We outline its proof in the following,
leaving a detailed one to \cref{sec:Appendix_Ancilla_Encoding}.
\begin{proof}[Proof Sketch]
    The state resulting from applying $U_{\mathsf{HAE}}$
    on $\ket{0}$ could be computed 
    by using the Pauli decomposition
    studied in \cref{prop:Bell_Dispersion}.
    
    The Trotter error $\epsilon_\mathrm{Trotter}$
    can be bounded by directly applying
    \cref{prop:Trotter_Error}.
    We could therefore replace $U_{\mathsf{UAE}}$
    with $e^{\ii H_{\mathrm{res}}t}$ in the following analysis at the cost of adding a term
    $\epsilon_\mathrm{Trotter}$.

    Utilizing \cref{prop:Bell_Dispersion},
    we could establish the relation between
    $p_1(t)$ and the Pauli coefficients
    of $e^{\ii H_{\mathrm{res}}t}$. The upper and
    lower bounds of $p_1(t)$ then follows by using \cref{lm:Pauli_Lower_Bound2,lm:Pauli_Upper_Bound}.
\end{proof}

This lemma establishes a direct connection between the signal amplitude $p_1(t)$ and the Pauli coefficients of $H_\mathrm{res}$, providing us with a quantum encoding that will be the foundation for our certification protocols in subsequent subsections.

\subsection{Robust coherent Hamiltonian certification}

In this part, we present our algorithm $\mathsf{RCHC}$ (Algorithm~\ref{alg:TCHC})
for solving the robust version of
the Hamiltonian certification problem, where \enquote{robust} indicates that we consider non-zero $\varepsilon_1$'s.
The $\mathsf{RCHC}$ algorithm is conceptually concise.
Given the target Hamiltonian $H_0$ and query 
access to the time evolution operator $e^{-\ii Ht}$,
the algorithm first construct the 
Hamiltonian amplitude encoding $U_{t,r}$
using the circuit in \cref{fig:HAE},
with carefully chosen parameter $t$
and $r$.
As shown in \cref{prop:iter_v_Tolerant}, 
this choice results in a significant difference in the signal amplitude 
of the Hamiltonian amplitude encoding, depending on whether $H$ is close to	
$H_0$ or not.
Consequently, we could efficiently test the magnitude of the amplitude
through the procedure $\mathsf{AmpTest}$,
achieving an optimal total evolution time.

\begin{algorithm}[t]
\caption{Robust Coherent Hamiltonian Certification: $\mathsf{RCHC}(H_0,\mc O_H,m,B,\varepsilon_1,\varepsilon_2,\delta)$}\label{alg:TCHC}
    \SetKwInOut{Input}{Input}\SetKwInOut{Output}{Output}
    \SetKwInOut{Para}{Parameters}
    \Input{Target Hamiltonian $H_0$, query access $\mc O_H$ to the single-qubit controlled $e^{-\ii Ht}$ with its inverse, Pauli term number upper bound $m$, Pauli coefficient upper bound $B$, precision $\varepsilon_1,\varepsilon_2$, and failure probability $\delta$.}
    \Output{\textsc{Accept} when the two Hamiltonians are close; \textsc{Reject} otherwise.}
    $m_0\leftarrow2m$,
    $B_0\leftarrow2B$, $\eta\gets\frac{\varepsilon_2-\varepsilon_1}{\varepsilon_1}$\;
    \eIf{$\eta\leq1$}{
    $c_1\gets \eta+0.32\eta$, $c_2\gets \eta+0.68\eta^2-0.32\eta^3$, $\xi\gets\eta$\tcp*{See Propositions~\ref{prop:iter_v_Tolerant}}}{
    $c_1\gets 1.32$, $c_2\gets 0.68+0.68\eta$, $\xi\gets1$\;
    }
    $t\leftarrow \frac{\xi}{2m_0^{3/2}B_0}$, and $r\leftarrow \Theta(1)$\;
    Let $U_{t,r}$ denote $\mathsf{HAE}(H_0,Q_H,t,r)$\;
    $l\leftarrow \textsf{AmpTest}\left(U_{t,r}, \frac{c_2\varepsilon_1}{m_0^{3/2}B_0}, \frac{c_1\varepsilon_1}{m_0^{3/2}B_0},\delta\right)$\tcp*{See \cref{thm:sqrt-amp-testing}}
        \eIf{$l=\textsc{Large}$}{
            \Return \textsc{Reject}\;
        }{
            \Return \textsc{Accept}\;
        }
\end{algorithm}

The following proposition shows
how the choices of the parameters
in the Hamiltonian amplitude encoding
affect the signal amplitude 
in both cases,
which plays a crucial role in analyzing
the correctness of the algorithm $\mathsf{RCHC}$.
\begin{proposition}\label{prop:iter_v_Tolerant}
Let Hamiltonians $H$, $H_0$ and $H_{\mathrm{res}}$, vector $\vec{s}$,
    number of terms $m_0$, set $\mathcal{S}$, 
    $\ell_{\infty}$-norm $S$, and signal amplitude $p_1(t)$ be the same as defined in \cref{prop:noisy_bound}.
    Suppose $S$ is upper bounded by $B_0$.
Then for any $\varepsilon\in(0,1/2)$, and $\eta>0$, setting $\xi=\min(\eta,1)$, $t=\frac{\xi}{2m_0^{3/2}B_0}$ and $r=\Omega(\|H_0\|^{1/2}m_0^{-1/2}B_0^{-1/2}+\|H_0\|m_0^{-1}B_0^{-1})$ yields:
    \begin{itemize}
        \item If $\|\Vec{s}\|_2\leq\varepsilon$, then the signal amplitude of $\mathsf{HAE}(H_0,\mc O_H,t,r)$ satisfies $p_1(t)\le \frac{c_1\varepsilon}{m_0^{3/2}B_0}$;
        \item If $\|\Vec{s}\|_2\geq(1+\eta)\varepsilon$, then the signal amplitude of $\mathsf{HAE}(H_0,\mc O_H,t,r)$ satisfies $p_1(t)\geq \frac{c_2\varepsilon}{m_0^{3/2}B_0}$.
    \end{itemize}
    Particularly, we have $c_1=\eta+0.32\eta^2$ and $c_2=\eta+0.68\eta^2-0.32\eta^3$ for $\eta\leq1$, and $c_1=1.32$ and $c_2=0.68+0.68\eta$ for $\eta>1$.
\end{proposition}

The detailed proof of this result is provided in Section~\ref{sec:Appendix_Tolera_Hamilt_Certi}.

Now we are able to present our main
result for the robust 
Hamiltonian certification task.

\begin{theorem}\label{thm:tolerant}
    Suppose we are given a classical description of a traceless Hamiltonian $H_0$ and a query access $\mathcal{O}_H$ to the controlled time evolution of an unknown traceless Hamiltonian $H$, including negative time.
    Suppose both Hamiltonians consist of at most $m$ non-zero Pauli operators and their coefficients are bounded by $B$.
    For arbitrary $0<\varepsilon_1<\varepsilon_2<1$, and $\delta\in(0,1/3]$, running $\mathsf{RCHC}(H_0,\mc O_H,m,B,\varepsilon_1,\varepsilon_2,\delta)$ (Alg.~\ref{alg:TCHC}) can distinguish with success probability at least $1-\delta$ the following two cases:
    \begin{itemize}
        \item \textsc{Accept}: $\|H-H_0\|_F\leq\varepsilon_1$,
        \item \textsc{Reject}: $\|H-H_0\|_F\geq\varepsilon_2$,
    \end{itemize}
    promised that it is in either case, where $\|\cdot\|_F$ is the normalized Frobenius norm.
    The algorithm requires 
    \[
    T=O\left((\varepsilon_2-\varepsilon_1)^{-1}\log\delta^{-1}\right)
    \]
    queried evolution time, and a total single-qubit measurement count of
    \[
    O\left(\log(mB(\varepsilon_2-\varepsilon_1)^{-1})+\log(\delta^{-1})\right).
    \] 
    The query complexity for $\varepsilon_2\leq2\varepsilon_1$ is
    \[O\left(m^{3/2}B\varepsilon_1(\varepsilon_2-\varepsilon_1)^{-2}\log(\delta^{-1})\right).
    \]
    Otherwise, it is
    \[O\left(m^{3/2}B(\varepsilon_2-\varepsilon_1)^{-1}\log(\delta^{-1})\right).
    \]
\end{theorem}

The proof of the above theorem is relatively straightforward by applying 
previous propositions and theorems
about Hamiltonian amplitude encoding 
and quantum amplitude testing.
For completeness, we give a detailed 
proof in \cref{sec:Appendix_Tolera_Hamilt_Certi}.

\subsection{Coherent Hamiltonian certification with one-sided error}

Our robust coherent Hamiltonian certification algorithm (with respect
to the normalized Frobenius norm)
naturally leads to several meaningful results, which we discuss below.

An important scenario is the one-sided error
case, where in the \textsc{Accept} case
we are guaranteed that $H= H_0$.
Applying \cref{thm:tolerant} for $\varepsilon_1=\varepsilon/2>0$ and $\varepsilon_2=\varepsilon$,
we obtain the following result.

\begin{theorem}\label{thm:Bell+control}
    Suppose we are given a classical description of a traceless Hamiltonian $H_0$ and a query access $\mathcal{O}_H$ to the controlled time evolution of an unknown traceless Hamiltonian $H$ with its inverse.
    Assume both Hamiltonians consist of at most $m$ non-zero Pauli terms and their Pauli coefficients are bounded by $B$.
    For any $\varepsilon\in\interval[open]{0}{1}$ and $\delta\in(0,1/3]$, 
    there is a quantum algorithm
    $\mathsf{CHC}(H_0,\mc O_H,m,B,\varepsilon,\delta)$ that,
    with probability at least $1-\delta$, 
    can distinguish  the following two cases:
    \begin{itemize}
        \item \textsc{Accept}: $H=H_0$,
        \item \textsc{Reject}: $\|H-H_0\|_F\geq\varepsilon$,
    \end{itemize}
    promised that it is in either case, where $\|\cdot\|_F$ is the normalized Frobenius norm.
    Moreover, the algorithm requires \[T=O\left(\varepsilon^{-1}\log\delta^{-1}\right)\] queried evolution time, 
    \[
    O\left(m^{3/2}B\varepsilon^{-1}\log(\delta^{-1})\right)
    \]
    queries, and a total single-qubit measurement count of$$O\left(\log(mB\varepsilon^{-1})+\log(\delta^{-1})\right).$$
\end{theorem}

Noting that the normalized Frobenius norm 
coincides with both the Pauli $2$-norm and the normalized Schatten $2$-norm, 
we derive the following results using the monotonicity and equivalence of these various norms.

\begin{theorem}\label{thm:Ancilla_certification_Pauli}
    Suppose we are given a classical description of a traceless Hamiltonian $H_0$ and a query access $\mathcal{O}_H$ to the controlled time evolution of an unknown traceless Hamiltonian $H$ with its inverse.
    Suppose both Hamiltonians consist of at most $m$ non-zero Pauli operators and their coefficients are bounded by $B$.
    For arbitrary $\varepsilon\in\interval[open]{0}{1}$ and $\delta\in(0,1/3]$, there exists an algorithm that can distinguish with success probability at least $1-\delta$ the following two cases for any $p\in[1,\infty]$:
    \begin{itemize}
        \item \textsc{Accept}: $H=H_0$,
        \item \textsc{Reject}: $\|H-H_0\|_{\mathrm{Pauli},p}\geq\varepsilon$,
    \end{itemize}
    promised that it is in either case, where $\|\cdot\|_{\mathrm{Pauli},p}$ is the Pauli $p$-norm.
    For any $p\geq2$, the algorithm requires 
    \[
    T=
    O\left(\varepsilon^{-1}\log\delta^{-1}\right)
    \] queried evolution time, 
    \[
    O\left(m^{3/2}B\varepsilon^{-1}\log(\delta^{-1})\right)
    \]
    queries, and a total single-qubit measurement count of$$O\left(\log(mB\varepsilon^{-1})+\log(\delta^{-1})\right).$$
    For any $1\leq p<2$, the algorithm requires $T=O\left(m^{1/p-1/2}\varepsilon^{-1}\log\delta^{-1}\right)$ queried evolution time, $O\left(m^{1+1/p}B\varepsilon^{-1}\log(\delta^{-1})\right)$ queries, and a total single-qubit measurement count of$$O\left(\log(mB\varepsilon^{-1})+\log(\delta^{-1})\right).$$
\end{theorem}

\begin{theorem}\label{thm:Ancilla_certification_Schatten}
    Suppose we are given a classical description of a traceless Hamiltonian $H_0$ and a query access $\mathcal{O}_H$ to the controlled time evolution of an unknown traceless Hamiltonian $H$ with its inverse.
    Suppose both Hamiltonians consist of at most $m$ non-zero Pauli operators and their coefficients are bounded by $B$.
    For arbitrary $\varepsilon\in\interval[open]{0}{1}$ and $\delta\in(0,1/3]$, there exists an algorithm that can distinguish with success probability at least $1-\delta$ the following two cases for any $p\in[1,2]$:
    \begin{itemize}
        \item \textsc{Accept}: $H=H_0$,
        \item \textsc{Reject}: $\|H-H_0\|_{\mathrm{Schatten},p}\geq\varepsilon$,
    \end{itemize}
    promised that it is in either case, where $\|\cdot\|_{\mathrm{Schatten},p}$ is the normalized Schatten $p$-norm.
    The algorithm requires
    \[T=O\left(\varepsilon^{-1}\log\delta^{-1}\right)\] queried evolution time, 
    \[
    O\left(m^{3/2}B\varepsilon^{-1}\log(\delta^{-1})\right)
    \]
    queries, and a total single-qubit measurement count of$$O\left(\log(mB\varepsilon^{-1})+\log(\delta^{-1})\right).$$
\end{theorem}

For completeness, we give detailed proofs of the above theorems in \cref{sec:coherent-hamiltonian-certification-one-side}.

\section{Ancilla-free certification}\label{sec:Ancilla-free}
While the Hamiltonian certification protocols presented in the previous sections offer significant advantages in terms of evolution time and query complexity, they rely on ancillary qubits and controlled operations that may be challenging to implement in near-term quantum devices. 
In this section, we address these practical concerns by developing an ancilla-free certification framework that eliminates the need for additional quantum resources beyond the system being certified. 

\begin{algorithm}[t]
\caption{Stabilizer Bernoulli Sampling: $\mathsf{SBS}(U,\mc S_0)$}\label{alg:SBS}
    \SetKwInOut{Input}{input}\SetKwInOut{Output}{output}
    \Input{An $n$-qubit maximal stabilizer group $\mc S_0$ of ${\sf P}^n$, an $n$-qubit unitary $U$.}
    \Output{A Bernoulli sample: $0$ or 1.}
    Select a syndrome Pauli from the anticommutant $\theta\in\mc A_{\mc S_0}$ uniformly at random\;
    Prepare the $n$-qubit syndrome state of $\mc S_0$: $\ket{\psi(0)}=\ket{\phi_{\sigma(\theta)}}$\;
    Apply $U$ to get the final state: $\ket{\psi}\leftarrow U\ket{\psi(0)}$\;
    Measure $\ket{\psi}$ using syndrome measurement and get the syndrome $\sigma(\theta')$\;
    \eIf{$\sigma(\theta')=\sigma(\theta)$}{
        \Return $0$\;
    }{
        \Return $1$\;
    }
\end{algorithm}

\subsection{Hamiltonian stabilizer sampling}

In this part, we introduce Hamiltonian stabilizer sampling as a tool to indicate the Pauli coefficients of any unknown unitary.
This sampling makes use of the stabilizer Bernoulli sampling, which employs an arbitrary maximal stabilizer group $\mc S_0$ of the Pauli group $\mathbb{P}^n$, elements of which all belong to ${\sf P}^n$.
Since $\mc S_0\trianglelefteq \mathbb{P}^n$ is a normal subgroup, we can find the quotient group $\mc A_{\mc S_0}$, which we refer to as the \emph{anticommutant}.
According to properties of the stabilizer formalism introduced in \cref{sec:Pauli_formalism}, we first summarize the stabilizer Bernoulli sampling in Algorithm~\ref{alg:SBS}.

In essence, this algorithm prepares a uniformly random syndrome state, applies a unitary to this state, and then performs a syndrome measurement to record the syndrome shift caused by the unitary. 
This process effectively probes the Pauli information of the unitary, as demonstrated in the following lemma:

\begin{lemma}\label{lm:stabilizer_sample}
    Let $U$ be an $n$-qubit unitary,  and $\vec{u}$ be its Pauli coefficient vector.  
    For an arbitrary maximal stabilizer group $\mc S_0$ with all $2^n$ elements belonging to ${\sf P}^n$, 
    let  $Z\gets \mathsf{SBS}(U,\mc S_0)$ 
    be the output of the stabilizer
    Bernoulli sampling algorithm.
    Then, we have
    \[
        \Pr(Z=1)=\left\|\vec{u}[{\sf P}^n\backslash\mc S_0]\right\|_2^2.
    \]
\end{lemma}

Lemma~\ref{lm:stabilizer_sample} establishes the connection between the signal probability from the sampling and the Pauli coefficients of the unitary, confirming that stabilizer Bernoulli sampling serves as an effective estimator of the unitary's Pauli composition.
We delay the proof in \cref{sec:Appendix_Stab_Bernoulli}.

In our Hamiltonian certification context, we want to test the time evolution $e^{-\ii H_\mathrm{res}t}$ governed by the residual Hamiltonian $H_\mathrm{res}\coloneqq H-H_0$.
By applying this stabilizer sampling approach, we can determine whether the evolution unitary has small Pauli coefficients corresponding to specific subsets.
Nevertheless, since we only have query access to $H$ and a classical description of $H_0$, we must employ the simulation methods to construct $e^{-\ii H_\mathrm{res}t}$.
To this end, we utilize the second-order Trotter-Suzuki to approximate the time evolution.
This sampling approach, as demonstrated in subsequent lemmas, yields a signal probability that effectively captures the information of the residual Hamiltonian, even accounting for errors introduced by the Trotter-Suzuki approximation.

\begin{algorithm}[t]
\caption{Hamiltonian Stabilizer Sampling: $\mathsf{HSS}(\mc S_0,H_0,\mc O_H,t,r)$}\label{alg:HSS}
    \SetKwInOut{Input}{input}\SetKwInOut{Output}{output}
    \Input{An $n$-qubit maximal stabilizer group $\mc S_0$ of ${\sf P}^n$, $n$-qubit Hamiltonian $H_0$, query access $\mc O_H$ to the black box $e^{-\ii Ht}$, evolving time $t$, Trotter-step number $r$.}
    \Output{A Bernoulli sample: $Z=0$ or $1$.}
    Let $U$ denote $\left(e^{\ii H_0t/(2r)}\mc O_H(t/r)e^{\ii H_0t/(2r)}\right)^r$\tcp*{Second-order Trotter formula}
    \Return $\mathsf{SBS}(U,\mc S_0)$\;
\end{algorithm}

\begin{lemma}\label{lm:noisy_bound_stab}
Suppose we are given a classical description of a traceless Hamiltonian $H_0 $, and query access $\mathcal{O}_H$ to the time evolution of an unknown traceless Hamiltonian $H$.
Let $H_\mathrm{res}\coloneqq H-H_0=\sum_{\alpha\in{\sf P}^n}s_\alpha P_\alpha$, and suppose its
Pauli coefficient vector $\vec{s}$ consists of at most $m_0$ non-zero coordinates.
Denote the set $\mc S$ the support of $\vec{s}$ and $S\coloneqq\|\vec{s}\|_\infty$.
For an arbitrary maximal stabilizer group $\mc S_0$ with all $2^n$ elements belonging to ${\sf P}^n$, $t\in[0, 1/(m_0S)]$ and integer $r\geq1$,
let $Z\gets\mathsf{HSS}(\mc S_0,H_0,Q_H,t,r)$  be the output of the Hamiltonian Stabilizer sampling algorithm (Alg.~\ref{alg:HSS}),
and denote $R\coloneqq\sum_{k=2}^\infty\frac{m_0^{k-1}(St)^k}{k!}$, $\mc X\coloneqq{\sf P}^n\backslash\mc S_0$.
We call $\Pr(Z=1)$ the signal probability.
Then, we have
\begin{itemize}
    \item The Trotter error in diamond norm can be bounded as
    \[
    \varepsilon_{\mathrm{Trotter},\diamond}\coloneqq \|\mc U_\mathrm{res}(t)-\mc U_{t,r}\|_\diamond =O\left(\frac{Sm_0\|H_0\|(\|H_0\|+Sm_0)t^3}{r^2}\right),
    \]
    where $\mc U_\mathrm{res}(t)$ is the corresponding quantum channel of $e^{\ii H_{\mathrm{res}}t}$, 
    $\mc U_{t,r}$ is the corresponding quantum channel of $\left(e^{\ii H_0t/(2r)}\mc O_H(t/r)e^{\ii H_0t/(2r)}\right)^r$.
    \item The signal probability has a  lower bound
    \[
      \Pr(Z=1)\geq \max(\|\Vec{s}[\mc S\cap\mc X]\|_2t-\sqrt{m_0}R,0)^2-\epsilon_\mathrm{Trotter,\diamond}.
    \]
    \item The signal probability
    has an upper bound
    \[
      \Pr(Z=1)\leq \sum_{\alpha\in \mathcal{S}\cap\mc X}(\abs{s_\alpha} t+R)^2+(m_0-\abs{\mc S\cap\mc X})R^2+\epsilon_\mathrm{Trotter,\diamond}.
    \]
\end{itemize}

\end{lemma}
\noindent 
This relationship between the signal probability and Pauli coefficients forms the basis for our ancilla-free certification algorithm.
We delay the proof in \cref{sec:Appendix_Stab_Bernoulli}.

\subsection{Dual-Stabilizer Hamiltonian certification}

In this part, we adopt the Hamiltonian stabilizer sampling to estimate the Pauli information and complete the certification.
However, a key limitation emerges with stabilizer Bernoulli sampling: we can only probe the Pauli coefficients of the residual Hamiltonian within a specific subset, ${\sf P}^n\backslash\mc S_0$, as established in \cref{lm:noisy_bound_stab}.
To overcome this limitation and achieve comprehensive certification across the entire set of Pauli coefficients, we employ two groups, $\langle I,Z\rangle^{\otimes n}$ and $\langle I,X\rangle^{\otimes n}$, that intersect with each other only at the identity element $I_{2^n}$.
By alternating between these complementary stabilizer groups, our detection protocol effectively covers all possible Pauli coefficients, ensuring complete verification of the residual Hamiltonian.
We summarize this idea in Algorithm~\ref{alg:SHC}.

\begin{algorithm}[t]
\caption{Stabilizer Hamiltonian Certification: $\mathsf{SHC}(H_0,Q_H,M,S,\varepsilon,\delta)$}\label{alg:SHC}
    \SetKwInOut{Input}{Input}\SetKwInOut{Output}{Output}\SetKwInOut{Para}{Parameters}
    \Input{Target Hamiltonian $H_0$, query $\mc O_H$ to the black box $e^{-\ii Ht}$, Pauli term number upper bound $m$, Pauli coefficient upper bound $B$, precision $\varepsilon$, failure rate $\delta$.}
    \Output{\textsc{Accept} when the two Hamiltonians are close; \textsc{Reject} otherwise.}
    $m\leftarrow 2m_0$, $k\leftarrow\lceil\log(B/\varepsilon)\rceil$, $\delta'\leftarrow \delta/(k+1)$\;
    $c_1\gets 0.02$, $c_2\gets 0.025$, $c_3\gets 0.09$, $c_4\gets 0.12$\tcp*{See Propositions~\ref{prop:iter_v_stab} and~\ref{prop:fina_v_stab}}
    \For{$\mathrm{round}\leftarrow1,2$}{
    \eIf{$\mathrm{round}=1$}{$\mc S_0\leftarrow\langle I,Z\rangle^{\otimes n}$\;}{$\mc S_0\leftarrow\langle I,X\rangle^{\otimes n}$\;}
    \For{$j\leftarrow k,k-1,\cdots,1$}{
        $b\leftarrow2^{j+1}\varepsilon$, $t\leftarrow {m_0^{-3/2}b^{-1}}/2$, $r\leftarrow\Theta\big(m_0^{3/4}Bb^{-1}\big)$, $\delta'\leftarrow \delta/(k+1)$\;
        Let $U_{t,r}$ denote $\mathsf{HSS}(\mc S_0,H_0,\mc O_H,t,r)$\;
        $l\leftarrow \textsf{BernoulliTest}\left(U_{t,r}, c_2m_0^{-3}, c_1 m_0^{-3},\delta'\right)$\tcp*{See \cref{prop:Bernoulli-distribution-parameter-testing}}
        \If{$l=\textsc{Large}$}{
            \Return \textsc{Reject}\tcp*{The $(k-j+1)$th check}
        }
    }
    $t\leftarrow{m_0^{-3/2}\varepsilon^{-1}}/4$, $r\leftarrow\Theta\big(m_0^{3/4}B\varepsilon^{-1}\big)$\;
    Let $U_{t,r}$ denote $\mathsf{HSS}(\mc S_0,H_0,\mc O_H,t,r)$\;
    $l\leftarrow \textsf{BernoulliTest}\left(U_{t,r}, c_4 m_0^{-3}, c_3m_0^{-3},\delta'\right)$\tcp*{See \cref{prop:Bernoulli-distribution-parameter-testing}}
        \If{$l=\textsc{Large}$}{
            \Return \textsc{Reject}\tcp*{The final check}
        }
    }
    \Return \textsc{Accept}\;
\end{algorithm}

We show the following propositions to specify the choices of parameters in different estimations across our certification, which maintain the gaps between different Hamiltonian cases.

\begin{proposition}\label{prop:iter_v_stab}
    Let Hamiltonians $H$, $H_0$ and $H_{\mathrm{res}}$, oracle $\mathcal{O}_H$, vector $\vec{s}$,
    number of terms $m_0$, set $\mathcal{S}$, 
    and $\ell_{\infty}$-norm $S$
    be the same as those defined in \cref{lm:noisy_bound_stab}.
    Then, 
    for any $\varepsilon>0$,  $b\geq4\varepsilon$, and a maximal stabilizer group $\mc S_0$ with all $2^n$ elements belonging to ${\sf P}^n$, setting $t={m_0^{-3/2}b^{-1}}/2$ and $r=\Omega(m_0^{1/4}\|H_0\|^{1/2}b^{-1/2}+m_0^{-1/4}\|H_0\|b^{-1})$,
    the output $Z\gets \mathsf{HSS}(\mc S_0, H_0,\mc O_H,t,r)$ of the algorithm satisfies
    \begin{itemize}
        \item if $\|\Vec{s}\|_2\leq\varepsilon$, 
        then
        $\Pr(Z=1)\leq 0.02 m_0^{-3}$;
        \item if $b/2<S\leq b$ and $\|\vec{s}[{\sf P}^n\backslash\mc S_0]\|_2\geq\|\vec{s}[\mc S_0]\|_2$,
        then 
        $\Pr(Z=1)\geq 0.025 m_0^{-3}$.
    \end{itemize}
\end{proposition}

\noindent
This result, whose detailed proof appears in \cref{sec:Appendix_SHC}, establishes the mathematical basis for our binary check procedure. The identified gap in signal probabilities enables us to apply the \textsf{BernoulliTest} procedure to effectively distinguish between the different cases with high confidence.

Another complementary proposition establishes the gap of the remaining cases, particularly those where individual coefficients are small but their collective magnitude exceeds our $\varepsilon$ threshold.

\begin{proposition}\label{prop:fina_v_stab}
    Let Hamiltonians $H$, $H_0$ and $H_{\mathrm{res}}$, oracle $\mathcal{O}_H$, vector $\vec{s}$,
    number of terms $m_0$, set $\mathcal{S}$, 
    and $\ell_{\infty}$-norm $S$,
    be the same as those defined in \cref{lm:noisy_bound_stab}.
    Then, 
    for any $\varepsilon>0$ and a maximal stabilizer group $\mc S_0$ with all $2^n$ elements belonging to ${\sf P}^n$, setting $t={m_0^{-3/2}\varepsilon^{-1}}/4$ and $r=\Omega(m_0^{1/4}\|H_0\|^{1/2}\varepsilon^{-1/2}+m_0^{-1/4}\|H_0\|\varepsilon^{-1})$,
    the output $Z\gets \mathsf{HSS}(\mc S_0, H_0,\mc O_H,t,r)$ of the algorithm satisfies
    \begin{itemize}
        \item if $\|\Vec{s}\|_2\leq\varepsilon$, 
        then
        $\Pr(Z=1)\leq 0.09 m_0^{-3}$;
        \item if $S\le 2\varepsilon$ and $\|\vec{s}[{\sf P}^n\backslash\mc S_0]\|_2\ge 2\varepsilon$, 
        then 
        $\Pr(Z=1)\geq 0.12 m_0^{-3}$.
    \end{itemize}
\end{proposition}

\noindent  We delay the proof in \cref{sec:Appendix_SHC}.

In both rounds of the algorithm, we can go through all cases and determine between different $H$ using gaps specified in \cref{prop:iter_v_stab,prop:fina_v_stab}.
Combining these two rounds, the algorithm serves as an ancilla-free certification protocol, as shown in the following theorem.
\begin{theorem}\label{thm:Stabilizer}
    Suppose we are given a classical description of a traceless Hamiltonian $H_0$, and a query access $\mathcal{O}_H$ to the \emph{positive} time evolution of an unknown traceless Hamiltonian $H$.
    Suppose both Hamiltonians consist of at most $m$ non-zero Pauli operators and their coefficients are bounded by $B$.
    For any $\varepsilon\in\interval[open]{0}{1/4}$ and $\delta\in(0,1/3]$, running $\mathsf{SHC}(H_0,\mc O_H,m,B,\varepsilon,\delta)$ (Alg.~\ref{alg:SHC}) can distinguish with probability at least $1-\delta$ the following two cases:
    \begin{itemize}
        \item \textsc{Accept}: $\|H-H_0\|_F\leq\varepsilon$,
        \item \textsc{Reject}: $\|H-H_0\|_F\geq4\varepsilon$.
    \end{itemize}
    The algorithm requires 
    \[T=O\left(m^{3/2}\varepsilon^{-1}\left(\log\log(B\varepsilon^{-1})+\log\delta^{-1}\right)\right)
    \]
    queried evolution time,
    \[O\left(m^{15/4}B\varepsilon^{-1}\left(\log\log(B\varepsilon^{-1})+\log(\delta^{-1})\right)\right)\] 
    queries to $\mc O_H$,
    and a total measurement count of\[O\left(m^{3}\log(B\varepsilon^{-1})\left(\log\log(B\varepsilon^{-1})+\log(\delta^{-1})\right)\right).\]
\end{theorem}

\section{Lower bounds and hardness results}
\subsection{Lower bounds on queried evolution time}
To complete our theoretical characterization of Hamiltonian certification, we now prove matching lower bounds that establish the fundamental limits of what any certification algorithm can achieve.
To this end, we reduce the hypothesis testing problem to our certification test, as also adopted by~\cite{huang2023learning}.
We delay proofs of following theorems in \cref{sec:append_proof_lower_bounds}.
\begin{theorem}\label{thm:lower_Pauli}
    Let $H$ (unknown) and $H_0$ (known) be two traceless $n$-qubit Hamiltonians with their Pauli coefficients bounded by a constant.
    Given access to the unknown $H$ through its (controlled) time evolution as well as its inverse, for $0 \leq \varepsilon_1 < \varepsilon_2 \leq 1$ and $p\in[1,\infty]$, any algorithm requires $\Omega(\frac{1}{\varepsilon_2-\varepsilon_1})$ queried evolution time to distinguish
    with success probability at least $2/3$ the two cases:
    \begin{itemize}
        \item \textsc{Accept}: $\norm{H-H_0}_{\mathrm{Pauli},p} \le \varepsilon_1$,
        \item \textsc{Reject}: $\norm{H-H_0}_{\mathrm{Pauli},p} \ge \varepsilon_2$,
    \end{itemize}
    promised that it is in either case, where $\|\cdot\|_{\mathrm{Pauli},p}$ is the Pauli $p$-norm.
    Particularly, given $p\in[1,2)$ and that both Hamiltonians consist of at most $m$ non-zero Pauli operators with $m\leq2n+1$, any algorithm requires $\Omega(\frac{m^{1/p-1/2}}{\varepsilon_2-\varepsilon_1})$ queried evolution time to distinguish these two cases.
\end{theorem}

\begin{remark}
    \rm In fact, we can relax the requirement on $m$ to  any $m=O(n)$ case.
    This is due to the fact that we can find polynomially many anti-commuting sets with cardinalities of $n$.
\end{remark}

For completeness, we also establish similar bounds for normalized Schatten norms:

\begin{theorem}\label{thm:lower_Schatten}
    Let $H$ (unknown) and $H_0$ (known) be two traceless Hamiltonians with their Pauli coefficients bounded by a constant.
    Given access to the unknown $H$ through its (controlled) time evolution as well as its inverse, for $0 \leq \varepsilon_1 < \varepsilon_2 \leq 1$ and $p\in[1,\infty]$, any algorithm requires $\Omega(\frac{1}{\varepsilon_2-\varepsilon_1})$ queried evolution time to distinguish
    with success probability at least $2/3$ the two cases:
    \begin{itemize}
        \item \textsc{Accept}: $\norm{H-H_0}_{\mathrm{Schatten},p} \le \varepsilon_1$,
        \item \textsc{Reject}: $\norm{H-H_0}_{\mathrm{Schatten},p} \ge \varepsilon_2$,
    \end{itemize}
    promised that it is in either case, where $\norm{\cdot}_{\mathrm{Schatten},p}$ is the normalized Schatten $p$-norm.
\end{theorem}

Taken together, these lower bounds complete our theoretical characterization of Hamiltonian certification complexity across different norm families. 
They confirm that our algorithms achieve optimal scaling for both robust and one-sided certification.

\subsection{\textsf{coQMA}-hardness with respect to the operator norm}

In this part,
we show the $\mathsf{co}\QMA$-hardness
of the robust $k$-local Hamiltonian
certification problem
with respect to the operator norm,
which could be seen as a 
worst-case hardness result for
the Hamiltonian certification problem.
We first define the 
robust $k$-local Hamiltonian
certification problem as follows.

\begin{problem}
    For positive integer $n$ and $k$,
    positive reals $a$, $b$
    satisfying $0\le  a < b\le 1$,
     the robust $k$-local Hamiltonian
    certification problem $(n,a,b)$
    with respect to the operator norm
    is the decision problem 
    defined as follows.
    Let $H$ and $H_0$
    be two \emph{traceless}
    $k$-local $n$-qubit
    Hamiltonians with
    their classical descriptions 
    given, i.e., $H = \sum_j H^{(j)}$
    with $H^{(j)}$ acts non-trivially 
    on at most $k$ qubits and 
    $\norm{H^{(j)}}\le \textup{poly}(n)$ and similarly for $H_0$.
    Determine whether:
    \begin{itemize}
        \item Yes case: $\norm{H-H_0}\le a$,
        \item No case: $\norm{H-H_0}\ge b$,
    \end{itemize}
    promised that one of these
    to be the case, where $\norm{\cdot}$
    stands for the operator norm
    of matrices.
\end{problem}

Note that we require the Hamiltonians
to be traceless here in accordance
with previous settings. 
Since for $k$-local Hamiltonians,
one can easily implement its time
evolution,
any algorithm for \cref{prob:HC}
with respect to the operator norm
could be used to solve the
above problem.
Thus, the following $\mathsf{co}\QMA$-hardness
result implies there does not 
exist
algorithms for the robust Hamiltonian certification problem
with respect to the operator norm
if $\BQP\ne \QMA$.

\begin{theorem}\label{thm:QMA-hard-HC-operator}
    The robust $k$-local Hamiltonian
    certification problem $(n,a,b)$
    with respect to the operator norm
    is $\mathsf{co}\QMA$-hard
    for $0\le a < b \le 1$
    with $b-a\ge 1/\textup{poly}(n)$
    even when $k = 3$.
\end{theorem}

The above theorem is proved
via a direct reduction to 
the $k$-local Hamiltonian problem
studied in~\cite{KSV02,NM07,KKR06}. 
For completeness, we give a detailed
proof of the theorem in \cref{sec:appendix-hardness}

\bibliographystyle{alphaurl}
\bibliography{ref.bib,HL}

\appendix

\section{Deferred proofs in Preliminaries}
\subsection{Stabilizer formalism}
In this part, we focus on the remaining proofs of results regarding the stabilizer groups.
\begin{proof}[Proof of Proposition~\ref{prop:Anticom_Isomorphic}]
The first statement is directly from Proposition 3.8 in~\cite{gottesman2016surviving}, so we only need to prove the isomorphism between the anticommutant $\mc A_{\mc S_0}$ and $\mathbb{Z}_2^n$.

Firstly, we can show the syndrome map $\sigma_{\mc S_0}$ is a bi-jection.
According to the first statement, $\sigma_{\mc S_0}$ is an injection.
Moreover, Proposition 3.16 in~\cite{gottesman2016surviving} shows that for any $n$-qubit string $b\in\mathbb{Z}_2^n$, there exists a Pauli $P$ such that $\sigma_{\mc S_0}(P)=b$, implying the surjection.

We then verify the homomorphism of $\sigma_{\mc S_0}$.
For any $P_\beta,P_\gamma\in\mc A_{\mc S_0}$, each coordinate $i$ satsifies
\begin{align}
    \sigma_{\mc S_0}(P_\beta P_\gamma)_i=\sigma_{\mc S_0}(\beta+\gamma)_i=\pip{\beta+\gamma}{\alpha_i}\equiv\pip{\beta}{\alpha_i}+\pip{\gamma}{\alpha_i}\ (\text{mod } 2).\notag
\end{align}
This completes the proof.    
\end{proof}

Another remaining proof is for Proposition~\ref{prop:Unique_Decomp}, which states the existence of the unique decomposition of an arbitrary Pauli operator in $\mathbb{P}^n$ by a phase, and two Pauli operators in the maximal stabilizer group and its anticommutant.
\begin{proof}[Proof of Proposition~\ref{prop:Unique_Decomp}]
By definition of the anticommutant, we can find a unique $P_\beta\in\mc A_{\mc S_0}$ such that
\begin{gather}
    P_\alpha\in P_\beta(\langle\ii I\rangle\times\mc S_0).\notag
\end{gather}
Within the coset, we can uniquely find $\gamma=(\alpha-\beta)\in\mc S_0$ and a phase factor $v_{\beta,\gamma}\in\langle\ii\rangle$, which satisfies
\begin{gather}
    P_\alpha=v_{\beta,\gamma}P_\beta P_\gamma.\notag
\end{gather}
\end{proof}

\subsection{Discriminating and testing Bernoulli distributions}\label{sec:Appendix_Bernoulli_Test}

To begin the proof, we first illustrate a fact about the properties of the 
Hellinger distance.
Typically, this distance is closely related to the total
variation distance and can be easily calculated especially when we consider the joint distribution generated by the tensor of the original.

\begin{fact}\label{fact:TV-Hellinger}
For discrete probability distributions $\mc P, \mc Q$, and for any
positive integer $m \in \mathbb{N}_{+}$,
it holds that
\begin{align}
d_H^2(\mc P, \mc Q) &\leq d_{TV}(\mc P, \mc Q) \leq \sqrt{2}d_H(\mc P, \mc Q),\notag
\\ d_H^2({\mc P}^{\otimes m}, {\mc Q}^{\otimes m}) &= 1 - (1 - d_H^2(\mc P, \mc Q))^m \le m d_H^2(\mc P, \mc Q)).\notag
\end{align}
\end{fact}

Moreover, we recall a property of the binomial
distribution.
\begin{fact}
\label{fact:cdf-binomial-distribution}
Let $X\sim B(n,p)$ be a random variable following a binomial distribution with
parameter $n\in\mathbb{N}$ and $p\in \interval{0}{1}$. Then, the cumulative
distribution function of $X$ can be expressed as
\[
 \Prob{X\le k} =  F(k;n,p) = \sum_{j = 0}^k \binom{n}{j} p^j(1-p)^{n-j}
  = (n-k)\binom{n}{k}\int_0^{1-p} t^{n-k-1} (1-t)^k \textup{d}t.
\]
Moreover, the cumulative distribution function of $X$ is monotonically non-increasing
with respect to $p$.
\end{fact}

Here, we continue the proof of Proposition~\ref{prop:Bernoulli-distribution-parameter-testing}.
To this end, we first explicitly construct the stated \textsf{BernoulliTest} algorithm as follows.
\begin{algorithm}[t]
\caption{Bernoulli Distribution Parameter Testing $\mathsf{BernoulliTest}(S,a,b,\delta)$}\label{alg:bernoullidisc}
    \SetKwInOut{Input}{Input}
    \SetKwInOut{Output}{Output}
    \Input{Sample access $S$
    to a Bernoulli distribution $\mathcal{P}$
    with parameter $p$, $a, b\in \interval{0}{1}$ with $a > b$,
    failure probability $\delta$.}
    \Output{Decide $p\ge a$  or $p \le b$.}

    $m\gets 4/{(\sqrt{a}-\sqrt{b})}^2$\;
    \For{$t = 1, \dots, 2\log(1/\delta)$}{
    \For{$i_t=1,\ldots,36$}{
        Use $S$ to get $m$ samples
        $j_{t,i,1}, j_{t,i,2}, \dots, j_{t,i,m}$
        where $j_{t,i,k}\in \{0, 1\}$\;
        \eIf{$\Pi_{k=1}^m (j_{t,i,k}a+(1-j_{t,i,k})(1-a)) \ge \Pi_{k=1}^m (j_{t,i,k}b+(1-j_{t,i,k})(1-b))$}{
        $cnt_{t,i} \gets 1$;
    }{
        $cnt_{t,i} \gets 0$;
    }
    }
    $\mu_t = \sum_{i=1}^{36} cnt_{t,i}/36$\;}
    \eIf{$\sum_t \mu_t/2\log(1/\delta) > 1/2$}{
        \Return \textsc{Large};
    }{
        \Return \textsc{Small};
    }
\end{algorithm}

\begin{proof}[Proof of Proposition~\ref{prop:Bernoulli-distribution-parameter-testing}]
    Let $\mathcal{P}_0$ and $\mathcal{Q}_0$
    denote the Bernoulli distributions
    with parameter $a$ and $b$, respectively.
    By direct computation, we have
    \[
        d_{H}^2(\mathcal{P}_0, \mathcal{Q}_0)  = \frac{(\sqrt{a}-\sqrt{b})^2}{2}.
    \]
    Then, setting $\ell = 4/(\sqrt{a}-\sqrt{b})^2$,
    by \cref{fact:TV-Hellinger}, we have
    \[
    d_{TV}(\mathcal{P}_0^{\otimes \ell}, \mathcal{Q}_0^{\otimes \ell}) \ge 1 - \exp (-\ell
    d_{H}^2(\mathcal{P}_0, \mathcal{Q}_0)) \ge \frac{3}{4}.
    \]
    Therefore,
    for $A \coloneqq \{(j_1, j_2, \dots, j_{\ell})\mid \mathcal{P}_0^{\otimes
      \ell}(j_1, j_2, \dots, j_m) \ge \mathcal{Q}_0^{\otimes \ell}(j_1, j_2, \dots,
    j_{\ell})\}$,
    by the operational meaning of the total variation distance,
    we have
    \[
    \mathcal{P}_0^{\otimes \ell} (A) - \mathcal{Q}_0^{\otimes \ell} (A) =
    d_{TV}(\mathcal{P}_0^{\otimes \ell}, \mathcal{Q}_0^{\otimes \ell}) \ge 3/4.
    \]
    Since $\mathcal{P}_0^{\otimes \ell} (A), \mathcal{Q}_0^{\otimes \ell} (A) \in \interval{0}{1}$,
    we have $0\le \mathcal{Q}_0^{\otimes m} (A) \le 1/4$ and
    $3/4 \le \mathcal{P}_0^{\otimes m} (A) \le 1$.

    Moreover, by the definition of $A$, we know $(j_1, j_2, \dots, j_m)\in A$ if and
    only if
    \[
    \prod_k \mathcal{P}_0(j_k) = a^{\sum_k j_k}{(1-a)}^{\ell - \sum_k j_k} \ge \prod_k \mathcal{Q}_0(j_k) = b^{\sum_k j_k}{(1-b)}^{\ell-\sum_k j_k},
    \]
    which means $\sum_k j_k\ge t$ for some $t$.
    Thus, we can write
    \[
    \mathcal{P}_0^{\otimes \ell}(A) = \sum_{(j_{1}, j_2, \dots, j_m)\in A} \mathcal{P}(j_1, j_2, \dots, j_m)
    = F(\floor{t};\ell,a),
    \]
    where $F(\floor{t};\ell,a)$ is the cumulative distribution function of a
    binomial distribution $B(m, a)$ discussed
    in Fact~\ref{fact:cdf-binomial-distribution}.
    Similarly,
    $\mathcal{Q}_0^{\otimes \ell}(A) = F(\floor{t};\ell, b)$.
    Therefore, by the monotonicity (Fact~\ref{fact:cdf-binomial-distribution})
    of the binomial cumulative distribution function,
    we have
    $\mathcal{P}^{\otimes \ell}(A) \ge \mathcal{P}_{0}^{\otimes \ell}(A) \ge 3/4$.
    if $p \ge a$
    and $\mathcal{P}^{\otimes \ell}(A) \le \mathcal{Q}_{0}^{\otimes \ell}(A) \le 1/4$
    if $p \le b$.

    Now, given sample access $S$,
    we can collect $\ell$ samples $j_1,j_2, \dots,
    j_{\ell}$. 
    Then, we can decide whether $(j_1, j_2, \dots, j_{\ell}) \in A$ by directly
    computing $\prod_k \mathcal{P}_0(j_k) = a^{\sum_k j_k}{(1-a)}^{\ell-\sum_k j_k}$
    and $\prod_k \mathcal{Q}_0(j_k) = b^{\sum_k j_k} {(1-b)}^{\ell - \sum_k j_k}$ and making a
    comparison in $O(\ell)$ time. Then, by Hoeffding's inequality, with $36$ samples we can estimate the value
    of $\Pr(X\in A)= \E[1_A] $ within $1/6$ additive error, with success
    probability at least $2/3$. We denote this estimate value as $\mu$. If $p\ge a$, then with high
    probability, $1/2< 3/4-1/6\le \mu \le 1$. Otherwise, with high probability, $0\le \mu \le 1/4 + 1/6 <
    1/2$. Therefore, let the algorithm output \textsc{Large} if $\mu>1/2$, and
    \textsc{Small} otherwise.
    The claim then follows by adopting the
    majority vote technique and
    applying the Hoeffding's inequality.
\end{proof}

\section{Deferred theorems and proofs of Pauli coefficient analysis}
\label{sec:pauli-coeffcient-analysis-proof}

\begin{proof}[Proof of \cref{lemma:evolution-Pauli-efficient-bound}]
Using Taylor's expansion, we have
\[
e^{-\ii Ht}  = I - \ii t\sum_{\alpha\in{\mc S}}s_\alpha P_\alpha +\sum_{k=2}^\infty \frac{(-\ii t)^k(\sum_{\alpha\in{\mc S}}s_\alpha P_\alpha)^k}{k!}.
\]
Then, we know
\[
A_k = \rbra*{\sum_{s\in \mathcal{S}} s_{\alpha} P_{\alpha}}^k
= \sum_{(j_1, j_2, \dots, j_k)\in \mathcal{S}^k} s_{j_1}s_{j_2}\cdots s_{j_k} P_{j_1}P_{j_2}\cdots P_{j_k}.
\]
Now, for a fixed $P_{\ell}$, consider
\[
  \mathcal{S}_{k,\ell} = \{(j_1,j_2, \dots, j_k) \in \mathcal{S}^k
  \mid \exists a\in \{0,1,2,3\},P_{j_1}P_{j_2}\cdots P_{j_k} = \ii^a P_{\ell}\}.
\]
Note that if $P_{j_1},P_{j_2},\ldots, P_{j_{k-1}}$ are given, then there is at most
$1$ solution for $P_{j_k}$ to satisfy the constraint. Therefore, we conclude
$\abs{\mathcal{S}_{\ell,k}}\le m^{k-1}$, and the number of non-empty $\mathcal{S}_{\ell,k}$ for a fixed
$k$ is at least $m$, giving
\[
    \abs{a_{\ell,k}}= \abs{\sum_{(j_1,\dots, j_k)\in \mathcal{S}_{k,\ell}} s_{j_1}s_{j_2}\cdots s_{j_k}}
    \le \abs{\mathcal{S}_{\ell,k}} S^k \le m^{k-1}S^k.
\]
For the other statement, note that $\{S_{k,\ell}\}_{\ell\in \mathbb{P}^n}$ forms a
partition of $\mathcal{S}^k$, meaning that $\sum_{\ell} \abs{S_{\ell, k}} = m^k$. We then
have
\[
      \sum_{\ell} \abs{a_{\ell,k}}
    \le \sum_{\ell} \abs{\mathcal{S}_{\ell,k}} s^k \le m^k S^k. \qedhere
    \]
\end{proof}

\begin{proposition}\label{prop:opt-high-order-terms}
    Let $N,m,r$ be positive integers with $m\le N$.
    Suppose $\vec{b}$ is an $N$-dimensional
    vector with non-negative coordinates
    $b_1, b_2, \dots, b_N$, and $\vec{c}$ is an $r$-dimensional
    vector with non-negative coordinates
    $c_1, c_2, \dots, c_r$.
    Let $A$ be an $r\times N$ matrix 
    whose coordinates $a_{k,j}$
    are integer variables satisfying 
    $0\le a_{k,j} \le m^{k-1}$  for all $j\in [N]$, and $\sum_j a_{k,j} \le m^k$.
    Then, the function
    \[
    f(A)
    = \sum_{j=1}^N \rbra*{b_j+ \sum_{k=1}^r c_k a_{k,j}}^2
    = \Abs{\vec{b}+A^{\intercal} \vec{c}}^2
    \]
    attains its maximum when
    $a_{k,j} = m^{k-1}$ if $j\in S$,
    and $a_{k,j} = 0$ otherwise,
    where
    \[
    S \in {\arg\max}\cbra{
    \sum_{j\in T} b_j\mid T\subseteq [N], \abs{T}=m}.
    \]
\end{proposition}

\begin{proof}
    We expand $f$ as
    \[
    \begin{aligned}
        f(A)
    &= \sum_{j=1}^N \rbra*{b_j^2  + b_j \sum_{k=1}^r c_k a_{k,j} + \rbra*{\sum_{k=1}^r c_k a_{k,j}}^2} \\
    &= \sum_{j=1}^N b_j^2
    + \sum_{k=1}^r c_k \sum_{j=1}^N b_j a_{k,j}
    + \sum_{k=1}^r c_k^2 \sum_{j=1}^N a_{j,k}^2
    + 2\sum_{1\le k_1 < k_2\le r} c_{k_1} c_{k_2}\sum_{j=1}^N a_{k_1,j} a_{k_2,j}.
    \end{aligned}
    \]
    Noting that $0\le a_{k,j} \le m^{k-1}$,
    we have
    $a_{k,j}^2\le m^{k-1} a_{k,j}$.
    This gives
    \[
    \sum_{j=1}^N a_{k,j}^2\le m^{k-1}  \sum_{j=1}^N a_{k,j} = m^{2k-1},
    \]
    and the equality holds if
    there are $m$ $a_{k,j}$'s being $m^{k-1}$
    while others being $0$.
    Furthermore,
    AM-GM inequality implies
    \[
    \begin{aligned}
        a_{k_1,j} a_{k_2,j}  &=
    m^{-k_1-k_2} (m^{k_2} a_{k_1,j})
    (m^{k_1} a_{k_2,j}) \\
    & \le  m^{-k_1-k_2} \frac{(m^{2k_2} a_{k_1,j}^2 + m^{2k_1} a_{k_2,j}^2)}{2} \\
    &= \frac{m^{k_2-k_1} a_{k_1,j}^2 + m^{k_1-k_2} a_{k_2,j}^2}{2}.
    \end{aligned}
    \]
    We have
    \[
    \begin{aligned}
        2\sum_{1\le k_1 < k_2\le r} c_{k_1} c_{k_2} \sum_{j=1}^N a_{k_1,j} a_{k_2,j} 
        &\le
    \sum_{1\le k_1 < k_2\le r} c_{k_1} c_{k_2}  \sum_{j=1}^N
    \rbra[\big]{m^{k_2-k_1} a_{k_1,j}^2 +  m^{k_1-k_2}a_{k_2,j}^2} \\
    & \le 2\sum_{1\le k_1 < k_2\le r} c_{k_1} c_{k_2}m^{k_1+k_2-1}.
    \end{aligned}
    \]
    The first equality holds if $m^{k_2} a_{k_1,j} = m^{k_1} a_{k_2,j}$.
    The second equality holds if
    there are $m$ $a_{k_1,j}$'s being $m^{k_1-1}$,
    $m$ $a_{k_2,j}$'s being $m^{k_2-1}$,
    and others being $0$.

    The term
    $\sum_{j=1}^N b_j a_{k,j}$
    is upper bounded by
    \[m^{k-1} \max_{\substack{T\subseteq [N],\\ \abs{T}=m}} \sum_{j\in T} b_j,
    \]
    via a standard greedy argument,
    with the upper bound attained
    when
    $a_{k,j} = m^{k-1}$ if $j\in S$,
    and $a_{k,j} = 0$ otherwise,
    where
    \[
    S \in {\arg\max}\cbra{
    \sum_{j\in T} b_j\mid T\subseteq [N], \abs{T}=m}.
    \]

    Combining all above,
    we know that
    each term in the expansion
    can be upper bounded,
    and all the upper bounds
    can be attained
    if we take
        $a_{k,j} = m^{k-1}$ if $j\in S$,
     $a_{k,j} = 0$ otherwise.
    This gives the desired result.
\end{proof}

\begin{proof}[Proof of \cref{lm:Pauli_Lower_Bound2}]
  By Taylor's expansion, we can write
  \[
    e^{-\ii Ht} = I - \ii H t + \sum_{k=2}^{\infty} \frac{(-\ii)^kA_k}{k!}
    t^k,
  \]
  where
  $A_k = \rbra*{\sum_{\alpha\in \mathcal{S}} s_{\alpha} P_{\alpha}}^k$.
  Then we have, for any $\beta\in \mathsf{P}^n\setminus \{I\}$,
  \begin{equation*}
    v_{\beta} = \ii t s_{\beta} + \sum_{k=2}^{\infty} \frac{(-\ii)^k\tr(P_{\beta} A_k)}{2^n k!} t^k.
  \end{equation*}
  Let $\vec{w}$ denote the $4^n$-dimensional vector labelled by the $n$-qubit
  Pauli group with coordinate
  \[
    w_{\beta} = \sum_{k=2}^{\infty} \frac{(-\ii)^k\tr(P_{\beta} A_k) t^k}{2^n k!},
  \]
  We can write $\vec{v} = \ii t \vec{s} + \vec{w}$ by definition.

  Now, for $w$, we have
  \[
    \begin{aligned}
      \Abs{\vec{w}[\mathcal{X}]}_2^2
      &= \sum_{\beta\in \mathcal{X}} \abs{w_{\beta}}^2 \\
      &= \sum_{\beta\in \mathcal{X}} \abs{\sum_{k=2}^{\infty} \frac{(-\ii)^k\tr(P_{\beta} A_k) t^k}{2^nk!}}^2 \\
      &\le \sum_{\beta\in \mathcal{X}} \rbra*{\sum_{k=2}^{\infty} \frac{\abs{\tr(P_{\beta} A_k)} t^k}{2^n k!}}^2 \\
      &\le \sum_{\beta\in \mathcal{X}} \rbra*{\sum_{k=2}^{\infty} \frac{\abs{\mathcal{S}_{k,\beta}}S^k t^k}{k!} }^2,
    \end{aligned}
  \]
  where
  $\mathcal{S}_{k,\beta} = \{(j_1,j_2, \dots, j_k) \in \mathcal{S}^k \mid \exists a\in \{0,1,2,3\},P_{j_1}P_{j_2}\cdots P_{j_k} = \ii^a P_{\beta}\}$
  with $\abs{\mathcal{S}_{k,\beta}}\le m^{k-1}$ and $\cup_{\beta} S_{k,\beta} = \mathcal{S}^k$, as implied
  by \cref{lemma:evolution-Pauli-efficient-bound}.
  Let $ \vec{w}^{(r)}$ be the vector with coordinate
  \[
    w^{(r)}_{\beta} = \sum_{k=2}^{r} \frac{\abs{\mathcal{S}_{k,\beta}}S^k t^k}{k!}
  \]
  for some integer $r\ge 2$.

  Since $0\le \abs{\mathcal{S}_{k,\beta}} \le m^{k-1}$ and $\sum_{\beta\in \mathcal{X}}\abs{\mathcal{S}_{k,\beta}}\le m^{k}$
  as implied by \cref{lemma:evolution-Pauli-efficient-bound},
  using \cref{prop:opt-high-order-terms}, we have
  \[
    \Abs{\vec{w}^{(r)}[\mathcal{X}]}_2^2
    = \sum_{\beta\in \mathcal{X}} \rbra*{\sum_{k=2}^{r} \frac{\abs{\mc{S}_{k,\beta}}S^k t^k}{k!}}^2
    \le m \rbra*{\sum_{k=2}^{r} \frac{m^{k-1}S^k t^k}{k!}}^2 \le m R^2.
  \]
  Taking the limit $r\to \infty$ of the above inequality, we have
  \[
   \Abs{\vec{w}[\mathcal{X}]}_2^2 \le \lim_{r\to \infty} \Abs{\vec{w}^{(r)}[\mathcal{X}]}_2^2 \le mR^2.
  \]
  Then, by triangle inequality, we know
  \[
    \Abs{\vec{v}[\mathcal{X}]}_2 = \Abs{\ii t\vec{s}[\mathcal{X}] + \vec{w}[\mathcal{X}]}_2
    \ge \Abs{\vec{s}[\mathcal{X}]}t -\Abs{\vec{w}[\mathcal{X}]}_2
    \ge \Abs{\vec{s}[\mathcal{S}\cap\mathcal{X}]}t - \sqrt{m} R,
  \]
  where we use $\vec{s}[\mathcal{X}] =\vec{s}[\mathcal{S}\cap\mathcal{X}]$ in the last line.
\end{proof}

\begin{proof}[Proof of \cref{lm:Pauli_Upper_Bound}]
    By Taylor's expansion,
    we can write
    \[
    e^{-\ii Ht} = I - \ii H t + \sum_{k=2}^{\infty}
    \frac{(-\ii)^kA_k}{k!} t^k,
    \]
    where $A_k = \rbra*{\sum_{\alpha\in \mathcal{S}} s_{\alpha} P_{\alpha}}^k$.
    Then, we have, for any $\beta\in \mathbb{P}^n\setminus \{I\}$,
    \begin{equation*}
        v_{\beta} = \ii t s_{\beta} + \sum_{k=2}^{\infty}\frac{(-\ii)^k\tr(P_{\beta} A_k)}{2^n k!} t^k.
    \end{equation*}
    We can write
    \[
    \begin{aligned}
        \Abs{\vec{v}[\mathcal{X}]}_2^2
        &= \sum_{\beta\in \mathcal{X}} \rbra*{\ii t s_{\beta}
    + \sum_{k=2}^{\infty}
    \frac{(-\ii)^k\tr(P_{\beta} A_k)}{2^nk!} t^k}^2 \\
    & \le \sum_{\beta\in \mathcal{X}} \rbra*{t \abs{s_{\beta}}
    + \sum_{k=2}^{\infty}
    \frac{\abs{\tr(P_{\beta} A_k)}}{2^n k!} t^k}^2 \\
    & \le \sum_{\beta\in \mathcal{X}} \rbra*{t \abs{s_{\beta}} + \sum_{k=2}^{\infty}
    \frac{\abs{S_{k,\beta}}S^k t^k}{k!}}^2, \\
    \end{aligned}
    \]
    with $0\le \abs{\mathcal{S}_{k,\beta}}\le m^{k-1}$
    and $\sum_{\beta\in \mathcal{X}} \abs{\mathcal{S}_{k,\beta}} \le m^k$,
    as implied by \cref{lemma:evolution-Pauli-efficient-bound}.
    Let $ v^{(r)}$ be the vector with coordinate
    \[
    v^{(r)}_{\beta} =  t\abs{s_{\beta}} + \sum_{k=2}^{r}
    \frac{\abs{\mathcal{S}_{k,\beta}}s^k t^k}{k!}
    \]
    for some integer $r\ge 2$.

    By \cref{prop:opt-high-order-terms},
    we have
    \[
    \begin{aligned}
       \Abs{v^{(r)}[ \mathcal{X} ]}_2^2
       &= \sum_{\beta\in \mathcal{X}} \rbra*{t\abs{s_{\beta}}+ \sum_{k=2}^{r} \frac{\abs{\mathcal{S}_{k, \beta}}S^k t^k}{k!}}^2 \\
       &\le \sum_{\beta\in \mathcal{Y}} \rbra*{t\abs{s_{\beta}} + \sum_{k=2}^{r} \frac{m^{k-1} S^k t^k}{k!}}^2 \\
       & \le \sum_{\beta\in \mathcal{S}\cap \mathcal{X} } \rbra*{t\abs{s_{\beta}}+ R}^2 + (m-\abs{\mathcal{S}\cap\mathcal{X}})R^2,
    \end{aligned}
    \]
    where $\mathcal{Y}$ is a set
    of size $m$
    satisfying $(\mathcal{S}\cap \mathcal{X}) \subseteq \mathcal{Y}
    \subseteq \mathcal{X}$.
    Taking the limit $r\to \infty$
    of the above inequality,
    we have
    \[
    \Abs{\vec{v}[ \mathcal{X}]}_2^2 = \lim_{r\to \infty} \Abs{\vec{v}^{(r)} \mathcal{X}}_2^2 \le \sum_{\beta\in \mathcal{S}\cap \mathcal{X}} \rbra*{t\abs{s_{\beta}}+ R}^2 + (m-\abs{\mathcal{S}\cap\mathcal{X}})R^2.
    \qedhere
    \]
\end{proof}

\section{Deferred theorems and proofs in Bell-sate assisted certification}\label{sec:Appendix_Ancilla}

\subsection{Hamiltonian amplitude encoding}\label{sec:Appendix_Ancilla_Encoding}

\begin{proof}[Proof of \cref{prop:Bell_Dispersion}]
    By definition, we have
    \begin{gather}
    (U\otimes I_{2^n})\ket{\Phi^+}=\sum_{\alpha\in{\sf P}^n}u_\alpha (P_\alpha\otimes I_{2^n})\ket{\Phi^+}.\notag
\end{gather}
As for the Pauli-rotated states, we observe that for any $P_\alpha,P_\beta\in{\sf P}^n$:
\begin{gather}
    \bra{\Phi^+}(P_\alpha^\dag\otimes I_{2^n}) \cdot (P_\beta\otimes I_{2^n})\ket{\Phi^+}=\frac{\Tr(P_\alpha P_\beta)}{2^n}=\delta(\alpha,
    \beta).\notag
\end{gather}
Denoting $\ket{\Phi_\alpha}\coloneqq (P_\alpha\otimes I_{2^n})\ket{\Phi^+}$ completes the proof.
\end{proof}

\begin{proof}[Proof of \cref{prop:noisy_bound}]

We first calculate the state in \cref{eq:validity-appendix} explicitly as follows.
For $U_{\mathsf{HAE}} = VU_{t,r}U_{\mathsf{Prep}}$,
we can decompose $U_{t,r}$ as
$U_{t,r} = \sum_{\alpha} v_{\alpha} P_{\alpha}$.
Then, we have
\[
\begin{aligned}
    U_{\mathsf{HAE}}\ket{0}^{\otimes 2n+1}
    &= VU_{t,r}U_{\mathsf{Prep}}\ket{0}^{\otimes 2n+1} \\
    &=V\cdot (I_2 \otimes U_{t,r}\otimes I_{2^n})\ket{0}\ket{\Phi^+} \\
    &=V\ket{0}\left(\sum_{\alpha\in{\sf P}^n}v_\alpha P_\alpha\otimes I\ket{\Phi^+}\right)\\
    &= v_{I}\ket{0}\ket{\Phi^+}+\ket{1}\left(\sum_{\alpha\in{\sf P}^n,\alpha\neq I}v_\alpha P_\alpha\otimes I\ket{\Phi^+}\right)\\
    &=p_0(t)\ket{0}\ket{\Phi^+}+p_1(t)\ket{1}\ket{\psi},
\end{aligned}
\]
where $p_0(t)=v_I$, $p_1(t)=\sqrt{\sum_{\alpha\in{\sf P}^n,\alpha\neq I}\abs{v_\alpha}^2}$, and
\[
\ket{\psi} = \frac{1}{p_1(t)}\left(\sum_{\alpha\in{\sf P}^n,\alpha\neq I}v_\alpha P_\alpha\otimes I\ket{\Phi^+}\right).
\]

Now, applying \cref{prop:Trotter_Error},
we have
\[
\begin{aligned}
    \varepsilon_{\mathrm{Trotter}}& = \|e^{-\ii H_\mathrm{res} t}-U_{t,r}\| \\
    &\leq\frac{t^3}{12r^2}\|[H,[H,H_0]]\|+\frac{t^3}{24r^2}\|[H_0,[H_0,H]]\| \\
    &=\frac{t^3}{12r^2}\|[H,[H_\mathrm{res},H_0]]\|+\frac{t^3}{24r^2}\|[H_0,[H_0,H]]\| \\
    &=O\left(\frac{Sm_0\|H_0\|(\|H_0\|+Sm_0)t^3}{r^2}\right),
\end{aligned}
\]
where $\|[H,[H_\mathrm{res},H_0]]\| = 
O(\Abs{H}\Abs{H_{\mathrm{res}}}\Abs{H_0}) = O(S^2m_0^2 \Abs{H_0}+Sm_0 \Abs{H_0}^2)$,
and $\|[H_0,[H_0,H]]\| = O(\Abs{H_0}^2 S m_0)$, which follows from the bounds $\Abs{H} \le m_0 S$ and $\Abs{H_\mathrm{res}} \le \Abs{H} + \Abs{H_0} \le m_0 S + \Abs{H_0}$.

We denote the Pauli-coefficient vectors of $U_{t,r}$ and $e^{-\ii H_\mathrm{res}t}$ by $\Vec{v}$ and $\vec{w}$, respectively.
The distance between $\vec{v}$ and $\vec{w}$ can be bounded as
\[
\begin{aligned}
    \abs{\|\vec{v}[{\sf P}^n\backslash\{I\}]\|_2-\|\vec{w}[{\sf P}^n\backslash\{I\}]\|_2} 
    &\leq \|\vec{v}[{\sf P}^n\backslash\{I\}]-\vec{w}[{\sf P}^n\backslash\{I\}]\|_2 \\
    &\leq\|\vec{v}-\vec{w}\|_2 \\
    &=\|U_{t,r}-e^{-\ii H_\mathrm{res}t}\|_F \\
    & \leq\|U_{t,r}-e^{-\ii H_\mathrm{res}t}\|=\epsilon_\mathrm{Trotter},
\end{aligned}
\]
where the last inequality comes from the Fact~\ref{fact:Schatten_order}.

For the signal amplitude $u_1(t)$, we notice that $p_1(t)=\|\vec{v}[{\sf P}^n\backslash\{I\}]\|_2$.
Therefore, we have
\begin{gather}\label{eq:Frob_bound}
    \|\vec{w}[{\sf P}^n\backslash\{I\}]\|_2-\epsilon_\mathrm{Trotter}\leq p_1(t)\leq\|\vec{w}[{\sf P}^n\backslash\{I\}]\|_2+\epsilon_\mathrm{Trotter}.
\end{gather}

Applying 
\cref{lm:Pauli_Upper_Bound}
with $\mc X={\sf P}^n\backslash\{I\}$,
we have the following bounds given $R\leq1$,
\begin{align}
       \|\vec{w}[{\sf P}\backslash\{I\}]\|_2^2\leq\sum_{\alpha\in \mathcal{S}}(\abs{s_\alpha} t+R)^2.\notag
    \end{align}
Moreover, if we further know that $\|\Vec{s}\|_2t\geq \sqrt{m_0}R$, we can apply \cref{lm:Pauli_Lower_Bound2} to get,
\[
        \vec{w}[{\sf P}\backslash\{I\}]\|_2^2\geq(\|\Vec{s}\|_2t-\sqrt{m_0}R)^2.
\]

    Combining these results with Eq.~\eqref{eq:Frob_bound} completes the proof.
\end{proof}

\subsection{Robust Hamiltonian certification}\label{sec:Appendix_Tolera_Hamilt_Certi}

\begin{proof}[Proof of \cref{prop:iter_v_Tolerant}]
Given that $r=\Omega(\|H_0\|^{1/2}m_0^{-1/2}B_0^{-1/2}+\|H_0\|m_0^{-1}B_0^{-1})$ and $t = \frac{\xi}{2m_0^{3/2}B_0}$, we can bound the Trotter error in both cases as
\[
\begin{aligned}
    \epsilon_\mathrm{Trotter} &= O\left(\frac{(S^2m_0^2\|H_0\|+Sm_0\|H_0\|^2)t^3}{r^2}\right) \\
     &=O\left(\frac{\|\vec{s}\|_2}{m_0^{3/2}B_0}\right)
\end{aligned}
\]
by noting that $S\le \|\vec{s}\|_2$ always holds.
By carefully choosing constants, we can require $\epsilon_\mathrm{Trotter}\leq \frac{c\|\vec{s}\|_2}{2m_0^{3/2}B_0}$ such that $c<\frac{\xi^2}{60}$.

First, we bound the signal amplitude for the case $\Abs{\vec{s}}_2\le \varepsilon$.
We first observe that
$S:=\Abs{\vec{s}}_{\infty}\le \Abs{\vec{s}}_2\le \varepsilon$ and $S\leq B_0$.
Since $t=\xi/(2m_0^{3/2}B) \le 1/(2m_0S)$,  applying \cref{prop:noisy_bound},
we have
\[
R:= \sum_{k=2}^\infty\frac{m_0^{k-1}(St)^k}{k!}
= \frac{\exp(m_0St)-1-m_0St}{m_0}
\le \frac{3m_0^2 S^2 t^2}{5m_0} \le \frac{3\xi^2S^2}{20m_0^2B_0^2},
\]
by noting that $\exp(x)\le 1+ x+ 3x^2/5$ for $x\in \interval{0}{1/2}$.
Therefore, 
applying \cref{prop:noisy_bound} gives

\begin{align}\label{eq:up_1}
      p_1(t) &\leq \sqrt{\sum_{\alpha\in \mathcal{S}}\left(\abs{s_\alpha} t+R\right)^2}+\epsilon_\mathrm{Trotter} \notag\\
      &\leq \|\vec{s}\|_2t+\sqrt{m_0}R+\epsilon_\mathrm{Trotter} \notag\\
      &\leq \frac{\xi\|\vec{s}\|_2}{2m_0^{3/2}B_0}+\frac{3\xi^2S^2}{20m_0^{3/2}B_0^2}+\epsilon_\mathrm{Trotter}\notag\\
      &\leq\frac{\|\vec{s}\|_2}{2m_0^{3/2}B_0}\left(\xi+\frac{3}{10}\xi^2+c\right)\notag\\
      &\leq\frac{\varepsilon}{2m_0^{3/2}B_0}\left(\xi+\frac{3}{10}\xi^2+c\right),
\end{align}
where we use $S\leq B_0$ and $S\leq\|\vec{s}\|_2$.
For $\eta\leq1$, we have $c_1\geq(\eta+\frac{19\eta^2}{60})$.
If $\eta>1$, we have $\xi=1$ and $c_1\geq\frac{79}{60}$.

For $\|\vec{s}\|_2\geq(1+\eta)\varepsilon$, noting that
$t=\xi/(2m_0^{3/2}B) \le 1/(2m_0S)$,
we have
    \[
    R:=\sum_{k=2}^\infty\frac{m_0^{k-1}(St)^k}{k!}
    = \frac{\exp(m_0 S t)-1-m_0 St}{m_0}
    \le \frac{3m_0^2 S^2 t^2}{5m_0} \le 
    \frac{3 \xi^2 S^2}{20m_0^2 B_0^2},
    \]  
    where we use $\exp(x)-1-x\le 3x^2/5$ 
    for $x\in \interval[open]{0}{1/2}$.
    Then, applying the lower bound in~\cref{prop:noisy_bound}, 
    we have
    \begin{align}\label{eq:low_1}
        p_1(t) &\ge \|\vec{s}\|_2t-\sqrt{m_0}R-\epsilon_\mathrm{Trotter} \notag\\
      &\geq \frac{\xi\|\vec{s}\|_2}{2m_0^{3/2}B_0}-\frac{3\xi^2S^2}{20m_0^{3/2}B_0^2}-\epsilon_\mathrm{Trotter}\notag\\
      &\geq\frac{\|\vec{s}\|_2}{2m_0^{3/2}B_0}\left(\xi-\frac{3}{10}\xi^2-c\right)\notag\\
      &\geq\frac{(1+\eta)\varepsilon}{2m_0^{3/2}B_0}\left(\xi-\frac{3}{10}\xi^2-c\right).
    \end{align}
For $\eta\leq1$, we have $c_2\leq(\eta+\frac{41\eta^2-19\eta^3}{60})$.
If $\eta>1$, we have $\xi=1$ and $c_2\leq\frac{41(1+\eta)}{60}$.
\end{proof}

\begin{proof}[Proof of Theorem~\ref{thm:tolerant}]
Let $\eta=\frac{\varepsilon_2-\varepsilon_1}{\varepsilon_1}$ and $\xi=\min(\eta,1)$.
Let $ H_\mathrm{res} \coloneqq H - H_0 $ 
and $ \vec{s} $ be its Pauli-coefficient vector. 
By the assumption, we know that $ \vec{s} $ contains at most $ m_0 = 2m $ non-zero coordinates, each bounded by $B_0= 2B $. Furthermore, $ H_\mathrm{res} $ is traceless.

In the following, we will
consider two cases 
respectively, and 
proves the algorithm
will output the correct
answer with probability 
at least $1-\delta$
in both cases.

\textbf{Case 1:} $\Abs{H-H_0}_{F}\le \varepsilon_1$.

In this check, we have
$t = \frac{\xi}{2m_0^{3/2}B_0}$,
and $r = \Theta(1)$.
Note that
$\|H_0\|\leq mB$ by the assumption, and we can calculate:
\[
r=\Theta(1)=\Omega\left(\|H_0\|^{1/2}m_0^{-1/2}B_0^{-1/2}+\|H_0\|m_0^{-1}B_0^{-1}\right).
\]
This fulfills the requirement in~\cref{prop:iter_v_Tolerant} for $r$.
Thus, applying
the proposition,
we know the signal amplitude
$p_1(t)$ given by
$\mathsf{HAE}(\mathcal{O}_H, H_0, t, r)$
satisfies
$p_1(t)\le \frac{c_1\varepsilon_1}{m_0^{3/2}B_0}$.
Therefore, applying 
\cref{thm:sqrt-amp-testing}, 
we know the \textsf{AmpTest} can output \textsc{Large} with probability at most $\delta$.

\textbf{Case 2:} $\|H-H_0\|_F\geq(1+\eta)\varepsilon_1$.

Similarly, we have:
\[
r=\Theta(1)=\Omega\left(\|H_0\|^{1/2}m_0^{-1/2}B_0^{-1/2}+\|H_0\|m_0^{-1}B_0^{-1}\right).
\]
This fulfills the requirement in~\cref{prop:iter_v_Tolerant} for $r$.
Thus, under this case,
we know the signal amplitude
$p_1(t)$ given by
$\mathsf{HAE}(\mathcal{O}_H, H_0, t, r)$
satisfies
$p_1(t)\ge \frac{c_2\varepsilon_1}{m_0^{3/2}B_0}$.
Therefore, applying 
\cref{thm:sqrt-amp-testing}, 
we know the \textsf{AmpTest} can output \textsc{Small} with probability at most $\delta$.

Regarding the complexities, we compute them for two distinct scenarios.

For $\eta\leq 1$,
it is clear to check that the gap between signal amplitudes from different cases is:
    \begin{equation*}
        \frac{\varepsilon_1}{2m_0^{3/2}B_0}\left[(1+\eta)\left(\eta-\frac{3\eta^2}{10}-c\right)-\eta+\frac{3\eta^2}{10}+c\right]\geq\frac{\eta^2\varepsilon_1}{40m_0^{3/2}B_0}.
    \end{equation*}
According to \cref{thm:sqrt-amp-testing}, the \textsf{AmpTest} will query \textsf{HAE} and its inverse for a total of $$O\left(\frac{m_0^{3/2}B_0}{\eta^2\varepsilon_1}\log(\frac{1}{\delta})\right)$$
times. 
In each query, the subroutine \textsf{HAE} implements the second-order Trotter formula, requiring $r=\Theta(1)$ queries of the controlled evolution from $\mc O_H$.
Therefore, the query complexity is
\[ 
    O\left(\frac{m_0^{3/2}B_0}{\eta^2\varepsilon_1}\log(\frac{1}{\delta})\right)=O\left(\frac{m^{3/2}B\varepsilon_1}{(\varepsilon_2-\varepsilon_1)^2}\log(\delta^{-1})\right)
\]
    Since all the subroutines query with the same evolution time, the total time queried is
\[
\begin{aligned}
T&=O\left(\frac{m_0^{3/2}B_0}{\eta^2\varepsilon_1}\log(\frac{1}{\delta})\right)\times\frac{\eta}{2m_0^{3/2}B_0}=O((\varepsilon_2-\varepsilon_1)^{-1}\log(\delta^{-1})).
\end{aligned}
\]
The total measurement complexity can also be inherited from the calling of \textsf{AmpTest}.
According to \cref{thm:sqrt-amp-testing}, the measurement complexity is
\[
\begin{aligned}
            O\left(\log(\frac{m_0^{3/2}B_0}{\eta^2\varepsilon_1})+\log(\delta^{-1})\right)=O(\log(mB(\varepsilon_2-\varepsilon_1)^{-1})+\log(\delta^{-1})).
\end{aligned}
\]
In the above derivation, we use $m_0=2m$, $B_0=2B$, and that $\eta=\frac{\varepsilon_2-\varepsilon_1}{\varepsilon_1}$.

For $\eta> 1$,
it is clear to check that the gap between signal amplitudes from different cases is:
    \begin{equation*}
        \frac{\varepsilon_1}{2m_0^{3/2}B_0}\left[\frac{41(1+\eta)-79}{60}\right]\geq\frac{(41\eta-38)\varepsilon_1}{120m_0^{3/2}B_0}.
    \end{equation*}
According to \cref{thm:sqrt-amp-testing}, the \textsf{AmpTest} will query \textsf{HAE} and its inverse for a total of $$O\left(\frac{m_0^{3/2}B_0}{\eta\varepsilon_1}\log(\frac{1}{\delta})\right)$$
times. 
In each query, the subroutine \textsf{HAE} implements the second-order Trotter formula, requiring $r=\Theta(1)$ queries of the controlled evolution from $\mc O_H$.
Therefore, the query complexity is
\[ 
    O\left(\frac{m_0^{3/2}B_0}{\eta\varepsilon_1}\log(\frac{1}{\delta})\right)=O\left(\frac{m^{3/2}B}{\varepsilon_2-\varepsilon_1}\log(\delta^{-1})\right)
\]
    Since all the subroutines query with the same evolution time, the total time queried is
\[
\begin{aligned}
T&=O\left(\frac{m_0^{3/2}B_0}{\eta\varepsilon_1}\log(\frac{1}{\delta})\right)\times\frac{1}{2m_0^{3/2}B_0}=O((\varepsilon_2-\varepsilon_1)^{-1}\log(\delta^{-1})).
\end{aligned}
\]
The total measurement complexity can also be inherited from the calling of \textsf{AmpTest}.
According to \cref{thm:sqrt-amp-testing}, the measurement complexity is
\[
\begin{aligned}
            O\left(\log(\frac{m_0^{3/2}B_0}{\eta\varepsilon_1})+\log(\delta^{-1})\right)=O(\log(mB(\varepsilon_2-\varepsilon_1)^{-1})+\log(\delta^{-1})).
\end{aligned}
\]
In the above derivation, we use $m_0=2m$, $B_0=2B$, and that $\eta=\frac{\varepsilon_2-\varepsilon_1}{\varepsilon_1}$.

Combining the above arguments, 
we have finished the proof.
\end{proof}

\subsection{Coherent Hamiltonian cerification with one-sided error}
\label{sec:coherent-hamiltonian-certification-one-side}

\begin{proof}[Proof of Thoerem~\ref{thm:Ancilla_certification_Pauli}]
    Note that for all (Hamiltonian) operators $H$ and $H_0$ with their Pauli-coefficient vectors $\vec{s}$ and $\vec{s}^{(0)}$, we have the following equation:
    \begin{gather}
        \|H-H_0\|_F=\|H-H_0\|_{\mathrm{Pauli},2}.
    \end{gather}
    Therefore, according to the order in Fact~\ref{fact:Pauli_order}, we have 
    \begin{gather}
        \|H-H_0\|_F\geq\|H-H_0\|_{\mathrm{Pauli},p},\notag
    \end{gather}
    for all $p\in[2,\infty]$.

    We start from the case to distinguish between different Hamiltonian cases with the Pauli $p$-norm for some $p\in[2,\infty]$.
    If the underlying Hamiltonians satisfy the \textsc{Accept} case, this implies that $H=H_0$, which coincides with the \textsc{Accept} case in Theorem~\ref{thm:Bell+control}.
    Otherwise, if the underlying Hamiltonians satisfy the \textsc{Reject} case, we have
    \begin{align}
        \|H-H_0\|_F\geq\|H-H_0\|_{\mathrm{Pauli},p}\geq\varepsilon,\notag
    \end{align}
    which also indicates the \textsc{Reject} case in Theorem~\ref{thm:Bell+control}.
    Based on these equavilences, running algorithm $\textsf{CHC}(H_0,\mc O_H,m,B,\varepsilon,\delta)$ can distinguish these two cases with stated complexities in Theorem~\ref{thm:Bell+control}, which completes the proof in this theorem for $p\in[2,\infty]$.

    Suppose the task aims at distinguishing between different Hamiltonian cases with the Pauli $p$-norm for some $p\in[1,2)$.
    If the underlying Hamiltonians satisfy the \textsc{Accept} case, this implies that $H=H_0$, which coincides with the \textsc{Accept} case in Theorem~\ref{thm:Bell+control}.
    Similarly, if the underlying Hamiltonians satisfy the \textsc{Reject} case, we have
    \begin{align}
        \|H-H_0\|_F\geq m^{1/2-1/p}\|H-H_0\|_{\mathrm{Pauli},p}\geq m^{1/2-1/p}\varepsilon,\notag
    \end{align}
    where we have used the relation between $\ell_p$ norms.
    This \textsc{Reject} case of the $p$-norm also indicates the \textsc{Reject} case with an updated threshold, $m^{1/2-1/p}\varepsilon$ for the normalized Frobenius norm in Theorem~\ref{thm:Bell+control}.
    Based on these equavilences, running algorithm $\textsf{CHC}(H_0,\mc O_H,m,B,m^{1/2-1/p}\varepsilon,\delta)$ can distinguish these two cases with stated complexities in Theorem~\ref{thm:Bell+control} by replacing the original $\varepsilon$ by $m^{1/2-1/p}\varepsilon$.

    Combining these two cases completes the proof.
\end{proof}

\begin{proof}[Proof of Theorem~\ref{thm:Ancilla_certification_Schatten}]
    Note that for all (Hamiltonian) operators $H$ and $H_0$, the normalized Frobenius norm just equal to the normalized Schatten 2-norm:
    \begin{gather}
        \|H-H_0\|_F=\|H-H_0\|_{\mathrm{Schatten},p}.\notag
    \end{gather}
    Therefore, according to the order in Fact~\ref{fact:Schatten_order}, we have 
    \begin{gather}
        \|H-H_0\|_F\geq\|H-H_0\|_{\mathrm{Schatten},p},\notag
    \end{gather}
    for all $p\in[1,2]$.

    Suppose we need to distinguish between different Hamiltonian cases with normalized Schatten $p$-norm for some $p\in[2,\infty]$.
    If the underlying Hamiltonians satisfy the \textsc{Accept} case, this implies that $H=H_0$, which coincides with the \textsc{Accept} case in Theorem~\ref{thm:Bell+control}.
    Otherwise, if the underlying Hamiltonians satisfy the \textsc{Reject} case, we have
    \begin{align}
        \|H-H_0\|_F\geq\|H-H_0\|_{\mathrm{Schatten},p}\geq\varepsilon,\notag
    \end{align}
    which also indicates the \textsc{Reject} case in Theorem~\ref{thm:Bell+control}.
    Based on these equivalences, running algorithm $\textsf{CHC}(H_0,\mc O_H,m,B,\varepsilon,\delta)$ can distinguish these two cases with stated complexities in Theorem~\ref{thm:Bell+control}, which completes the proof.
\end{proof}

\section{Defered theorems and proofs in ancilla-free certification}
\subsection{Stabilizer Bernoulli sampling after Hamiltonian evolution}\label{sec:Appendix_Stab_Bernoulli}

\begin{proof}[Proof of \cref{lm:stabilizer_sample}]
    It is direct to verify that \textsf{SBS} is a valid Bernoulli sampler.
    Thus, we focus on the derivation of its signal probability.

    For an arbitrary unitary operator $U$ on $\mc H(2^n)$, the Pauli decomposition of it can be written as 
\[
    U=\sum_{\alpha\in{\sf P}^n}u_\alpha P_\alpha.
\]
For $\theta\in\mc A_{\mc S_0}$ and $\ket{\phi_0}$ the state stabilized by ${\mc S_0}$,
applying $U$ to a state $\ket{\phi_{\sigma(\theta)}} = P_{\theta} \ket{\phi_0}$ 
yields the final state 
\[
\begin{aligned} 
    U\ket{\phi_{\sigma(\theta)}}&=\sum_{\alpha\in{\sf P}^n}u_\alpha P_\alpha P_\theta\ket{\phi_0} \\
    &=\sum_{\beta\in\mc A_{\mc S_0}}\sum_{\gamma\in\mc S_0}u_{\beta+\gamma} v_{\beta,\gamma}P_\beta P_\gamma P_\theta\ket{\phi_0} \\
     &=\sum_{\beta\in\mc A_{\mc S_0}}\sum_{\gamma\in\mc S_0} (-1)^{\pip{\gamma}{\theta}}u_{\beta+\gamma}v_{\beta,\gamma}P_\beta  P_\theta P_\gamma\ket{\phi_0} \\
    &=\sum_{\beta\in\mc A_{\mc S_0}}\left(\sum_{\gamma\in\mc S_0}(-1)^{\pip{\gamma}{\theta}}u_{\beta+\gamma}v_{\beta,\gamma}\right)P_\beta P_\theta\ket{\phi_0},
\end{aligned}
\]
where we use \ref{prop:Unique_Decomp}
for decomposing $P_{\alpha}$.

Therefore, performing syndrome measurements,
the probability of getting $\sigma(\theta)$
just 
\[
\abs{\sum_{\gamma\in\mc S_0}(-1)^{\pip{\gamma}{\theta}}u_{\gamma}}^2
= \sum_{\gamma\in\mc S_0} \sum_{\gamma'\in\mc S_0}(-1)^{\pip{\gamma+\gamma'}{\theta}}u_{\gamma} \bar{u}_{\gamma'}.
\]

Averaging over the uniformly random 
choice of $\theta$, we get
\[
\begin{aligned}
    \Pr(Z = 0) &= \frac{1}{2^n} \sum_{\theta\in\mc A_{\mc S_0}} \sum_{\gamma\in\mc S_0} \sum_{\gamma'\in\mc S_0}(-1)^{\pip{\gamma+\gamma'}{\theta}}u_{\gamma} \bar{u}_{\gamma'} \\
    &= \frac{1}{2^n}  \sum_{\gamma\in\mc S_0} \sum_{\gamma'\in\mc S_0} u_{\gamma} \bar{u}_{\gamma'} \sum_{\theta\in\mc A_{\mc S_0}} (-1)^{\pip{\gamma+\gamma'}{\theta}}\\
    &= \sum_{\gamma\in\mc S_0} \abs{u_{\gamma}}^2, \\
\end{aligned}
\]
where we have used \cref{lm:Sum_Delta}
to simplify the inner summation in the second last line.
The result then follows directly
by noting $\Abs{\vec{u}}^2 = 1$
as $U$ is a unitary matrix.
\end{proof}

\begin{proof}[Proof of \cref{lm:noisy_bound_stab}]
We first bound the Trotter error
in diamond norm as follows.
Similar to the proof of \cref{prop:noisy_bound},
we have
\[
\begin{aligned}
\varepsilon_{\mathrm{Trotter},\diamond}
&\coloneqq \|\mc U_\mathrm{res}(t)-\mc U_2(t)\|_\diamond \\
&\leq 2\|e^{-\ii H_\mathrm{res}t}-U_2(t/r)^r\| \\
&\leq \frac{t^3}{6r^2}\|[H,[H,H_0]]\|+\frac{t^3}{12r^2}\|[H_0,[H_0,H]]\| \\
&=O\left(\frac{Sm_0\|H_0\|(\|H_0\|+Sm_0)t^3}{r^2}\right),
\end{aligned}
\]
where we have used \cref{prop:Trotter_Error,prop:diamond}.

Now, let $Y$ be the random variable
of the output of the algorithm
which is the same as $\mathsf{HSS}(\mathcal{S}_0, H_0, Q_H, t, r)$
except that  $\left(e^{\ii H_0t/(2r)}\mc O_H(t/r)e^{\ii H_0t/(2r)}\right)^r$
is replaced by
$e^{\ii H_{\mathrm{res}}t}$.
Then, by the operational meaning of diamond norm,
regarding the algorithm $\mathsf{HSS}$
as a composition of quantum channels,
we know
\[
\abs{\Pr(Y=1)-\Pr(Z=1)} \le \varepsilon_{\mathrm{Trotter},\diamond}.
\]
Furthermore,
we have
\[
\Pr(Y=1) = \|\vec{v}[{\sf P}^n\backslash\mc S_0]\|_2^2 
\]
by \cref{lm:stabilizer_sample}, where $\vec{v}$ is the Pauli coefficient
vector of $e^{\ii H_{\mathrm{res}}t}$.
The upper and lower bounds 
then follows directly by  \cref{lm:Pauli_Lower_Bound2,lm:Pauli_Upper_Bound}.
\end{proof}

\subsection{Stabilizer Hamiltonian certification}\label{sec:Appendix_SHC}

\begin{proof}[Proof of \cref{prop:iter_v_stab}]
Given that $r=\Omega(m_0^{1/4}\|H_0\|^{1/2}b^{-1/2}+m_0^{-1/4}\|H_0\|b^{-1})$
and $t=m_0^{-3/2}b^{-1}/2$, we can bound the Trotter error as
\[
\epsilon_\mathrm{Trotter,\diamond}=O\left(\frac{(S^2m_0^2\|H_0\|+Sm_0\|H_0\|^2)t^3}{r^2}\right)=O(m_0^{-3}),
\]
by noting that $S\le b$.
By choosing the constant
to be sufficiently large, 
we can require $\epsilon_\mathrm{Trotter,\diamond}\leq0.002/m_0^3$.

Now, assume $\Abs{\vec{s}}_2\le \varepsilon$.
We first observe that
$S:=\Abs{\vec{s}}_{\infty}\le \Abs{\vec{s}}_2\le \varepsilon\le b/4$.
Since $t=1/(2m_0^{3/2}b) \le \varepsilon/(2b m_0 S) \le 1/(8m_0S)$,  applying \cref{prop:noisy_bound},
we have
\[
R:= \sum_{k=2}^\infty\frac{m_0^{k-1}(St)^k}{k!}
= \frac{\exp(m_0St)-1-m_0St}{m_0}
\le \frac{12m_0^2 S^2 t^2}{23m_0} \le \frac{3}{368m_0^2},
\]
by noting that $\exp(x)\le 1+ x+ 12x^2/23$ for $x\in \interval{0}{1/8}$.
Now, applying \cref{lm:noisy_bound_stab},
we have
\[
\begin{aligned}
    \Pr(Z=1) &\le \sum_{\alpha\in \mc S\cap\mc X}\left(\abs{s_\alpha} t+ R\right)^2+(m_0-\abs{\mc S\cap\mc X})R^2+\epsilon_\mathrm{Trotter,\diamond} \\
    &= \Abs{\vec{s}[\mathcal{S}\cap \mathcal{X}]}^2_2 t^2 +2Rt\Abs{\vec{s}[\mathcal{S}\cap \mathcal{X}]}_1+m_0R^2+\epsilon_\mathrm{Trotter,\diamond} \\
    &\le \Abs{\vec{s}}^2_2 t^2 +2Rt\Abs{\vec{s}}_2\sqrt{m_0}+m_0R^2+\epsilon_\mathrm{Trotter,\diamond} \\
    &\le \frac{0.02}{m_0^3}.
\end{aligned}
\]

Now we consider the case when
$\|\vec{s}[{\sf P}^n\backslash\mc S_0]\|_2\geq\|\vec{s}[\mc S_0]\|_2$
with $b/2 < S\le b$.
Since $t= 1/(2m_0^{3/2}b)\leq 1/(2m_0S)$, 
we have
\[
R:= \sum_{k=2}^\infty\frac{m_0^{k-1}(St)^k}{k!}
= \frac{\exp(m_0St)-1-m_0St}{m_0}
\le \frac{3m_0^2 S^2 t^2}{5m_0} = \frac{3S^2}{20m_0^2b^2},
\]
where we use $e^x-1-x\leq\frac{3x^2}{5}$ for $x\in\interval{0}{1/2}$.
Noting that 
\[
\|\vec{s}[\mc S\cap\mc X]\|_2
= \|\vec{s}[\mathsf{P}^n\backslash\mathcal{S}_0]\|_2
\ge \|\vec{s}[\mathcal{S}_0]\|_2
\]
and $S\le \max\{\|\vec{s}[\mc S\cap\mc X]\|_2,  \|\vec{s}[\mathcal{S}_0]\|_2\}$,
we have
\[
\|\vec{s}[\mc S\cap\mc X]\|_2
\ge S,
\]
By our choice of $t= 1/(2m_0^{3/2}b)$,
we know
\[
  \|\vec{s}[\mc S\cap\mc X]\|_2 t\geq\frac{S}{2m_0^{3/2}b}\geq\sqrt{m_0}R,
\]
Then, 
applying the lower bound in \cref{lm:noisy_bound_stab},
we have
\[
\begin{aligned}
    \Pr(Z=1) 
    &\ge \left( \|\vec{s}[\mc S\cap\mc X]\|_2 t -\sqrt{m_0}R\right)^2 - \epsilon_\mathrm{Trotter,\diamond} \\
    &\ge \frac{0.025}{m_0^3}.
\end{aligned}
\]
Combining the above discussions,
we know the claim holds as we want.
\end{proof}

\begin{proof}[Proof of \cref{prop:fina_v_stab}]
  Given that $r=\Omega(m_0^{1/4}\|H_0\|^{1/2}\varepsilon^{-1/2}+m_0^{-1/4}\|H_0\|\varepsilon^{-1})$
  and 
  $t = m_0^{-3/2}\varepsilon^{-1}/4$, we can bound the Trotter error as
  \[
  \epsilon_\mathrm{Trotter}=\order{\frac{(S^2m_0^2\|H_0\|+Sm_0\|H_0\|^2)t^3}{r^2}}=\order{m_0^{-3}}.
  \]
  By choosing the constant to be 
  sufficiently large, 
  we can require  $\epsilon_\mathrm{Trotter,\diamond}\leq0.002/m_0^3$.

  We now bound the remainder term
  $R$ in both cases.
  Since $t= 1/(4m_0^{3/2}\varepsilon)\leq 1/(4m_0S)$, 
    we have
\[
    R:= \sum_{k=2}^\infty\frac{m_0^{k-1}(St)^k}{k!}
    = \frac{\exp(m_0St)-1-m_0St}{m_0}
\le \frac{6m_0^2 S^2 t^2}{11m_0} = \frac{3S^2}{88m_0^2\varepsilon^2},
\]
where we use $e^x-1-x\leq 6x^2/11$ for $x\in\interval{0}{1/4}$.

   We first consider the case where
   $\Abs{\vec{s}}_2\le \varepsilon$.
Then, given that $t\leq 1/(m_0S)$, we can use the upper bound from \cref{lm:noisy_bound_stab} as follows
\[
\begin{aligned}
    \Pr(Z=1) &\le \sum_{\alpha\in \mc S\cap\mc X}\left(\abs{s_\alpha} t + R\right)^2+(m_0-\abs{\mc S\cap\mc X})R^2 + \epsilon_\mathrm{Trotter,\diamond} \\
    &= \Abs{\vec{s}[\mathcal{S}\cap \mathcal{X}]}_2^2 t^2 + 
    2\Abs{\vec{s}[\mathcal{S}\cap \mathcal{X}]}_1 t R
    + m_0 R^2 + \epsilon_\mathrm{Trotter,\diamond}
    \\
     &\le \Abs{\vec{s}}_2^2 t^2 + 
    2\sqrt{m_0} Rt\Abs{\vec{s}}_2
    + m_0 R^2 + \epsilon_\mathrm{Trotter,\diamond}
    \\
    &\le \frac{0.09}{m_0^3}.
\end{aligned}
\]

    We then consider the case where $S\le 2\varepsilon$ and $\|\vec{s}[{\sf P}^n\backslash\mc S_0]\|_2\ge 2\varepsilon$.
    In this case, we have
    \[
     \|\Vec{s}[\mc S\cap\mc X]\|_2t\geq\frac{1}{2m_0^{3/2}}\geq\sqrt{m_0}R.
    \]
    Therefore,
    applying \cref{lm:noisy_bound_stab},
    we have
    \[
        \Pr(Z = 1) \ge 
        \left(\|\Vec{s}[\mc S\cap\mc X]\|_2t-\sqrt{m_0}R\right)^2-\epsilon_\mathrm{Trotter,\diamond} \ge \frac{0.12}{m_0^3}.
    \]

    Combining the above discussions,
    we have proved the desired result.
\end{proof}

\begin{proof}[Proof of \cref{thm:Stabilizer}]
    From the assumption, 
    we know the residual Hamiltonian $H_\mathrm{res} \coloneqq H - H_0$
    is traceless, contains
    at most $m_0 = 2m$
    non-zero Pauli terms,
    and 
    its Pauli-coefficient vector $\vec{s}$
    satisfies $\Abs{\vec{s}}_{\infty}\le 2B $.

    To establish the correctness of this algorithm, we will bound the failure probabilities, specifically the false rejection and false acceptance rates.
    During this analysis, we will extensively use \cref{prop:iter_v_stab} as well as~\cref{prop:fina_v_stab}.
     For clarity of reference, we denote the condition $\|\vec{s}\|_2\leq\varepsilon$ in the first case of both propositions as the \enquote{small-case} condition, representing perfect matching between Hamiltonians. 
     Similarly, we refer to the conditions in the second case of either proposition as the \enquote{large-case} condition, indicating significant deviations that should trigger rejection.

    During the false rejection, the algorithm returns \textsc{Reject} if any of the checks report \textsc{Large}.
    Given the \textsc{Accept} case where $\|H - H_0\|_F \leq \varepsilon $, the residual Hamiltonian satisfies $ \|\vec{s}\|_2 \leq \varepsilon $, aligning with the small case in \cref{prop:iter_v_stab} or \cref{prop:fina_v_stab}.

    For each check $ i $ in the first $ k$ checks of either round $j$, the algorithm utilizes the sampler \textsf{HSS} with parameters $ t = m_0^{-3/2}b^{-1}/{2}$ and $ r = O(m_0^{3/4}Bb^{-1}) $, where $ b = 2^{k-i+2}\epsilon$. 
    Given that $ \|H_0\| \leq mB $, we can express $ r $ as:
    \begin{gather}
        r=O(m_0^{3/4}Bb^{-1})=\Omega(m_0^{1/4}\|H_0\|^{1/2}b^{-1/2}+m_0^{-1/4}\|H_0\|b^{-1}),\notag
    \end{gather}
    which satisfies the setting requirement of \cref{prop:iter_v_stab}.
    Consequently, the proposition implies the signal probability in the small case as:
    \begin{gather}
        \Pr(Z=1\,|\,\text{case }\textsc{Accept})=\Pr(Z=1\,|\,\text{Small case})\leq\frac{c_1}{m_0^{3}}.\notag
    \end{gather} 
    According to Proposition~\ref{prop:Bernoulli-distribution-parameter-testing}, the \textsf{BernoulliTest} subroutine can output the correct answer with high probability, given the specified bounds $c_2/m_0^3>c_1/m_0^3$ for signal probabilities.
    Thus, the false-rejection rate for this check is bounded as follows:
    \begin{align}
        &\Pr(\text{Rejected by the } i\text{th check of the }j\text{th round}\,|\, \textsc{Accept})\notag\\
        =\:&\Pr(\textsf{BernoulliTest}\text{ returns \textsc{Large} in the }i\text{th check of the }j\text{th round}\,|\, \text{Small case})\notag\\
        \leq\:& \delta'=\frac{\delta}{2(k+1)}.\notag
    \end{align}

    In the final check of each round $j$, the algorithm again employs \textsf{HSS} with $t=\frac{1}{4m_0^{3/2}\varepsilon}$ and $r=O(m_0^{3/4}B\varepsilon^{-1})$.
    Similarly, this leads to
    \begin{gather}
        r=O(m_0^{3/4}B\varepsilon^{-1})=\Omega(m_0^{1/4}\|H_0\|^{1/2}\varepsilon^{-1/2}+m_0^{-1/4}\|H_0\|\varepsilon^{-1}),\notag
    \end{gather}
    satisfying the requirement in \cref{prop:fina_v_stab}.
    Thus, we can bound the signal probability in this small case as:
    \begin{gather}
        \Pr(Z=1\,|\,\text{case }\textsc{Accept})=\Pr(Z=1\,|\,\text{Small Case})\leq\frac{c_3}{m_0^{3}}.\notag
    \end{gather}
    Based on Proposition~\ref{prop:Bernoulli-distribution-parameter-testing}, the \textsf{BernoulliTest} subroutine can output the correct answer with high probability , given that the bounds $c_3/m_0^3$ and $c_4/m_0^3$ are known.
    Therefore, the false-rejection rate in this check is bounded as follows,
    \begin{align}
        &\Pr(\text{Rejected by the final check of the }j\text{th round}\,|\, \textsc{Accept})\notag\\
        =\:&\Pr(\textsf{BernoulliTest}\text{ returns \textsc{Large} in the final check of the }j\text{th round}\,|\, \text{Small case})\notag\\
        \leq\:&\delta'=\frac{\delta}{2(k+1)}.\notag
    \end{align}

    Since the rejection occurs in exactly one check among all $2(k+2)$ checks across two rounds, we can decompose the false-rejection rate using the addition law of disjoint events:
    \begin{align}
        &\Pr(\text{Rejected by Alg.~\ref{alg:SHC}}\,|\, \textsc{Accept})\notag\\
        =&\,\sum_{j=1}^2\sum_{i=1}^{k+1}\Pr(\text{Rejected by the } i\text{th check of the }j\text{th round}\,|\, \textsc{Accept})\leq\delta.\notag
    \end{align}

    For false acceptance, it occurs only when the underlying residual Hamiltonian is large, specifically when $ \|\vec{s}\|_2 \geq 4\varepsilon $, which is equivalent to the Large-case condition of both propositions.
    For simplicity, we define with respect to $ \mathcal{S}_Z \coloneqq \langle I,Z \rangle^{\otimes n} $ and $ \mathcal{S}_X \coloneqq \langle I,X \rangle^{\otimes n} $.
    The Pauli coefficients of $ H_\mathrm{res} $ must satisfy \textbf{at least one} of the following cases:
    \begin{itemize}
        \item Case $\mathrm{R}_{1,1}$: $2^k\varepsilon<S\coloneqq\|\Vec{s}\|_\infty\leq 2^{k+1}\varepsilon$, and $\|\vec{s}[{\sf P}^n\backslash\mc S_{Z}]\|_2\geq\|\vec{s}[\mc S_{Z}]\|_2 $;
        \item Case $\mathrm{R}_{1,2}$: $2^k\varepsilon<S\coloneqq\|\Vec{s}\|_\infty\leq 2^{k+1}\varepsilon$, and $\|\vec{s}[{\sf P}^n\backslash\mc S_{X}]\|_2\geq\|\vec{s}[\mc S_{X}]\|_2 $;
        \item Case $\mathrm{R}_{2,1}$: $2^{k-1}\varepsilon<S\leq 2^k\varepsilon$, and $\|\vec{s}[{\sf P}^n\backslash\mc S_{Z}]\|_2\geq\|\vec{s}[\mc S_{Z}]\|_2 $;
        \item Case $\mathrm{R}_{2,2}$: $2^{k-1}\varepsilon<S\leq 2^k\varepsilon$, and $\|\vec{s}[{\sf P}^n\backslash\mc S_{X}]\|_2\geq\|\vec{s}[\mc S_{X}]\|_2 $;
        \item $\cdots$
        \item Case $\mathrm{R}_{k+1,1}$: $S\leq 2\varepsilon$, and $\|\vec{s}[{\sf P}^n\backslash\mc S_{Z}]\|_2\geq2\varepsilon$;
        \item Case $\mathrm{R}_{k+1,2}$: $S\leq 2\varepsilon$, and $\|\vec{s}[{\sf P}^n\backslash\mc S_{X}]\|_2\geq2\varepsilon$.
    \end{itemize}
    To justify the completeness of this decomposition, we first note that $\mc S_Z\cap\mc S_X=\{I\}$ and the residual Hamiltonian is traceless.
    Suppose $\|\vec{s}[{\sf P}^n\backslash\mc S_{Z}]\|_2<\|\vec{s}[\mc S_{Z}]\|_2 $.
    In this case, we have
    \begin{gather}
        \|\vec{s}[{\sf P}^n\backslash\mc S_{X}]\|_2\geq\|\vec{s}[\mc S_{Z}]\|_2>\|\vec{s}[{\sf P}^n\backslash\mc S_{Z}]\|_2\geq\|\vec{s}[\mc S_{X}]\|_2.\notag
    \end{gather}
    Similarly, we can prove $\|\vec{s}[{\sf P}^n\backslash\mc S_{Z}]\|_2\geq\|\vec{s}[\mc S_{Z}]\|_2 $ given $\|\vec{s}[{\sf P}^n\backslash\mc S_{X}]\|_2<\|\vec{s}[\mc S_{X}]\|_2 $.
    Noticing that $\|\vec{s}\|_2\geq4\varepsilon$, we can guarantee either $\|\vec{s}[{\sf P}^n\backslash\mc S_{Z}]\|_2\geq2\varepsilon$ or $\|\vec{s}[{\sf P}^n\backslash\mc S_{X}]\|_2\geq2\varepsilon$.
    Therefore, by further enumerating all scenarios of the largest coefficient $S$, we can conclude that the preceding cases cover all possible situations.

    Due to the equation of two events, $\bigcup_{i,j}\text{case }\mathrm{R}_{i,j}=\textsc{Reject}$, we can calculate the false-acceptance rate as follows:
\begin{align}\label{eq:false-acceptance2}
        &\Pr(\text{Accepted by Alg.~\ref{alg:SHC}}\,|\, \textsc{Reject})=\frac{\Pr(\text{Accepted by Alg.~\ref{alg:SHC}},\, \bigcup_{i,j}\text{case }\mathrm{R}_{i,j})}{\Pr(\textsc{Reject})}\notag\\
        \leq& \sum_{j=1}^2\sum_{i=1}^{k+1}\Pr(\text{Accepted by Alg.~\ref{alg:SHC}}\,|\, \text{case }\mathrm{R}_{i,j})\Pr( \text{case }\mathrm{R}_{i,j}\,|\,\textsc{Reject})\notag\\
        \leq& \sum_{j=1}^2\sum_{i=1}^{k+1}\Pr(\text{Accepted by the }i\text{th check of the }j\text{th round}\,|\, \text{case } \mathrm{R}_{i,j}),
    \end{align}
    where the first inequality comes from the union bound, and the second inequality is because it must be accepted by all checks to get \textsc{Accept} from the algorithm.
    Based on this calculation, we need to analyze the false-acceptance rate of each individual check to construct the overall rate.

    For each check $ i $ in the first $ k$ checks of either round $j$, the check requires \textsf{BernoulliTest} to test the sampling from \textsf{HSS} with $t=\frac{1}{2m_0^{3/2}b}$ and $r=O(m_0^{3/4}Bb^{-1})$.
    These inputs, as analyzed previously, satisfy the requirement of \cref{prop:iter_v_stab}.
    Given that $\mathrm{R}_{i,j}$ is equal to the Large case stated in the proposition, we can lower bound the signal probability as
    \begin{gather}
        \Pr(Z=1\,|\,\text{case }\mathrm{R}_{i,j})=\Pr(Z=1\,|\,\text{Large case})\geq\frac{c_2}{m_0^3}.\notag
    \end{gather}
    According to Proposition~\ref{prop:Bernoulli-distribution-parameter-testing}, the \textsf{BernoulliTest} subroutine can output the correct answer with high probability, given the specified bounds $c_2/m_0^3>c_1/m_0^3$.
    The false-acceptance rate in this check is bounded as follows,
    \begin{align}
        &\Pr(\text{Accepted by the }i\text{th check of the }j\text{th round}\,|\, \text{case } \mathrm{R}_{i,j})\notag\\
        =&\Pr(\textsf{BernoulliTest}\text{ returns \textsc{Small} in the }i\text{th check of the }j\text{th round}\,|\, \text{Large case})\notag\\
        \leq&\, \delta'=\frac{\delta}{2(k+1)}.\notag
    \end{align}

    For the final check of each round, \emph{i.e.}, $i=k+1$ and $j\in[2]$, the sampler \textsf{HSS} has input $t=\frac{1}{4m_0^{3/2}\varepsilon}$ and $r=O(m_0^{3/4}B\varepsilon^{-1})$, which satisfy the requirement of \cref{prop:fina_v_stab}.
    Given that the corresponding case $\mathrm{R}_{i,j}$ is equal to the Large case as stated in \cref{prop:fina_v_stab}, we can bound the signal probability of this sampling as follows,
    \begin{gather}
        \Pr(Z=1\,|\,\text{case }\mathrm{R}_{i,j})=\Pr(Z=1\,|\,\text{Large case})\geq\frac{c_4}{m_0^3}.\notag
    \end{gather}
    According to Proposition~\ref{prop:Bernoulli-distribution-parameter-testing}, the \textsf{BernoulliTest} subroutine can output the correct answer with high probability, given the specified bounds $c_4/m_0^3>c_3/m_0^3$.
    The false-acceptance rate in this check is bounded as follows,
    \begin{align}
        &\Pr(\text{Accepted by the final check of the }j\text{th round}\,|\, \text{case } \mathrm{R}_{i,j})\notag\\
        =&\Pr(\textsf{BernoulliTest}\text{ returns \textsc{Small} in the final check of the }j\text{th round}\,|\, \text{Large case})\notag\\
        \leq&\delta'=\frac{\delta}{2(k+1)}.\notag
    \end{align}

    Therefore, the overall false-acceptance rate follows from Eq.~\eqref{eq:false-acceptance2} is
    \begin{align}
        \Pr(\text{Accepted by Alg.~\ref{alg:SHC}}\,|\, \textsc{Reject})\leq \sum_{j=1}^2\sum_{i=1}^{k+1}\frac{\delta}{2(k+1)}\leq\delta.\notag
    \end{align}
Therefore, we have bounded both types of the failure probabilities by $\delta$.

Regarding the complexities of running Algorithm~\ref{alg:SHC}, we can enumerate all $2(k+1)$ checks.
According to Proposition~\ref{prop:Bernoulli-distribution-parameter-testing}, for each check, the \textsf{BernoulliTest} will query \textsf{HSS} for a total number of$$O\left(m_0^{3}\log(\frac{2(k+1)}{\delta})\right)=O\left(m^{3}\left(\log\log(B\varepsilon^{-1})+\log(\delta^{-1})\right)\right).$$
    In each query, the algorithm conducts one syndrome measurement.
    Therefore, the measurement complexity is:
    \begin{align}
        &O\left(m^{3}\left(\log\log(B\varepsilon^{-1})+\log(\delta^{-1})\right)\right)\times \left(2k+2\right)\notag\\
        =&O\left(m^{3}\log(B\varepsilon^{-1})\left(\log\log(B\varepsilon^{-1})+\log(\delta^{-1})\right)\right).
    \end{align}

    The \textsf{BernoulliTest} subroutine implements the second-order Trotter formula, requiring the queries of $\mc O_H$ for $r$ times.
    Overall, we can sum up the query complexity as follows:
    \[
    \begin{aligned}
        &\sum_{j=1}^2\sum_{i=1}^{k}O\left(m^{3}\left(\log\log(B\varepsilon^{-1})+\log(\delta^{-1})\right)\right)\times O\left(m_0^{3/4}2^{i}\epsilon^{-1}\right)\notag\\
        &+ O\left(m^{3}\left(\log\log(B\varepsilon^{-1})+\log(\delta^{-1})\right)\right)\times O(m_0^{3/4}B\varepsilon^{-1})\notag\\
        &=O\left(m^{15/4}B\varepsilon^{-1}\left(\log\log(B\varepsilon^{-1})+\log(\delta^{-1})\right)\right).
    \end{aligned}
    \]
    Similarly, we can sum up the total queried evolution time:
    \begin{align}
        T&=\sum_{i=1}^{k}\frac{O\left(m^{3}(\log\log(B\varepsilon^{-1})+\log\delta^{-1})\right)}{2m_0^{3/2}\times2^{k-i+2}\varepsilon}+\frac{O\left(m^{3}(\log\log(B\varepsilon^{-1})+\log\delta^{-1})\right)}{4m_0^{3/2}\varepsilon}\notag\\
        &=O\left(m^{3/2}\varepsilon^{-1}\left(\log\log(B\varepsilon^{-1})+\log\delta^{-1}\right)\right).\notag
    \end{align}
    In the above derivation of complexities, we have used facts $k=\lceil\log(B/\epsilon)\rceil$ and that $m_0=2m$.
\end{proof}

\section{Lower bounds of Hamiltonian certification}
\label{sec:lower_bound}

In this section, we provide foundations to prove the lower bound of queried evolution time among all Hamiltonian certification methods.
Typically, we first formalize the circuit models to cover all possible methods.
Based on the models, we introduce the total variance and link it to the failure probability of the certification.

\subsection{Models of quantum experiments}

We consider a query access to the $k$-qubit controlled time evolution of an unknown $n$-qubit Hamiltonian $H$, with $k\in\mathbb{N}$ and time $t\in\mathbb{R}$.
Based on this access, we define the experiment of the corresponding certification circuit as follows.
This experiment allows for controlled evolution:
\begin{definition}[A single round experiment, adapted from~{\cite[Definition 24]{huang2023learning}}]\label{def:Single_experiment}
Suppose we have the query access $\mc O_H$ to the  controlled time evolution $\mathsf{ctrl}(e^{-\ii H t})$ of an unknown $n$-qubit Hamiltonian $H$, with $t\in\mathbb{R}$.
A single experiment $E$ of using $\mc O_H$ can be specified by:
\begin{enumerate}
    \item an arbitrary $n'$-qubit initial state $\ket{\psi_0} \in \mc H(2^{n'})$ with an integer $n' \geq n+1$,
    \item an arbitrary POVM $\mathcal{F} = \{ M_i \}_i$ on $n'$-qubit system,
    \item an $n'$-qubit unitary of the following form,
    \begin{equation}
    U_{D+1} (\mc O_H(t_D) \otimes I) U_D \ldots U_3 (\mc O_H(t_2) \otimes I) U_2 (\mc O_H(t_1) \otimes I) U_1,\notag
    \end{equation}
    for some arbitrary integer $D$, arbitrary evolution times $t_1, \ldots, t_D \in \mathbb{R}$, and arbitrary $n'$-qubit unitaries $U_1, \ldots, U_D, U_{D+1}$.
    Here $I$ is the identity unitary on $n'-n$ qubits.
\end{enumerate}
A single execution of $E$ returns an outcome from performing the POVM $\mathcal{F}$ on the state
\begin{equation}
    U_{D+1} (\mc O_H(t_D) \otimes I) U_D \ldots U_3 (\mc O_H(t_2) \otimes I) U_2 (\mc O_H(t_1) \otimes I) U_1 \ket{\psi_0}.\notag
\end{equation}
The queried evolution time of the experiment is defined as $t(E)\coloneqq \sum_{i=1}^D |t_i|$.
\end{definition}

Besides the single experiment, the certification methods can adaptively choose multiple experiments based on earlier measurement outcomes.
Given this concern, we also introduce the general formalism for the adaptive case using the \emph{tree representation} from~\cite{chen2022exponential}.

\begin{definition}[Tree representation of multiple round adaptive experiments]\label{def:Adaptive_experiments}
Suppose we have the query access $\mc O_H$ to the controlled time evolution $\mathsf{ctrl}(e^{-\ii H t})$ of an unknown $n$-qubit Hamiltonian $H$, with $t\in\mathbb{R}$.
An arbitrary algorithm using $\mc O_H$ with $T$ measurements can be expressed as a rooted tree $\mc T$ of depth $T+1$.
The tree satisﬁes the following properties.
     \begin{itemize}
         \item Each non-leaf node $u$ encodes an experiment $E^{(u)}$ as in Definition~\ref{def:Single_experiment}.
         \item Each leaf node $l$ encodes the path from the root.
         \item Each node is associated with a probability $p(u)$.
         \item The probability of the root $r$ is $p(r)=1$.
         \item At each non-leaf node $u$, we measure a POVM $\{ M^u_s\}_s$ to obtain a classical outcome $s$. Each child node $v_s$ of the node $u$ is connected through the edge $e_{u,s}$.
         \item If $v_s$ is the child node of $u$ connected through the edge $e_{u,s}$, then
         \begin{gather}
             p(v_s)=p(u)\cdot \Pr(s\,|\,E^{(u)}).\notag
         \end{gather}
     \end{itemize}
     The queried evolution time of this algorithm is defined as
     \begin{gather}
         t(\mc T)\coloneqq\max_{P:\text{Path from the root to a leaf in }\mc T}\sum_{u\in P}t(E^{(u)}).\notag
     \end{gather}
\end{definition}

\subsection{Total variation distance and failure probabilities}
\begin{definition}
    Given two probability distributions $p$ and $q$ over a domain $\mc D$, the \emph{total variation distance} $\mathrm{TV}(p,q)$ is defined as:
    \begin{gather}
        \mathrm{TV}(p,q)\coloneqq\frac{1}{2}\sum_{x\in\mc D}\abs{p(x)-q(x)}.\notag
    \end{gather}
\end{definition}
\begin{lemma}\label{lm:TV_bound}
Given distribution $p,q$ over a domain $\mc D$, if $\mathrm{TV}(p,q)<\frac{1}{3}$, there is no algorithm $\mc A$ that distinguishes distributions $p$ versus $q$ using one sample with probability larger or equal to $\frac{2}{3}$.
\end{lemma}
\begin{proof}
    Let $\mc D'\subseteq \mc D$ be the subset where $\mc A$ will output $p$.
    We have
    \begin{align}
        \Pr_{x\sim p}(\mc A(x)=q)+\Pr_{x\sim 1}(\mc A(x)=p)=&1-p(\mc D')+q(\mc D')\geq1-\sum_{x\in\mc D'}\abs{p(x)-q(x)}\notag\\
        \geq&1-\mathrm{TV}(p,q)\geq\frac{2}{3}.\notag
    \end{align}
    Therefore, at least one of the preceding false rates is no smaller than $\frac{1}{3}$.
\end{proof}

\begin{lemma}[Adapted from {\cite[Lemma 4]{haah2021quantum}}]\label{lm:Matrix_expo}
    Given two Hermitian matrices $A,B\in\mathbb{C}^{N\times N}$ with size $N$, the operator-norm distance between matrix exponentials can be bounded by
    \begin{gather*}
        \|e^{\ii A}-e^{\ii B}\|\leq\|A-B\|.
    \end{gather*}
\end{lemma}

\begin{lemma}[Total variation of a single experiment, adapted from~{\cite[Lemma 28]{huang2023learning}}]\label{lm:TV_single}
Suppose we have the query access $\mc O_H$ to the controlled time evolution $\mathsf{ctrl}(e^{-\ii H t})$ of an unknown $n$-qubit Hamiltonian $H$, with $t\in\mathbb{R}$.
Given that the Hamiltonian $H$ is either $H_1$ or $H_2$, let $p_1$ and $p_2$ be corresponding distributions of measurement outcomes from an arbitrary single round experiment $E$ using $\mc O_H$ with queried evolution time $t(E)$.
The total variation distance is bounded by
    \begin{gather}
        \mathrm{TV}(p_1,p_2)\leq\min(2\|H_1-H_2\|\cdot t(E),1).\notag
    \end{gather}
\end{lemma}
\begin{proof}
    Suppose we have two query accesses of controlled time evolution of distinct Hamiltonians $H_1$ and $H_2$.
    The diamond-norm distance between these two queried unitaries can be bounded as
    \begin{align}
        \|\mc O_{H_1}(t)-\mc O_{H_2}(t)\|_\diamond\leq&2\|\mathsf{ctrl}(e^{-\ii H_1 t})-\mathsf{ctrl}(e^{-\ii H_2 t})\|\notag\\
        =&2\left\|\ket{1}\bra{1}\otimes(e^{-\ii H_1 t}-e^{-\ii H_2 t})\right\|\notag\\
        \leq&2\left\|e^{-\ii H_1 t}-e^{-\ii H_2 t}\right\|\leq2\|H_1- H_2\|\abs{t},\notag
    \end{align}
    where we have used Lemma~\ref{prop:diamond} and \cref{lm:Matrix_expo} in the derivation.

    W.l.o.g., we can assume the experiment $E$ implement a circuit:
    \begin{equation}
    \ket{\psi_j}=U_{D+1} (\mc O_{H_j}(t_D) \otimes I) U_D \ldots U_3 (\mc O_{H_j}(t_2) \otimes I) U_2 (\mc O_{H_j}(t_1) \otimes I) U_1\ket{\psi_0},\notag
    \end{equation}
    for some depth $D$ and $j\in[2]$.
    Therefore, the probabilities of all measurement outcomes can be determined by the POVM $\{M_i\}$ as
    \begin{gather}
        p_j(i)=\bra{\psi_j}M_i\ket{\psi_j}.\notag
    \end{gather}
    The corresponding total variation is:
    \begin{align}
        \mathrm{TV}(p_1,p_2)=&\frac{1}{2}\sum_i\abs{\bra{\psi_1}M_i\ket{\psi_1}-\bra{\psi_2}M_i\ket{\psi_2}}=\sum_{i\in\mc S}\bra{\psi_1}M_i\ket{\psi_1}-\bra{\psi_2}M_i\ket{\psi_2}\notag\\
        \leq&\Abs{\sum_{i\in\mc S} M_i}\cdot2^{n'}\|\ket{\psi_1}\bra{\psi_1}-\ket{\psi_2}\bra{\psi_2}\|_{\mathrm{Schatten},1}\notag\\
        \leq&2^{n'}\|\ket{\psi_1}\bra{\psi_1}-\ket{\psi_2}\bra{\psi_2}\|_{\mathrm{Schatten},1}\notag\\
        \leq&\sum_{d=1}^D\|\mc O_{H_1}(t_d)-\mc O_{H_2}(t_d)\|_\diamond\leq2\|H_1-H_2\|\cdot t(E),\notag
    \end{align}
    where we have denoted $\mc S$ the set of $i$ where $\bra{\psi_1}M_i\ket{\psi_1}\geq\bra{\psi_2}M_i\ket{\psi_2}$.
    To get the second line, we have used H\"{o}lder inequality for normalized Schatten norms.
    The last line comes from the triangle inequality and the definition of the diamond norm given that $2^{n'}\|\ket{\psi}\bra{\psi}\|_{\mathrm{Schatten},1}=1$ holds for all $n'$-qubit state $\ket{\psi}$.
    Note that the total variation is by definition at most $1$, which completes the proof.
\end{proof}

\begin{lemma}[Total variation of multiple experiments, adapted from~{\cite[Lemma 29]{huang2023learning}}]\label{lm:TV_multiple}
  Suppose we have the query access $\mc O_H$ to controlled time evolution $\mathsf{ctrl}(e^{-\ii H t})$ of an unknown $n$-qubit Hamiltonian $H$, with $t\in\mathbb{R}$.
Given that the Hamiltonian $H$ is either $H_1$ or $H_2$, let $p_1$ and $p_2$ be corresponding distributions of all possible outcomes of an arbitrary algorithm $\mc A$ with the overall queried evolution time $T$.
    The total variation distance is bounded by
    \begin{gather}
        \mathrm{TV}(p_1,p_2)\leq\min(2\|H_1-H_2\|\cdot T,1).\notag
    \end{gather}
\end{lemma}
\begin{proof}
    Given the tree representation $\mc T$ in Drefinition~\ref{def:Adaptive_experiments} of a set of multiple adaptive measurements used in the algorithm $\mc A$, all the results from the measurements during the algorithm are recorded in each leaf node.
    Moreover, the probability corresponding to a certain series of outcomes is equal to the probability associated with the leaf node.
    To prove the lemma, we use induction over all tree depths.

    For a tree with depth 2, it is essentially a single experiment.
    According to Lemma~\ref{lm:TV_single}, we have
    \begin{gather}
        \mathrm{TV}(p_1,p_2)\leq \min(2\|H_1-H_2\|T,1).\notag
    \end{gather}
    Hence, the induction hypothesis holds.

    For a tree $\mc T$ with depth $D\geq3$, we denote its root node by $r$.
    \begin{align}
    &1 - \mathrm{TV}(p_1^{(\mathcal{T})}, p_2^{(\mathcal{T})}) \notag\\
    &= \sum_{\substack{\ell \in \mathrm{leaf}( \mathcal{T}) }} \min\left( p_1^{(\mathcal{T})}(\ell), p_2^{(\mathcal{T})}(\ell) \right) \notag\\
    &= \sum_{u \in \mathrm{child}(r)} \sum_{\substack{\ell \in \mathrm{leaf}( \mathcal{T}_u ) }} \min\left( p_1(u)p_1^{(\mc T_u)}(\ell), p_2(u)p_2^{(\mc T_u)}(\ell) \right) \notag\\
    &\geq \sum_{u \in \mathrm{child}(r)} \min(p_1(u),p_2(u))\sum_{\substack{\ell \in \mathrm{leaf}( \mathcal{T}_u ) }} \min\left( p_1^{(\mc T_u)}(\ell), p_2^{(\mc T_u)}(\ell) \right) \notag\\
    &\geq \left( 1 - \mathrm{TV}(p_1^{(r)}, p_2^{(r)}) \right) \min_{u \in \mathrm{child}(r)} \left(1 - \mathrm{TV}\left(p_1^{(\mathcal{T}_u}, p_2^{(\mathcal{T}_u)}\right)\right),\notag
\end{align}
where we have used $\mc  T_u$ to denote the subtree rooted from node $u$ and $p^{(r)}$ to denote the distribution of the first experiment from the root $r$.
Note that the total variation is at most 1.
Using the definition the time $T$ in Defintion~\ref{def:Adaptive_experiments}, we have:
\[
    \mathrm{TV}(p_1^{(\mathcal{T})}, p_2^{(\mathcal{T})})\leq \min(2\|H_1-H_2\|(T_1+\max_{u \in \mathrm{child}(r)}T(\mc T_u)),1)=\min(2\|H_1-H_2\|T,1).\qedhere
\]
\end{proof}

\subsection{Proof of lower bounds}\label{sec:append_proof_lower_bounds}

\begin{proof}[Proof of Theorem~\ref{thm:lower_Pauli}]
    Suppose we have an algorithm $\mc A$ that solves the stated problem for an arbitrary pair of $0\leq\varepsilon_1<\varepsilon_2\leq1$.
    Consider we are given the query access $\mc O_H$ for the (controlled) time evolution and its inverse with $H$ being either of the following two cases:
    \begin{itemize}
        \item $H_1=\varepsilon_1 Z_1$,
        \item $H_2=\varepsilon_2 Z_1$.
    \end{itemize}
    By setting $H_0=0$, it is evident that the first and second Hamiltonians fall into the \textsc{Accept} and \textsc{Reject} cases, respectively.
    Consequently, $\mc A$ can be employed for one time to distinguish these two cases with success probability at least $2/3$.
    Suppose the algorithm using queries of different Hamiltonian evolutions with a total time $T$, and we denote the distinct distributions of the algorithm's measurement outcomes by $p_1$ and $p_2$, respectively.
    According to Lemma~\ref{lm:TV_multiple}, the total variation distance between $p_1$ and $p_2$ is bounded by
    \begin{gather}
        \mathrm{TV}(p_1,p_2)\leq\min(2\|H_1-H_2\|T,1)=\min(2(\varepsilon_2-\varepsilon_1) T,1),\notag
    \end{gather}
    From Lemma~\ref{lm:TV_bound}, to guarantee a success probability of $2/3$, we need the total variation to be at least $1/3$.
    This concludes that $\mc A$ must query $\mc O_H$ with a total evolution time
    \begin{gather}
        T\geq\Omega\left(\frac{1}{\varepsilon_2-\varepsilon_1}\right).\notag
    \end{gather}

    As for the case where $p\in[1,2)$, we similarly consider an algorithm $\mc A$ that solves the stated problem for an arbitrary pair of $0\leq\varepsilon_1<\varepsilon_2\leq1$.
    According to~\cite{sarkar2019sets}, given that $m\leq2n+1$, we can choose an anti-commutating set $\mc S\subset{\sf P}^n$ with cardinality $m$ such that every Pauli operator anti-commutes with all other Pauli therein.
    Therefore, we can consider the following scenario where we are given the query access $\mc O_H$ for the (controlled) time evolution and its inverse with $H$ being either of the following two cases:
    \begin{itemize}
        \item $H_1=\sum_{\alpha\in\mc S}\frac{\varepsilon_1}{m^{1/p}}P_\alpha$,
        \item $H_2=\sum_{\alpha\in\mc S}\frac{\varepsilon_2}{m^{1/p}}P_\alpha$.
    \end{itemize}
    By setting $H_0=0$, it is evident that the first and second Hamiltonians fall into the \textsc{Accept} and \textsc{Reject} cases of the certification task, respectively.
    Consequently, running $\mc A$ for one time is able to distinguish these two cases with a success probability at least $2/3$.
    Suppose the algorithm using queries of different Hamiltonian evolutions with a total time $T$, and we denote the distinct distributions of the algorithm's measurement outcomes by $p_1$ and $p_2$, respectively.
    According to Lemma~\ref{lm:TV_multiple}, the total variation distance between $p_1$ and $p_2$ is bounded by
    \begin{gather}
        \mathrm{TV}(p_1,p_2)\leq\min(2\|H_1-H_2\|T,1)\leq\min\left(\frac{2(\varepsilon_2-\varepsilon_1)T}{m^{1/p-1/2}},1\right),\notag
    \end{gather}
    where we use that
    \begin{align}
        \|H_1-H_2\|=&\max_{\ket{\psi}}\sqrt{\bra{\psi}(H_2-H_1)^\dag(H_2-H_1)\ket{\psi}}=\max_{\ket{\psi}}\sqrt{\bra{\psi}(H_2-H_1)(H_2-H_1)\ket{\psi}}\notag\\
        =&\max_{\ket{\psi}}\sqrt{\bra{\psi}\sum_{\alpha\in\mc S}\frac{(\varepsilon_2-\varepsilon_1)^2}{m^{2/p}}I\ket{\psi}}=\frac{\varepsilon_2-\varepsilon_1}{m^{1/p-1/2}}.\notag
    \end{align}
    From Lemma~\ref{lm:TV_bound}, to guarantee a success probability of $2/3$, we need the total variation to be at least $1/3$.
    This concludes that $\mc A$ must query $\mc O_H$ with a total evolution time
    \begin{equation*}
        T\geq\Omega\left(\frac{m^{1/p-1/2}}{\varepsilon_2-\varepsilon_1}\right). \qedhere
    \end{equation*}
\end{proof}

\begin{proof}[Proof of Theorem~\ref{thm:lower_Schatten}]
    Suppose we have an algorithm $\mc A$ that solves the stated problem for an arbitrary pair of $0\leq\varepsilon_1<\varepsilon_2\leq1$.
    Consider we are given the query access $\mc O_H$ of (controlled) time evolution and its inverse with $H$ being either of the following two cases:
    \begin{itemize}
        \item $H_1=\varepsilon_1 Z_1$,
        \item $H_2=\varepsilon_2 Z_1$.
    \end{itemize}
    By setting $H_0=0$, it is evident that the first and second Hamiltonians fall into the \textsc{Accept} and \textsc{Reject} cases, respectively.
    Consequently, $\mc A$ can be employed for one time to distinguish these two cases with success probability at least $2/3$.
    Suppose the algorithm using queries of different Hamiltonian evolutions with a total time $T$, and we denote the distinct distributions of the algorithm's measurement outcomes by $p_1$ and $p_2$, respectively.
    According to Lemma~\ref{lm:TV_multiple}, the total variation distance between $p_1$ and $p_2$ is bounded by
    \begin{gather}
        \mathrm{TV}(p_1,p_2)\leq\min(2\|H_1-H_2\|T,1)=\min(2(\varepsilon_2-\varepsilon_1) T,1).\notag
    \end{gather}
    From Lemma~\ref{lm:TV_bound}, to guarantee a success probability of $2/3$, we need the total variation to be at least $1/3$.
    This concludes that $\mc A$ must query $\mc O_H$ with a total evolution time
    \begin{equation*}
        T\geq\Omega\left(\frac{1}{\varepsilon_2-\varepsilon_1}\right).\qedhere
    \end{equation*}
\end{proof}

\subsection{Proof of \texorpdfstring{$\mathsf{coQMA}$}{coQMA}-hardness}\label{sec:appendix-hardness}
To establish the $\mathsf{coQMA}$-hardness
of the $k$-local robust Hamiltonian certification
with respect to the operator norm,
we first recall a $\QMA$-complete problem---the $k$-local Hamiltonian problem
introduced in~\cite{KSV02}.

\begin{problem}
    Let $H = \sum_{i=1}^m H^{(i)}$
    be an $n$-qubit Hamiltonian
    where $m = \textup{poly}(n)$,
    with $H_i$ acts non-trivially
    on $k$ qubits 
    and $\norm{H^{(i)}}\le \textup{poly}(n)$ for all $i\in [m]$.
    For $a < b$ with $b-a\ge 1/\textup{poly}(n)$,
    the $k$-local Hamiltonian
    problem $(n,a,b)$
    is to decide
    whether
    \begin{itemize}
        \item $H$ has an eigenvalue less than $a$,
        \item all eigenvalues
        of $H$ are greater than $b$,
    \end{itemize}
    promised one of these to be the case.
    The problem requires outputting \enquote{Yes} in the first case and \enquote{No} otherwise.
\end{problem}
    
The $k$-local Hamiltonian was proven
to be $\QMA$-complete for $k = O(\log n)$ in~\cite{KSV02}, and this result was subsequently improved to $k=3$
in~\cite{NM07} and $k=2$
in~\cite{KKR06}.

\begin{theorem}[Theorem 1.1 in~\cite{KSV02}]
    The $2$-local Hamiltonian problem
    is $\QMA$-complete.
\end{theorem}

We now reduce the above problem
to the robust $k$-local
Hamiltonian certification problem
with respect to the operator norm.

\begin{proof}[Proof of \cref{thm:QMA-hard-HC-operator}]
    Given any instance $H = \sum_iH^{(i)}$ of the $2$-local
    Hamiltonian problem $(n,a,b)$, 
    our reduction works as follows.
    First,
    note that $\norm{H}\le \sum_i\norm{H^{(i)}} \le \textup{poly}(n)$,
    we can assume that $\abs{a}\le \textup{poly}(n) $
    and $\abs{b}\le \textup{poly}(n) $.
    By scaling and adding $cI$ to $H$,
    we can further require that 
    $H$ is positive semidefinite
    with $\norm{H}\le 1$
    and
    $0\le  a < b \le 1$,
    where we still have $b-a \ge 1/\textup{poly}(n)$.
    Then,
    the problem
    is equivalent to
    decide the largest eigenvalue 
    of $I-H$
    is larger than $1-a$
    or smaller than $1-b$.

    Now, consider the following
    $3$-local traceless Hamiltonian
    \[
    H':= \ket{0}\bra{0}\otimes (I-H)
    + \ket{1}\bra{1}\otimes cI,
    \]
    where 
    $c = -\tr(I-H)/2^n = -1+\tr(H)/2^n$
    can be computed efficiently
    since $\tr(H) = \sum_i\tr(H_i)$
    and $\abs{c}\le 1 + \sum_i \abs{\tr(H_i)}/2^n \le \textup{poly}(n)$.
    By our choice of $c$ and $H'$,
    we know $\abs{c}=1-\tr(H)/2^n \le 1-\lambda_{\min}(H)$,
    and $\norm{H'} = 
    1 - \lambda_{\min} (H)$.
    Therefore, 
    for the robust $k$-local 
    Hamiltonian certification problem
    $(n, 1-b, 1-a)$
    with respect to the operator norm,
    consider certifying the instance 
    $H'$ and $H_0 = 0$.
    If the algorithm returns \enquote{Yes},
    then we know 
    $\norm{H' - H_0}= \norm{H'}\le 1-b $,
    meaning that $\lambda_{\min}(H)\ge b$, and $H$ is a \enquote{no} case
    for the $2$-local Hamiltonian problem.
    Otherwise we know
    $\norm{H' - H_0} = \norm{H'}\ge 1-a$,
    meaning that $\lambda_{\min}(H)\le a$, and $H$ is a \enquote{yes} case
    for the $2$-local Hamiltonian problem.
    This gives the desired $\mathsf{co}\QMA$-completeness.
\end{proof}

\end{document}